\documentclass[a4paper]{article}

\usepackage{graphicx}
\usepackage{amsmath}
\usepackage{amssymb}

\newcommand{\nat}{I\!\!N}
\newcommand{\real}{I\!\!R}
\newcommand{\eop}{\sqcap\!\!\!\!\sqcup}

\newcommand{\cho}{[]}
\newcommand{\rs}{\ {\sf rs}\ }
\newcommand{\sy}{\ {\sf sy}\ }

\newtheorem{definition}{Definition}[section]
\newtheorem{example}{Example}[section]
\newtheorem{proposition}{Proposition}[section]

\date{}

\title{Discrete time phased Petri box calculus dtphPBC}

\author{{\sc Igor V. Tarasyuk}\\
A.P. Ershov Institute of Informatics Systems,\\
Siberian Branch of the Russian Academy of Sciences,\\
Acad. Lavrentiev pr. 6, 630090 Novosibirsk, Russian Federation\\
{\tt itar@iis.nsk.su}}

\begin{document}

\maketitle

\begin{abstract}
We propose discrete time phased Petri box calculus (dtphPBC), an extension with phase type distributed multiaction
delays of discrete time stochastic and deterministic Petri box calculus (dtsdPBC), previously presented by I.V.
Tarasyuk. In dtphPBC, transition probability matrices (TPMs) of finite absorbing discrete time Markov chains (DTMCs)
with a single absorbing state specify discrete phase type (DPH) distributed delays (including zero delay) of the phased
multiactions that generalize stochastic and deterministic multiactions from dtsdPBC. The positively phased (timed)
multiactions have positive DPH delays represented by the non-empty TPM matrices over transient states (transient TPMs).
The zero phased (immediate) multiactions have zero DPH delay represented by the empty transient TPM.

The step operational semantics of dtphPBC is constructed via labeled probabilistic transition systems. The transition
systems incorporate the absorbing DTMCs of the DPH delays of the executed phased multiactions via the structural
operational semantics (SOS) rules. The SOS rules define a labeling with the empty set on the transitions among
transient states of the absorbing DTMC and on the self-loop in the absorbing state of it. The transitions going from
the transient states (positive phases) to the absorbing state (zero phase) are labeled with the executions, being the
positive phases-superscribed timed multiactions whose (positive) delays are defined by the absorbing DTMC. A series of
examples demonstrates how to construct the transition systems of the dynamic expressions, combined from timed and
immediate multiactions with different operations of the calculus.
\bigskip\\
{\bf Keywords:} stochastic process algebra, Petri box calculus, discrete time, phase type distribution, absorbing
discrete time Markov chain, phased multiaction, timed multiaction, immediate multiaction, execution, transition system,
operational semantics.
\end{abstract}

\tableofcontents

\section{Introduction}
\label{introduction.sec}

Algebraic process calculi, like Communicating Sequential Processes (CSP) \cite{Hoa85}, Algebra of Communicating
Processes (ACP) \cite{BK85} and Calculus of Communicating Systems (CCS) \cite{Mil89} are well-known formal models for
specification of computing systems and analysis of their behaviour. In such process algebras (PAs), systems and
processes are specified by formulas, and verification of their properties is accomplished at a syntactic level via
equivalences, axioms and inference rules. In recent decades, stochastic extensions of PAs were proposed, such as
Markovian TImed Processes and Performability (Performance and dependability) evaluation (MTIPP) \cite{HR94},
Performance Evaluation Process Algebra (PEPA) \cite{Hil94a,Hil96} and Extended Markovian Process Algebra (EMPA)
\cite{BGo98}. Unlike standard PAs, stochastic process algebras (SPAs) do not just specify actions which can occur
(qualitative features), but they associate with the actions the distribution parameters of their random time delays
(quantitative characteristics).

\subsection{Petri box calculus}

PAs specify concurrent systems in a compositional way via an expressive formal syntax. On the other hand, Petri nets
(PNs) provide a graphical representation of such systems and capture explicit asynchrony in their behaviour. To combine
the advantages of both models, a semantics of algebraic formulas via PNs was defined.

Petri box calculus (PBC) \cite{BDH92,BKo95,BDK01,BD24} is a flexible and expressive process algebra developed as a tool
for specification of the PNs structure and their interrelations. Its goal was also to propose a compositional semantics
for high level constructs of concurrent programming languages in terms of elementary PNs. Formulas of PBC are combined
not from single (visible or invisible) actions and variables, like in CCS, but from multisets of elementary actions and
their conjugates, called multiactions ({\em basic formulas}). The empty multiset of actions is interpreted as the
silent multiaction specifying some invisible activity. In contrast to CCS, synchronization is separated from
parallelism ({\em concurrent constructs}). Synchronization is a unary multi-way stepwise operation, based on
communication of actions and their conjugates. This extends the CCS approach with conjugate matching labels.
Synchronization in PBC is asynchronous, unlike that in Synchronous CCS (SCCS) \cite{Mil89}. Other operations are
sequence and choice ({\em sequential constructs}). The calculus includes also restriction and relabeling ({\em
abstraction constructs}). To specify infinite processes, refinement, recursion and iteration operations were added
({\em hierarchical constructs}). Thus, unlike CCS, PBC has an additional iteration operation to specify infinite
behaviour when the semantic interpretation in finite PNs is possible. PBC has a step operational semantics in terms of
labeled transition systems, based on the rules of structural operational semantics (SOS). The operational semantics of
PBC is of step type, since its SOS rules have transitions with (multi)sets of activities, corresponding to simultaneous
executions of activities (steps). A denotational semantics of PBC was proposed via a subclass of PNs equipped with an
interface and considered up to isomorphism, called Petri boxes. For more detailed comparison of PBC with other process
algebras and the reasoning about importance of non-interleaving semantics see \cite{BDH92,BDK01}.

The extensions of PBC with a deterministic, a nondeterministic or a stochastic model of time were presented.

\subsection{Time extensions of Petri box calculus}

To specify systems with time constraints, deterministic (fixed) or nondeterministic (interval) delays are used.

A time extension of PBC with a nondeterministic time model, called time Petri box calculus (tPBC), was proposed in
\cite{Kou00}. In tPBC, timing information is added by associating time intervals (the earliest and the latest firing
time) with instantaneous {\em actions}. tPBC has a step time operational semantics in terms of labeled transition
systems. Its denotational semantics was defined in terms of a subclass of labeled time Petri nets (LtPNs), based on
tPNs \cite{MFa76} and called time Petri boxes (ct-boxes).

Another time enrichment of PBC, called Timed Petri box calculus (TPBC), was defined in \cite{MF00,MF01}, it
ac\-com\-mo\-da\-tes a deterministic model of time. In contrast to tPBC, multiactions of TPBC are not instantaneous,
but have time durations. Additionally, in TPBC there exist no ``illegal'' multiaction occurrences, unlike tPBC. The
com\-ple\-x\-ity of ``illegal'' occurrences mechanism was one of the main intentions to construct TPBC though this
calculus appeared to be more complicated than tPBC. TPBC has a step timed operational semantics in terms of labeled
transition systems. The denotational semantics of TPBC was defined in terms of a subclass of la\-be\-led Timed Petri
nets (LTPNs), based on TPNs \cite{Ram73} and called Timed Petri boxes (T-boxes). tPBC and TPBC differ in ways they
capture time information, and they are not in competition but complement~each~other.

The third time extension of PBC, called arc time Petri box calculus (atPBC), was constructed in \cite{Nia05,NK05}, and
it implements a nondeterministic time. In atPBC, multiactions are associated with time delay intervals. atPBC possesses
a step time operational semantics in terms of labeled transition systems. Its denotational semantics was defined on a
subclass of labeled arc time Petri nets (atPNs), based of those from \cite{BLT90,Hani93}, where time restrictions are
associated with the arcs, called arc time Petri boxes (at-boxes).

tPBC, TPBC and atPBC, all adapt the discrete time approach, but TPBC has no immediate (multi)actions  (those with zero
delays).

\subsection{Stochastic extensions of Petri box calculus}

The set of states for the systems with deterministic or nondeterministic delays often differs drastically from that for
the timeless systems, hence, the analysis results for untimed systems may be not valid for the time ones. To solve this
problem, stochastic delays are considered, which are the random variables with a (discrete or continuous) probability
distribution. If the random variables governing delays have an infinite support then the corresponding SPA can exhibit
all the same behaviour as its underlying untimed PA.

A stochastic extension of PBC, called stochastic Petri box calculus (sPBC), was proposed in \cite{MVF01,MVCC04}. In
sPBC, multiactions have stochastic delays that follow (negative) exponential distribution. Each multiaction is equipped
with a rate that is a parameter of the corresponding exponential distribution. The instantaneous execution of a
stochastic multiaction is possible only after the corresponding stochastic time delay. The calculus has an interleaving
operational semantics defined via transition systems labeled with multiactions and their rates. Its denotational
semantics was defined in terms of a subclass of labeled continuous time stochastic PNs, based on CTSPNs
\cite{Mol82,Mar90,Bal01,Buc95,Bal07} and called stochastic Petri boxes (s-boxes). In sPBC, performance of the processes
is evaluated by analyzing their underlying continuous time Markov chains (CTMCs). In \cite{MVCF08}, a number of new
equivalence relations were proposed for regular terms of sPBC to choose later a suitable candidate for a congruence.

sPBC was enriched with immediate multiactions having zero delays in \cite{MVCR08,MVCRT16}. We call such an extension
generalized sPBC (gsPBC). An interleaving operational semantics of gsPBC was constructed via transition systems labeled
with stochastic or immediate multiactions together with their rates or probabilities. A denotational semantics of gsPBC
was defined via a subclass of labeled generalized stochastic PNs, based on GSPNs \cite{Mar90,MBCDF95,Buc98,Bal01,Bal07}
and called generalized stochastic Petri boxes (gs-boxes). The performance analysis in gsPBC is based on the semi-Markov
chains (SMCs).

PBC has a step operational semantics, whereas sPBC has an interleaving one. In step semantics, parallel executions of
activities (steps) are permitted while in interleaving semantics, we can execute only single activities. Hence, a
stochastic extension of PBC with a step semantics was needed to keep the concurrency degree of behavioural analysis at
the same level as in PBC. As mentioned in \cite{Mol81,Mol85}, in contrast to continuous time approach (used in sPBC),
discrete time approach allows for constructing models of common clock systems and clocked devices. In such models,
multiple transition firings (or executions of multiple activities) at time moments (ticks of the central clock) are
possible, resulting in a step semantics. Moreover, employment of discrete stochastic time fills the gap between the
models with deterministic (fixed) time delays and those with continuous stochastic time delays. As argued in
\cite{AHR00}, arbitrary delay distributions are much easier to handle in a discrete time domain. In
\cite{MVi08,MVi09,MABV12}, discrete stochastic time was preferred to enable simultaneous expiration of multiple delays.

In \cite{Tar05,Tar06,Tar07a,Tar14}, we presented an extension of the algebra PBC, called discrete time stochastic PBC
(dtsPBC). In dtsPBC, the residence time in the process states is geometrically distributed. A step operational
semantics of dtsPBC was constructed via labeled probabilistic transition systems. Its denotational semantics was
defined in terms of a subclass of labeled discrete time stochastic PNs (LDTSPNs), based on DTSPNs \cite{Mol81,Mol85}
and called discrete time stochastic Petri boxes (dts-boxes). The performance evaluation in dtsPBC is accomplished via
the underlying discrete time Markov chains (DTMCs) of the algebraic processes. A variety of stochastic equivalences
were proposed to identify stochastic processes with similar behaviour which are differentiated by the semantic
equivalence. The interrelations of all the introduced equivalences were studied. Since dtsPBC has a discrete time
semantics and geometrically distributed sojourn time in the process states, unlike sPBC with continuous time semantics
and exponentially distributed delays, the calculi apply two different approaches to the stochastic extension of PBC, in
spite of some similarity of their syntax and semantics inherited from PBC. The main advantage of dtsPBC is that
concurrency is treated like in PBC having step semantics, whereas in sPBC parallelism is simulated by interleaving,
obliging one to collect the information on causal independence of activities before constructing the semantics.

In \cite{TMV10,TMV13,TMV14,TMV15,TMV18}, an enhanced calculus discrete time stochastic and immediate PBC (dtsiPBC) was
proposed as an extension with immediate multiactions of dtsPBC. Immediate multiactions increase the specification
capability: they can model logical conditions, probabilistic branching, instantaneous probabilistic choices and
activities whose durations are negligible in comparison with those of others. They are also used to specify urgent
activities and the ones that are not relevant for performance evaluation. Thus, immediate multiactions can be
considered as a kind of instantaneous dynamic state adjustment and, in many cases, they result in a simpler and more
clear system representation. The step operational semantics of dtsiPBC was constructed with the use of labeled
probabilistic transition systems. Its denotational semantics was defined in terms of a subclass of labeled discrete
time stochastic and immediate PNs (LDTSIPNs), based on the extension of DTSPNs \cite{Mol81,Mol85} with transition
labeling and immediate transitions, called dtsi-boxes. The corresponding stochastic process, the underlying SMC, was
constructed and investigated, with the purpose of performance evaluation. In addition, the alternative solution methods
were develo\-ped, based on the underlying ordinary and reduced DTMCs. Step stochastic bisimulation equivalence of the
process expressions was defined to compare and reduce their transition systems and Markov chains, as well as to
identify the stationary behaviour.

In \cite{Tar19,Tar20a,Tar20b,Tar21,Tar23,Tar24,Tar25a,Tar25b}, we defined discrete time stochastic and deterministic
Petri box calculus (dtsdPBC) that enhances dtsiPBC with de\-ter\-mi\-nistic mul\-ti\-ac\-ti\-ons. In dtsdPBC, besides
the probabilities from the real-valued interval $(0;1)$, applied to calculate discrete time delays of sto\-chas\-tic
multiactions, also non-ne\-ga\-ti\-ve integers were used to specify fixed de\-lays of deterministic multiactions
(in\-clu\-ding zero delay, which is the case of im\-me\-di\-ate multiactions). To resolve conflicts among deterministic
multiactions, they were  equipped with positive real-valued weights. dtsdPBC has a step operational semantics, defined
via labeled proba\-bi\-lis\-tic transition systems. The de\-no\-ta\-ti\-o\-nal se\-man\-tics of dtsdPBC was defined in
terms of a subclass of labeled discrete time stochastic and deterministic Petri nets (LDTSDPNs), an enrichment of
DTSPNs \cite{Mol81,Mol85} with transition labeling and deterministic transitions, called dtsd-boxes. With the purpose
of performance evaluation in dtsdPBC, the underlying stochastic process of the process expressions was built and
analyzed, which is a semi-Markov chain (SMC). The alternative solution methods were developed, based on the underlying
discrete time Markov chain (DTMC) and its reduction (RDTMC) by eliminating vanishing states, i.e. those with zero
sojourn times. We proposed step stochastic bisimulation equivalence to identify the algebraic processes with similar
behaviour that are however differentiated by the semantics. We established consistency of the operational and
denotational semantics of dtsdPBC up to that equivalence, which was also used to reduce transition systems and SMCs,
DTMCs and RDTMCs of the process expressions while preserving the qualitative and quantitative characteristics.

\subsection{Our contributions}

In this paper, we extend dtsdPBC with discrete phase type distributed multiaction delays that are described by finite
absorbing DTMCs and include both geometric and non-Markovian (such as deterministic) delays as special cases. Discrete
phase type probability distributions can approximate with any required precision general discrete distributions over
the positive integers basis and are closed under minimum (alternative composition, conflict), maximum (parallel
composition, parallelism), finite convolution (sequential composition, precedence), finite weighted and infinite
geometric summations, as well as convex mixtures \cite{Neu81,BHST03,Ste09,HPV15,TBo17,LST19,HV23,HHPTV24,HRGBF24}. We
present an enrichment of dtsdPBC with phased multiactions, called {\em discrete time phased Petri box calculus}
(dtphPBC), which enhances the expressiveness of dtsdPBC and extends the application area of the associated
specification and analysis techniques. In dtphPBC, the transition probability matrices (TPMs) of finite absorbing DTMCs
(with all but one transient states and a single absorbing state) are used to specify phase type distributed delays of
phase multiactions (including the empty TPM specifying zero delay, which is the case of immediate multiactions). To
resolve conflicts among the executions with probability $1$ of the phase multiactions, they are additionally equipped
with positive real-valued weights. The step operational semantics of dtsdPBC is constructed with the use of labeled
probabilistic transition sys\-tems. With a number of examples we demonstrate how to construct the transition systems of
the dynamic expressions with different operations.

Our novel approach was inspired by some ideas on incorporating phase type delays into stochastic Petri nets and
stochastic process algebras. Some known SPNs with phase type transition de\-lays are: SPNs with phase-type distributed
transition times (PTDTT-SPNs) \cite{Cum85} and phased delay PNs (PDPNs) \cite{JC02} (the both SPN classes with discrete
and continuous time), as well as defective discrete phase SPNs (DDP-SPNs) \cite{Cia95}, discrete deterministic and
stochastic PNs (DDSPNs) \cite{ZC96,ZCH97} and non-Markovian SPNs (NMSPNs) \cite{HPST00} (the three SPN classes with
discrete time). Among all those SPN classes, only NMSPNs have a non-interleaving transition firing semantics being
technically involved though.

Some existing SPAs with phase type action delays are: a modification of $PA_{GS}$ \cite{Kat96}, $PEPA_{ph}^\infty$
\cite{ERKN99}, SB-LOTOS \cite{HK00}, PTP \cite{Wolf08}, CCC \cite{Pul09,PH15,BHKS23} and PHASE \cite{CR14,CR15,Cio23}
(the six SPAs with continuous time), as well as a variant of WSCCS \cite{Tof00} (with discrete time). All those SPAs
have only interleaving operational and no SPN-based denotational semantics. Those interleaving phase SPAs are either
theoretically-oriented and lacking practical case studies, or very specialized, or with restricted specification
capabilities. In detail, $PA_{GS}$ formalizes sophisticated compositions of stochastic bundle event structures,
$PEPA_{ph}^\infty$ describes certain classes of non-Mar\-ko\-vi\-an systems including phase-type delayed queueing
net\-works, SB-LOTOS separates actions and delays thus implementing orthogonal time, WSCCS has a technically complex
and non-sufficiently intuitive syn\-tax exploiting weights and priorities, PTP only specifies cooperating processes
that cannot be synchronized by shared activities, CCC does not have actions, synchronization and recursion, whereas
PHASE offers just a few operators and distinguishes between the action and delay transitions, aiming to easy
implementation. Unlike $PA_{GS}$, $PEPA_{ph}^\infty$, SB-LOTOS, PTP, CCC and PHASE with continuous time, WSCCS adapts a
discrete time model, but the semantics of WSCCS is still interleaving. Hence, it is actual to construct a discrete time
SPA with phase type delays and non-interleaving semantics: operational one (on the labeled transition systems with
parallel executions of ac\-ti\-vi\-ti\-es) and denotational one (on the SPNs with phase delays and parallel
firings~of~transitions).

Thus, the main contributions of the paper are the following.
\begin{itemize}

\item Syntax of new powerful and expressive discrete time SPA with phased activities, called dtphPBC.

\item Parallel step operational semantics of dtsdPBC in terms of labeled probabilistic transition systems.

\item Examples of the transition systems constructed for the process expressions with different operations.

\end{itemize}

\subsection{Structure of the paper}

The paper is organized as follows. In Section \ref{syntax.sec}, the syntax of the algebra dtphPBC is presented. In
Section \ref{opersem.sec}, we construct the step operational semantics of the calculus in terms of labeled
probabilistic transition systems. Finally, Section \ref{conclusion.sec} summarizes the results obtained and outlines
research perspectives in this area.

\section{Syntax}
\label{syntax.sec}

In this section, we propose the syntax of dtphPBC.

\subsection{Activities and operations}

We recall a definition of multiset that is an extension of the set notion by allowing several identical elements.

\begin{definition}
Let $X$ be a set. A finite {\em multiset (bag)} $M$ over $X$ is a mapping $M:X\rightarrow\nat$ such that $|\{x\in X\mid
M(x)>0\}|<\infty$, i.e. it can contain a finite number of elements only.
\end{definition}

We denote the {\em set of all finite multisets} over a set $X$ by $\nat_{fin}^X$. Let $M,M'\in\nat_{fin}^X$. The {\em
cardinality} of $M$ is defined as $|M|=\sum_{x\in X}M(x)$. We write $x\in M$ if $M(x)>0$ and $M\subseteq M'$ if
$\forall x\in X\ M(x)\leq M'(x)$. We define $(M+M')(x)=M(x)+M'(x)$ and $(M-M')(x)=\max\{0,M(x)-M'(x)\}$. When $\forall
x\in X,\ M(x)\leq 1,\ M$ can be interpreted as a proper set and denoted by $M\subseteq X$. The {\em set of all subsets
(powerset)} of $X$ is denoted by $2^X$.

We now define discrete phase type probability distributions.

\begin{definition}
A {\em (discrete) phase-type (DPH) distribution} is defined on $\real_{\geq 0}$ as the time to absorption in a finite
state discrete time Markov chain (DTMC) with a finite number of (or zero) transient states $n,\ldots ,1$ and a single
absorbing state $0$, called {\em phases} (often written in the decreasing order for convenience).

The {\em representation} of DPH distribution is a pair $(\varpi ,{\bf A})$, where the row vector $\varpi$ of $n$
elements $\varpi_i\ (1\leq i\leq n)$ denotes the {\em initial probability distribution over the transient states} and
the square matrix ${\bf A}$ of order $n=ord({\bf A})$ with the elements ${\cal A}_{ij},\ i,j\in\{1,\ldots ,n\}$, is the
sub-stochastic transient transition probability matrix (TPM) specifying the {\em transition probabilities among the
transient states}. The column vector ${\bf a}$ of $n$ elements ${\cal A}_i,\ i\in\{1,\ldots ,n\}$, denotes the {\em
transition probabilities towards absorption}, i.e. to the absorbing state $0$.

Thus, the {\em TPM of the absorbing DTMC} is
$\check{\bf A}=\left(\begin{array}{cc}
{\bf A} & {\bf a}\\
{\bf 0} & 1
\end{array}\right)$, where ${\bf 0}$ is a row vector of $n$ values $0$. The {\em initial probability distribution
vector of the absorbing DTMC} is $\pi =(\varpi ,\pi_0)$, where $\pi_0=1-\varpi{\bf 1}^T$ and ${\bf 1}$ is a row vector
of $n$ values $1$.

Let ${\bf I}$ is the identity matrix of order $n$. Each DPH distributed random variable $\xi$ has the following
characteristics:
\begin{itemize}

\item The {\em probability distribution function (PDF)} is $F_\xi (k)={\sf P}(\xi <k)=1-\varpi{\bf A}^k{\bf 1}^T\
(k\in\nat_{\geq 1})$.

\item The {\em probability mass function (PMF)} is $p_\xi (k)={\sf P}(\xi =k)=\varpi{\bf A}^{k-1}{\bf a}\
(k\in\nat_{\geq 1})$.

\item The {\em mean value (average, expectation)} is ${\sf M}(\xi )=\sum_{k=1}^\infty k p_\xi (k)=
\varpi({\bf I}-{\bf A})^{-1}{\bf 1}^T$.

\item The {\em dispersion (variance)} is ${\sf D}(\xi )=\sum_{k=1}^\infty (k-{\sf M}(\xi ))^2 p_\xi (k)=
2\varpi({\bf I}-{\bf A})^{-2}{\bf A}{\bf 1}^T$.

\end{itemize}
\end{definition}

In $\varpi$ and ${\bf A}$, the absorbing state $0$ is implicit, since its initial probability $\pi_0=1-\varpi{\bf 1}^T$
and the vector of transition probabilities towards it ${\bf a}=({\bf I}-{\bf A}){\bf 1}^T$ are fully determined by the
former ones. Let $\xi\sim DPH(\varpi ,{\bf A})$ denote that a random variable $\xi$ is DPH distributed with the
representation $(\varpi ,{\bf A})$.

For DPH distributions, we shall consider only {\em irreducible representations} $(\varpi ,{\bf A})$, i.e. those where
no transient state is visited with probability $0$ starting from the initial distribution $\varpi$. We are mainly
interested in {\em acyclic} DPH (ADPH) distributions \cite{BHST03}, whose transient TPM has the associated graph
without cycles. ADPH distribution has at least one representation with the transient TPM being (up to permutation of
states) in the upper triangular form, i.e. such that all its elements below the main diagonal are zeros.

The operations over DPH distributions are defined using the Kronecker (tensor) product of matrices.

\begin{definition}
Let ${\bf A}$ be a $n\times m$ matrix with the elements ${\cal A}_{ij}\ (1\leq i\leq n,\ 1\leq j\leq m)$ and ${\bf B}$
be a $k\times l$ matrix. The Kronecker (tensor) product of ${\bf A}$ and ${\bf B}$ is a $nk\times ml$ matrix, defined
as
$${\bf A}\otimes{\bf B}=\left(\begin{array}{ccc}
{\cal A}_{11}{\bf B} & \ldots & {\cal A}_{1m}{\bf B}\\
\vdots & \vdots & \vdots\\
{\cal A}_{n1}{\bf B} & \ldots & {\cal A}_{nm}{\bf B}
\end{array}\right).$$

Let ${\bf A}$ be a square matrix of order $n$ and ${\bf B}$ be a square matrix of order $m$. Let ${\bf I}_n$ is the
identity matrix of order $n$ and ${\bf I}_m$ is the identity matrix of order $m$. The Kronecker (tensor) sum of ${\bf
A}$ and ${\bf B}$ is a square matrix of order $nm$, defined as
$${\bf A}\oplus{\bf B}={\bf A}\otimes{\bf I}_m+{\bf I}_n\otimes{\bf B}.$$

\end{definition}

DPH distributions are closed under minimum and maximum (of the DPH distributed random variables). Let $\xi$ and $\zeta$
be two random variables with $\xi\sim DPH(\varpi_A,{\bf A})$ and $\zeta\sim DPH(\varpi_B,{\bf B})$.
\begin{itemize}

\item We have $\min (\xi ,\zeta )\sim DPH(\varpi_C,{\bf C})$, where
$${\bf C}={\bf A}\sqcap{\bf B}={\bf A}\otimes{\bf B}\mbox{ and }\varpi_C=\varpi_A\otimes\varpi_B.$$

Then the transition probabilities towards absorption vector and TPM of the absorbing DTMC are
$${\bf c}=({\bf I}-{\bf A}\otimes{\bf B}){\bf 1}^T\mbox{ and }
\check{\bf C}=\left(\begin{array}{cc}
{\bf A}\otimes{\bf B} & ({\bf I}-{\bf A}\otimes{\bf B}){\bf 1}^T\\
{\bf 0} & 1
\end{array}\right).$$

\item We have $\max (\xi ,\zeta )\sim DPH(\varpi_D,{\bf D})$, where
$${\bf D}={\bf A}\sqcup{\bf B}=\left(\begin{array}{ccc}
{\bf A}\otimes{\bf B} & {\bf A}\otimes{\bf b} & {\bf a}\otimes{\bf B}\\
0 & {\bf A} & 0\\
0 & 0 & {\bf B}
\end{array}\right)\mbox{ and }\varpi_D=(\varpi_A\otimes\varpi_B,(1-\varpi_B{\bf 1}^T)\varpi_A,
(1-\varpi_A{\bf 1}^T)\varpi_B).$$

Then the transition probabilities towards absorption vector and TPM of the absorbing DTMC are
$${\bf d}=\left(\begin{array}{c}
{\bf a}\otimes{\bf b}\\
{\bf a}\\
{\bf b}
\end{array}\right)\mbox{ and }
\check{\bf D}=\left(\begin{array}{cccc}
{\bf A}\otimes{\bf B} & {\bf A}\otimes{\bf b} & {\bf a}\otimes{\bf B} & {\bf a}\otimes{\bf b}\\
0 & {\bf A} & 0 & {\bf a}\\
0 & 0 & {\bf B} & {\bf b}\\
0 & 0 & 0 & 1
\end{array}\right).$$

\end{itemize}

For the non-empty transient TPMs, we shall only consider the initial distributions over transient states like $\varpi
=(1,0,\ldots ,0)={\bf e}_1$ (hence, $\pi_0=0$), i.e. we start from phase $n$ (the first phase in the decreasing order)
with probability $1$. In such a case, the transient TPM ${\bf A}\neq\emptyset$ will completely define the corresponding
DPH distribution ans we shall denote $DPH({\bf A})=DPH({\bf e}_1,{\bf A})$. This is not a severe restriction, since we
can always add to the absorbing DTMC a new starting state with the initial probability $1$ and outgoing transitions
towards the ``old'' initial states with the initial probabilities of those states (like the ``begin'' state from
\cite{Cia95,ZC96}). Then the empty transient TPM $\emptyset$ will define the Dirac distribution with the value
parameter $0$ (fixed, i.e. deterministic, zero delay) and support $\{0\}$, i.e. $DPH(\emptyset )=DPH(\varepsilon
,\emptyset )=Dirac(0)$, where $\varepsilon$ is the zero length vector.

\begin{example}
Let ${\bf A}$ be the transient TPM of the (discrete) uniform distribution $Unif(1,3)$ with the value border parameters
$1,3$ and support $\{1,2,3\}$ ($DPH({\bf A})=Unif(1,3)$). Let ${\bf B}$ be the transient TPM of the geo\-met\-ric
distribution $Geom(\frac{1}{3})$ with the probability parameter $\frac{1}{3}$ and support $\nat_{\geq 1}$ ($DPH({\bf
B})=Geom(\frac{1}{3})$). The transient TPMs ${\bf A},\ {\bf B}$, and the corresponding vectors of the transition
pro\-ba\-bi\-li\-ti\-es towards absorption ${\bf a},\ {\bf b}$ are

$${\bf A}=\left(\begin{array}{ccc}
0 & \frac{2}{3} & 0\\
0 & 0 & \frac{1}{2}\\
0 & 0 & 0
\end{array}\right),\
{\bf a}=\left(\begin{array}{c}
\frac{1}{3}\\
\frac{1}{2}\\
1
\end{array}\right),\
{\bf B}=\frac{2}{3},\
{\bf b}=\frac{1}{3}.$$

Let $\xi$ and $\zeta$ be two random variables with $\xi\sim DPH({\bf A})$ and $\zeta\sim DPH({\bf B})$.
\begin{itemize}

\item We have $\min (\xi ,\zeta )\sim DPH({\bf C})$, where
$${\bf C}={\bf A}\sqcap{\bf B}=\left(\begin{array}{ccc}
0 & \frac{4}{9} & 0\\
0 & 0 & \frac{1}{3}\\
0 & 0 & 0
\end{array}\right).$$

Then the transition probabilities towards absorption vector and TPM of the absorbing DTMC are
$${\bf c}=\left(\begin{array}{c}
\frac{5}{9}\\
\frac{2}{3}\\
1
\end{array}\right)\mbox{ and }
\check{\bf C}=\left(\begin{array}{cccc}
0 & \frac{4}{9} & 0 & \frac{5}{9}\\
0 & 0 & \frac{1}{3} & \frac{2}{3}\\
0 & 0 & 0 & 1\\
0 & 0 & 0 & 1
\end{array}\right).$$

\item We have $\max (\xi ,\zeta )\sim DPH({\bf D})$, where
$${\bf D}={\bf A}\sqcup{\bf B}=\left(\begin{array}{ccccccc}
0 & \frac{4}{9} & 0 & \frac{2}{9} & 0 & 0 & \frac{2}{9}\\
0 & 0 & \frac{1}{3} & 0 & 0 & \frac{1}{6} & \frac{1}{3}\\
0 & 0 & 0 & 0 & 0 & 0 & \frac{2}{3}\\
0 & 0 & 0 & 0 & \frac{2}{3} & 0 & 0\\
0 & 0 & 0 & 0 & 0 & \frac{1}{2} & 0\\
0 & 0 & 0 & 0 & 0 & 0 & 0\\
0 & 0 & 0 & 0 & 0 & 0 & \frac{2}{3}\\
\end{array}\right).$$

Then the transition probabilities towards absorption vector and TPM of the absorbing DTMC are
$${\bf d}=\left(\begin{array}{c}
\frac{1}{9}\\
\frac{1}{6}\\
\frac{1}{3}\\
\frac{1}{3}\\
\frac{1}{2}\\
1\\
\frac{1}{3}
\end{array}\right)\mbox{ and }
\check{\bf D}=\left(\begin{array}{cccccccc}
0 & \frac{4}{9} & 0 & \frac{2}{9} & 0 & 0 & \frac{2}{9} & \frac{1}{9}\\
0 & 0 & \frac{1}{3} & 0 & 0 & \frac{1}{6} & \frac{1}{3} & \frac{1}{6}\\
0 & 0 & 0 & 0 & 0 & 0 & \frac{2}{3} & \frac{1}{3}\\
0 & 0 & 0 & 0 & \frac{2}{3} & 0 & 0 & \frac{1}{3}\\
0 & 0 & 0 & 0 & 0 & \frac{1}{2} & 0 & \frac{1}{2}\\
0 & 0 & 0 & 0 & 0 & 0 & 0 & 1\\
0 & 0 & 0 & 0 & 0 & 0 & \frac{2}{3} & \frac{1}{3}\\
0 & 0 & 0 & 0 & 0 & 0 & 0 & 1
\end{array}\right).$$

\end{itemize}
\label{dphminmax.exm}
\end{example}

Let $Act=\{a,b,\ldots\}$ be the set of {\em elementary actions}. Then $\widehat{Act}=\{\hat{a},\hat{b},\ldots\}$ is the
set of {\em conjugated actions (conjugates)} such that $\hat{a}\neq a$ and $\hat{\hat{a}}=a$. Let ${\cal
A}=Act\cup\widehat{Act}$ be the set of {\em all actions}, and ${\cal L}=\nat_{fin}^{\cal A}$ be the set of {\em all
multiactions}. Note that $\emptyset\in{\cal L}$, this corresponds to an internal move, i.e. the execution of a
multiaction that contains no visible action names. The {\em alphabet} of $\alpha\in{\cal L}$ is defined as ${\cal
A}(\alpha )=\{x\in{\cal A}\mid\alpha (x)>0\}$.

A {\em timed (positively phased) multiaction} is a pair $(\alpha ,\partial_l^{\bf A})$, where $\alpha\in{\cal L},\ {\bf
A}\neq\emptyset$ is a non-empty transient TPM of the {\em positive delay} whose DPH distribution is represented by the
pair $({\bf e}_1,{\bf A})$, and $l\in\real_{>0}=(0;\infty )$ is the positive real-valued {\em weight} of the
multiaction $\alpha$. This weight is interpreted as a measure of importance (urgency, interest) or a bonus reward
associated with the {\em execution with probability $1$ of the timed multiaction while changing the current phase to
the zero phase}, i.e. at the time step at which the associated delay is finally expired. Such weights are used to
calculate the probability to execute the set of timed multiactions after their time delays when the execution
probability of each such timed multiaction is $1$. If the execution probability of each timed multiaction from the set
is less than $1$ then the weights of those timed multiactions are not taken into account. In such a case the overall
execution probability is calculated directly from the individual execution probabilities. Thus, the overall execution
probability can be calculated only if each individual execution probability is equal to $1$ or each of them is less
than $1$. We suppose a {\em priority} of the executions with probability $1$ (called {\em definite executions}) over
those with probabilities less than $1$ (called {\em probabilistic executions}). We do not mix the two kinds of
executions with a goal to avoid technical difficulties related to conditioning events with probability $0$, as well as
rather intricate many-variant computations. Let ${\cal TL}$ be the set of {\em all timed multiactions}.

The above reasoning on the probabilistic and definite executions of a timed multiaction is particularly inspired by the
two kinds of timed multiaction with the {\em one-element transient TPM} ${\bf A}={\cal A}_{11}$ of the delay. They have
the following analogues in dtsdPBC. A {\em geometric multiaction} $(\alpha ,\partial_l^{{\cal A}_{11}})$ with $0<{\cal
A}_{11}<1$ is also denoted by $(\alpha ,\rho )$ like a stochastic multiaction in dtsdPBC, where $\rho ={\cal
A}_1=1-{\cal A}_{11}$ is the {\em probability} of the multiaction, and its weight $l$ is omitted, as not used in the
resolving probabilistic choices. A {\em unit (tick) multiaction} $(\alpha ,\partial_l^0)$ with ${\cal A}_{11}=0$ is
also denoted by $(\alpha ,\natural_l^1)$ like a one-unit-delayed waiting multiaction in dtsdPBC, where the superscript
$1={\cal A}_1=1-{\cal A}_{11}$ is the {\em (fixed, deterministic) delay} of one time unit. Unit multiactions have a
{\em priority} over geometric ones, similar to waiting multiactions having a priority over stochastic ones in dtsdPBC.
The remaining kinds of timed multiactions cannot be ordered by priority, and they also cannot be compared by priority
with geometric and unit multiactions. In general, different executions of the same timed multiaction may have
probabilities as less than $1$ (like geometric multiactions), as equal to $1$ (like unit multiactions).

An {\em immediate (zero phased) multiaction} is a pair $(\alpha ,\partial_l^\emptyset )$, where $\alpha\in{\cal L},\
\emptyset$ is the empty transient TPM of the {\em zero delay} (whose DPH distribution is represented by the pair
$(\varepsilon ,\emptyset )$, with $\varepsilon$ being the zero length vector) and $l\in\real_{>0}$ is the positive
real-valued {\em weight} of the multiaction $\alpha$. This weight is interpreted like that for timed multiaction, but
it is associated with the {\em execution with probability $1$ of the immediate multiaction while staying in the zero
phase}. Immediate multiactions have a {\em priority} over timed ones. One can assume that all immediate multiactions
have (the highest) priority $1$ and all timed multiactions have (the lowest) priority $0$. This means that in a state
where the two kinds of multiactions can occur, immediate multiactions always occur before timed ones. Different types
of multiactions cannot participate together in some step (parallel execution), i.e. just the steps consisting only of
immediate multiactions or only of timed ones are allowed. An immediate multiaction $(\alpha ,\partial_l^\emptyset )$ is
also denoted by $(\alpha ,\natural_l^0)$ like immediate multiaction in dtsdPBC, where the superscript $0$ is the {\em
(fixed) delay} of zero time. Let ${\cal IL}$ be the set of {\em all immediate multiactions}.

Let us note that the same multiaction $\alpha\in{\cal L}$ may have different delays and weights in the same
specification. An {\em activity} is a timed or an immediate multiaction. Let ${\cal PHL}={\cal TL}\cup{\cal IL}$ be the
set of {\em all (phased) activities}. The {\em alphabet} of an activity $(\alpha ,\kappa )\in{\cal PHL}$ is defined as
${\cal A}((\alpha ,\kappa ))={\cal A}(\alpha )$. The {\em alphabet} of a multiset of activities
$\Upsilon\in\nat_{fin}^{\cal PHL}$ is defined as ${\cal A}(\Upsilon )=\cup_{(\alpha ,\kappa )\in\Upsilon}{\cal
A}(\alpha )$. For an activity $(\alpha ,\kappa )\in{\cal PHL}$ we define its {\em multiaction part} as ${\cal L}(\alpha
,\kappa )=\alpha$ and its {\em weight part} as $\Omega (\alpha ,\kappa )=l$, where $\kappa =\partial_l^{\bf A},\
l\in\real_{>0}$. The {\em multiaction part} of a multiset of activities $\Upsilon\in\nat_{fin}^{\cal PHL}$ is defined
as ${\cal L}(\Upsilon )=\sum_{(\alpha ,\kappa )\in\Upsilon}\alpha$.

Activities are combined into formulas (process expressions) by the following operations: {\em sequence} $;$, {\em
choice} $\cho$, {\em parallelism} $\|$, {\em relabeling} $[f]$ of actions, {\em restriction} $\!\!\rs\!\!$ over a
single action, {\em synchronization} $\!\!\sy\!\!$ on an action and its conjugate, and {\em iteration} $[\,*\,*\,]$
with three arguments: initialization, body and termination.

Sequence (sequential composition) and choice (choice composition) have a standard interpretation, like in other process
algebras, but parallelism (parallel composition) does not include synchronization, unlike the corresponding operation
in CCS \cite{Mil89}.

Relabeling functions $f:{\cal A}\rightarrow{\cal A}$ are bijections preserving conjugates, i.e. $\forall x\in{\cal A}\
f(\hat{x})=\widehat{f(x)}$. Relabeling is extended to multiactions in the usual way: for $\alpha\in{\cal L}$ we define
$f(\alpha )=\sum_{x\in\alpha}f(x)$. Relabeling is extended to activities: for $(\alpha ,\kappa )\in{\cal PHL}$ we
define $f(\alpha ,\kappa )=(f(\alpha ),\kappa )$. Relabeling is extended to the multisets of activities as follows: for
$\Upsilon\in\nat_{fin}^{\cal PHL}$ we define $f(\Upsilon )=\sum_{(\alpha ,\kappa )\in\Upsilon}(f(\alpha ),\kappa )$.

Restriction over an elementary action $a\in Act$ means that, for a given expression, any process behaviour containing
$a$ or its conjugate $\hat{a}$ is not allowed.

Let $\alpha ,\beta\in{\cal L}$ be two multiactions such that for some elementary action $a\in Act$ we have $a\in\alpha$
and $\hat{a}\in\beta$, or $\hat{a}\in\alpha$ and $a\in\beta$. Then, synchronization of $\alpha$ and $\beta$ by $a$ is
defined as $\alpha\oplus_a\beta =\gamma$, where

$$\gamma (x)=\left\{
\begin{array}{ll}
\alpha (x)+\beta (x)-1, & x=a\mbox{ or }x=\hat{a};\\
\alpha (x)+\beta (x), & \mbox{otherwise}.
\end{array}
\right.$$ In other words, we require that $\alpha\oplus_a\beta =\alpha +\beta -\{a,\hat{a}\}$, i.e. we remove one
exemplar of $a$ and one exemplar of $\hat{a}$ from the multiset sum $\alpha +\beta$, since the synchronization of $a$
and $\hat{a}$ produces $\emptyset$. Activities are synchronized with the use of their multiaction parts, i.e. the
synchronization by $a$ of two activities, whose multiaction parts $\alpha$ and $\beta$ possess the properties mentioned
above, results in the activity with the multiaction part $\alpha\oplus_a\beta$. We may synchronize activities of the
same type only: either both timed multiactions {\em with the same DPH distribution of delay} or both immediate ones,
since timed and immediate multiactions have different priorities, and diverse delays of multiactions would contradict
their joint timing. Hence, the multiactions of different types cannot be executed together (note also that the
execution of immediate multiactions takes no time, unlike that of timed ones). Synchronization by $a$ means that, for a
given expression with a process behaviour containing two concurrent activities that can be synchronized by $a$, there
exists also the process behaviour that differs from the former only in that the two activities are replaced with the
result of their synchronization.

In the iteration, the initialization subprocess is executed first, then the body is performed zero or more times, and
finally, the termination subprocess is executed.

\subsection{Process expressions}

Static expressions specify the structure of processes, i.e. how activities are com\-bined by operations in order to
con\-struct the composite process-algebraic formulas.

We assume that every timed multiaction has a {\em phase indicator} associated, whose value is the state (phase) of the
absorbing DTMC defining the DPH distribution of the corresponding delay. The zero phase is not written, since in this
phase the timed multiaction will be already executed and this will be explicit. Remember that the initial distribution
of the DTMC is always in the form of $\varpi =(1,0,\ldots ,0)={\bf e}_1$, so that the DPH distribution of the delay is
fully characterized by its transient TPM ${\bf A}$. Therefore, besides standard (non-superscribed) timed multiactions
in the form of $(\alpha ,\partial_l^{\bf A})\in{\cal TL}$ with ${\bf A}\neq\emptyset$, a special case of the {\em
superscribed} timed multiactions should be considered in the definition of static expressions. Each superscribed timed
multiaction in the form of $(\alpha ,\partial_l^{\bf A})^i$ has an extra superscript $i\in\{1,\ldots ,ord({\bf A})\}$
that indicates the {\em current} phase of the DPH distributed delay of that multiaction. The standard timed
multiactions have no superscripts, in order to demonstrate irrelevance of the delay phases for them. For example, their
phase indicators have not yet started or have already finished running after the absorbing DTMCs has entered the
absorbing state. The notions of the alphabet, multiaction part, weight part for (the multisets of) superscribed timed
multiactions are respectively defined like those for (the multisets of) non-superscribed timed multiactions.

By reasons of simplicity, we do not assign the delay phase superscripts $i$ to immediate multiactions, which are a
special (zero) case of phased multiactions $(\alpha ,\partial_l^\emptyset )$ with the delay $0$ represented by the
empty transient TPM $\emptyset$, since their delay phases can only be equal to $0$.

\begin{definition}
Let $(\alpha ,\kappa )\in{\cal PHL},\ (\alpha ,\partial_l^{\bf A})\in{\cal TL},\ i\in\{1,\ldots ,ord({\bf A})\}$ and
$a\in Act$. A {\em static expression} of dtsdPBC is defined as

$$E::=\ (\alpha ,\kappa )\mid (\alpha ,\partial_l^{\bf A})^i\mid E;E\mid E\cho E\mid E\| E\mid E[f]\mid
E\rs a\mid E\sy a\mid [E*E*E].$$

\end{definition}

Let $StatExpr$ denote the set of {\em all static expressions} of dtsdPBC.

To make the grammar above unambiguous, one can add parentheses in the productions with binary operations: $(E;E),\
(E\cho E),\ (E\| E)$. However, here and further we prefer the PBC approach and add them to resolve ambiguities only.

To avoid technical difficulties with the iteration operator, we should not allow any concurrency at the highest level
of the second argument of iteration. This is not a severe restriction though, since we can always prefix parallel
expressions by an activity with the empty multiaction part. Alternatively, we can use a different, safe, version of the
iteration operator, but its net translation has six arguments. See also \cite{BDK01} for discussion~on~this~subject.

\begin{definition}
Let $(\alpha ,\kappa )\in{\cal PHL},\ (\alpha ,\partial_l^{\bf A})\in{\cal TL},\ i\in\{1,\ldots ,ord({\bf A})\}$
and $a\in Act$. A {\em regular static expression} of dtsdPBC is defined as

$$\begin{array}{c}
E::=\ (\alpha ,\kappa )\mid (\alpha ,\partial_l^{\bf A})^i\mid E;E\mid E\cho E\mid E\| E\mid E[f]\mid E\rs a\mid
E\sy a\mid [E*D*E],\\
\mbox{where }D::=\ (\alpha ,\kappa )\mid (\alpha ,\partial_l^{\bf A})^i\mid D;E\mid D\cho D\mid D[f]\mid
D\rs a\mid D\sy a\mid [D*D*E].
\end{array}$$

\end{definition}

Let $RegStatExpr$ denote the set of {\em all regular static expressions} of dtsdPBC.

Let $E$ be a regular static expression. The {\em underlying phase-free regular static expression}
$\downharpoonleft\!\!E$ of $E$ is obtained by removing from it all delay phase superscripts.

The set of {\em all subexpressions of} a regular static expression $E$ is defined in a structural way. Let $(\alpha
,\kappa )\in{\cal PHL},\ (\alpha ,\partial_l^{\bf A})\in{\cal TL},\ i\in\{1,\ldots ,ord({\bf A})\},\ E,F,K\in
RegStatExpr$ and $a\in Act$. Then
\begin{enumerate}

\item $Sub((\alpha ,\kappa ))=\{(\alpha ,\kappa )\}$;

\item $Sub((\alpha ,\partial_l^{\bf A})^i)=\{(\alpha ,\partial_l^{\bf A})^i\}$;

\item $Sub(E\circ F)=\{E\circ F\}\cup Sub(E)\cup Sub(F)\ (\circ\in\{;,\cho ,\|\})$;

\item $Sub(E[f])=\{E[f]\}\cup Sub(E)$;

\item $Sub(E\circ a)=\{E\circ a\}\cup Sub(E)\ (\circ\in\{\!\!\rs\!\!,\!\!\sy\!\!\})$;

\item $Sub([E*F*K])=\{[E*F*K]\}\cup Sub(E)\cup Sub(F)\cup Sub(K)$.

\end{enumerate}

Further, the set of {\em all timed multiactions (from the syntax) of} a regular static expression $E$ is ${\cal
TL}(E)=\{(\alpha ,\partial_l^{\bf A})\mid ((\alpha ,\partial_l^{\bf A})\in Sub(E))\vee ((\alpha ,\partial_l^{\bf
A})^i\in Sub(E)),\ i\in\{1,\ldots ,ord({\bf A})\}$. The set of {\em all immediate multiactions (from the syntax) of}
$E$ is ${\cal IL}(E)=\{(\alpha ,\partial_l^\emptyset )\mid (\alpha ,\partial_l^\emptyset )\in Sub(E)\}$. Thus, the set
of {\em all activities (phased multiactions) (from the syntax) of} $E$ is ${\cal PHL}(E)={\cal TL}(E)\cup{\cal IL}(E)$.

The set of {\em all activities (phased multiactions) (from the syntax) of} a regular static expression can also be
defined in a structural way. Let $(\alpha ,\kappa )\in{\cal PHL},\ (\alpha ,\partial_l^{\bf A})\in{\cal TL},\
i\in\{1,\ldots ,ord({\bf A})\},\ E,F,K\in RegStatExpr$ and $a\in Act$. Then
\begin{enumerate}

\item ${\cal PHL}((\alpha ,\kappa ))=\{(\alpha ,\kappa )\}$;

\item ${\cal PHL}((\alpha ,\partial_l^{\bf A})^i)=\{(\alpha ,\partial_l^{\bf A})\}$;

\item ${\cal PHL}(E\circ F)={\cal PHL}(E)\cup{\cal PHL}(F)\ (\circ\in\{;,\cho ,\|\})$;

\item ${\cal PHL}(E[f])={\cal PHL}(E\rs a)={\cal PHL}(E\sy a)={\cal PHL}(E)$;

\item ${\cal PHL}([E*F*K])={\cal PHL}(E)\cup{\cal PHL}(F)\cup{\cal PHL}(K)$.

\end{enumerate}
The set of {\em all timed multiactions (from the syntax) of} $E$ is ${\cal TL}(E)={\cal PHL}(E)\cap{\cal TL}$. The set
of {\em all immediate multiactions (from the syntax) of} $E$ is ${\cal IL}(E)={\cal PHL}(E)\cap{\cal IL}$.

Dynamic expressions specify the states of processes, i.e. some particular points of the process behaviour. Dynamic
expressions are obtained from static ones, by annotating them with upper or lower bars which specify the active
components of the system at the current moment of time. The dynamic expression with upper bar (the overlined one)
$\overline{E}$ denotes the {\em initial}, and that with lower bar (the underlined one) $\underline{E}$ denotes the {\em
final} state of the process specified by a static expression $E$.

For every overlined superscribed timed multiaction in the form of $\overline{(\alpha ,\partial_l^{\bf A})^i}$, the
superscript $i\!\in\!\{1,\ldots ,ord({\bf A})\}$ specifies the {\em current} value of the {\em running} phase indicator
associated with the timed multiaction. That phase indicator is started with the {\em initial} value $ord({\bf A})$ (the
initial delay phase of that timed multiaction) at the moment when the timed multiaction becomes overlined. Then such a
newly overlined superscribed timed multiaction $\overline{(\alpha ,\partial_l^{\bf A})^{ord({\bf A})}}$ may be seen
similar to the freshly overlined non-superscribed timed multiaction $\overline{(\alpha ,\partial_l^{\bf A})}$. Such
similarity will be captured by the structural equivalence, to be defined later.

While the superscribed timed multiaction stays overlined with the specified process execution, the phase indicator
changes its value with each global time tick until the delay phase becomes $i\in\{1,\ldots ,ord({\bf A})\}$ such that
${\cal A}_i>0$, where ${\cal A}_i\ (1\leq i\leq ord({\bf A}))$ is the $i$-th element of the column vector ${\bf a}$ of
the transition probabilities towards absorption for the absorbing DTMC defining the DPH distributed delay of that timed
multiaction. Its execution in the phase $i$ should follow in the next moment with probability ${\cal A}_i$, in case
there are no conflicting with it immediate multiactions or conflicting timed multiactions with the similar property of
the delay phase, and it is not affected by restriction. An activity is said to be affected by restriction, if it is
within the scope of a restriction operation with the argument action, such that it or its conjugate is contained in the
multiaction part of that activity.

Thus, when overlined and superscribed, the timed multiaction $(\alpha ,\partial_l^{\bf A})$ can be executed in the phase $i$
probabilistically or definitely, depending on the (positive) value of ${\cal A}_i$. The {\em execution of the timed
multiaction $(\alpha ,\partial_l^{\bf A})$ in the phase $i$}, denoted by $(\alpha ,\partial_l^{\bf A})^i$, is {\em
probabilistic} if $0<{\cal A}_i<1$, or {\em definite} if ${\cal A}_i=1$. The set of {\em all probabilistic executions
(of timed multiactions)} is ${\cal PE}=\{(\alpha ,\partial_l^{\bf A})^i\mid (\alpha ,\partial_l^{\bf A})\in{\cal TL},\
i\in\{1,\ldots ,ord({\bf A})\},\ 0<{\cal A}_i<1\}$. The set of {\em all definite executions (of timed multiactions)} is
${\cal DE}=\{(\alpha ,\partial_l^{\bf A})^i\mid (\alpha ,\partial_l^{\bf A})\in{\cal TL},\ i\in\{1,\ldots ,ord({\bf
A})\},\ {\cal A}_i=1\}$. Definite executions have a {\em priority} over probabilistic ones. The set of {\em all timed
executions (executions of timed multiactions)} is ${\cal TE}={\cal PE}\cup{\cal DE}$. The executions of immediate
multiactions are called {\em immediate executions} and coincide with ${\cal IL}$, the set of all immediate
multiactions, since their zero delays have no positive phases to indicate. Immediate executions have a {\em priority}
over both definite and probabilistic ones, since immediate multiactions have a priority over timed multiactions.

\begin{definition}
Let $E\in StatExpr$ and $a\in Act$. A {\em dynamic expression} of dtsdPBC is defined as

$$G::=\ \overline{E}\mid\underline{E}\mid G;E\mid E;G\mid G\cho E\mid E\cho G\mid G\| G\mid G[f]\mid G\rs a\mid
G\sy a\mid [G*E*E]\mid [E*G*E]\mid [E*E*G].$$

\end{definition}

Let $DynExpr$ denote the set of {\em all dynamic expressions} of dtsdPBC.

Let $G$ be a dynamic expression. The {\em underlying static (line-free) expression} $\lfloor G\rfloor$ of $G$ is
obtained by removing from it all upper and lower bars.

\begin{definition}
A dynamic expression $G$ is {\em regular} if its underlying static expression $\lfloor G\rfloor$ is regular.
\end{definition}

Let $RegDynExpr$ denote the set of {\em all regular dynamic expressions} of dtsdPBC.

Let $G$ be a regular dynamic expression. The {\em underlying phase-free regular dynamic expression}
$\downharpoonleft\!\!G$ of $G$ is obtained by removing from it all delay phase superscripts.

Further, the set of {\em all timed multiactions (from the syntax) of} $G$ is ${\cal TL}(G)={\cal TL}(\lfloor G\rfloor
)$. The set of {\em all immediate multiactions (from the syntax) of} $G$ is ${\cal IL}(G)={\cal IL}(\lfloor G\rfloor
)$. Thus, the set of {\em all activities (from the syntax) of} $G$ is ${\cal PHL}(G)={\cal TL}(G)\cup{\cal IL}(G)$.

\section{Operational semantics}
\label{opersem.sec}

In this section, we define the step operational semantics in terms of labeled transition systems.

\subsection{Inaction rules}

The inaction rules for dynamic expressions describe their structural transformations in the form of
$G\Rightarrow\widetilde{G}$ which do not change the states of the specified processes. The goal of those syntactic
transformations is to obtain the well-structured resulting expressions called operative ones to which no inaction rules
can be further~applied.

Thus, an application of every inaction rule does not require any discrete time delay, i.e. the dynamic expression
transformation described by the rule is accomplished instantly.

In Table \ref{inactrulesphm1.tab}, we define inaction rules for regular dynamic expressions being overlined and
underlined static ones. In this table, $(\alpha ,\partial_l^{\bf A})\in{\cal TL},\ i\in\{1,\ldots ,ord({\bf A})\},\
E,F,K\in RegStatExpr$ and $a\in Act$. The first inaction rule sug\-gests that the delay phase of each newly overlined
timed multiaction is set to the order of the absorbing DTMC defining the DPH distributed delay of that timed
multiaction.

\begin{table}[h]
\caption{Inaction rules for overlined and underlined regular static expressions}
\label{inactrulesphm1.tab}
\begin{center}
$\small\begin{array}{|llll|}
\hline
\rule{0mm}{4mm}
\overline{(\alpha ,\partial_l^{\bf A})}\Rightarrow\overline{(\alpha ,\partial_l^{\bf A})^{ord({\bf A})}} &
\overline{E;F}\Rightarrow\overline{E};F &
\underline{E};F\Rightarrow E;\overline{F} &
E;\underline{F}\Rightarrow\underline{E;F}\\[1mm]

\overline{E\cho F}\Rightarrow\overline{E}\cho F &
\overline{E\cho F}\Rightarrow E\cho\overline{F} &
\underline{E}\cho F\Rightarrow\underline{E\cho F} &
E\cho\underline{F}\Rightarrow\underline{E\cho F}\\[1mm]

\overline{E\| F}\Rightarrow\overline{E}\|\overline{F} &
\underline{E}\|\underline{F}\Rightarrow\underline{E\| F} &
\overline{E[f]}\Rightarrow\overline{E}[f] &
\underline{E}[f]\Rightarrow\underline{E[f]}\\[1mm]

\overline{E\rs a}\Rightarrow\overline{E}\rs a &
\underline{E}\rs a\Rightarrow\underline{E\rs a} &
\overline{E\sy a}\Rightarrow\overline{E}\sy a &
\underline{E}\sy a\Rightarrow\underline{E\sy a}\\[1mm]

\overline{[E*F*K]}\Rightarrow [\overline{E}*F*K] &
[\underline{E}*F*K]\Rightarrow [E*\overline{F}*K] &
[E*\underline{F}*K]\Rightarrow [E*\overline{F}*K] &
[E*\underline{F}*K]\Rightarrow [E*F*\overline{K}]\\[1mm]

[E*F*\underline{K}]\Rightarrow\underline{[E*F*K]} & & & \\[1mm]
\hline
\end{array}$
\end{center}
\end{table}

In Table \ref{inactrulesphm2.tab}, we introduce inaction rules for regular dynamic expressions in the arbitrary form. In
this table, $E,F\in RegStatExpr,\ G,H,\widetilde{G},\widetilde{H}\in RegDynExpr$ and $a\in Act$. By reason of brevity,
two distinct inaction rules with the same premises are collated in some cases, resulting in the inaction rules with
double conclusion.

\begin{table}[h]
\caption{Inaction rules for arbitrary regular dynamic expressions}
\label{inactrulesphm2.tab}
\begin{center}
$\small\begin{array}{|llll|}
\hline
\rule{0mm}{6.5mm}
\dfrac{G\Rightarrow\widetilde{G},\ \circ\in\{;,\cho\}}{G\circ E\Rightarrow\widetilde{G}\circ E,\ E\circ G\Rightarrow
E\circ\widetilde{G}} &
\dfrac{G\Rightarrow\widetilde{G}}{G\| H\Rightarrow\widetilde{G}\| H,\ H\| G\Rightarrow H\|\widetilde{G}} &
\dfrac{G\Rightarrow\widetilde{G}}{G[f]\Rightarrow\widetilde{G}[f]} &
\dfrac{G\Rightarrow\widetilde{G},\ \circ\in\{\!\!\rs\!\!,\!\!\sy\!\!\}}{G\circ a\Rightarrow\widetilde{G}\circ a}\\[4mm]

\dfrac{G\Rightarrow\widetilde{G}}{[G*E*F]\Rightarrow [\widetilde{G}*E*F]} &
\multicolumn{3}{l|}{\dfrac{G\Rightarrow\widetilde{G}}{[E*G*F]\Rightarrow [E*\widetilde{G}*F],\ [E*F*G]\Rightarrow
[E*F*\widetilde{G}]}}\\[4mm]
\hline
\end{array}$
\end{center}
\end{table}

\begin{example}
Let $E=(\{a\},\partial_1^{\bf A})\cho (\{b\},\partial_2^{\bf B})$, where ${\bf A}$ is the transient TPM of the
(discrete) uniform distribution $Unif(1,3)$ ($DPH({\bf A})=Unif(1,3)$) and ${\bf B}$ is the transient TPM of the
geometric distribution $Geom(\frac{1}{3})$ ($DPH({\bf B})=Geom(\frac{1}{3})$), i.e. ${\bf A}$ and ${\bf B}$ are from
Example \ref{dphminmax.exm}. The following inferences by the inaction rules are possible from $\overline{E}$:
$$\begin{array}{l}
\overline{(\{a\},\partial_1^{\bf A})\cho (\{b\},\partial_2^{\bf B})}\Rightarrow
\overline{(\{a\},\partial_1^{\bf A})}\cho (\{b\},\partial_2^{\bf B})\Rightarrow
\overline{(\{a\},\partial_1^{\bf A})^3}\cho (\{b\},\partial_2^{\bf B}),\\
\overline{(\{a\},\partial_1^{\bf A})\cho (\{b\},\partial_2^{\bf B})}\Rightarrow
(\{a\},\partial_1^{\bf A})\cho\overline{(\{b\},\partial_2^{\bf B})}\Rightarrow
(\{a\},\partial_1^{\bf A})\cho\overline{(\{b\},\partial_2^{\bf B})^1}.
\end{array}$$
\label{overlinph.exm}
\end{example}

\begin{definition}
A regular dynamic expression $G$ is {\em operative} if no inaction rule can be applied to it.
\end{definition}

Let $OpRegDynExpr$ denote the set of {\em all operative regular dynamic expressions} of dtsdPBC.

Note that any dynamic expression can be always transformed into a (not necessarily unique) operative one by using the
inaction rules.

In the following, we consider regular expressions only and omit the word ``regular''.

\begin{definition}
The relation $\approx\ =(\Rightarrow\cup\Leftarrow )^*$ is a {\em structural equivalence} of dynamic expressions in
dtsdPBC, where $^*$ is the reflexive and transitive closure operation. Thus, two dynamic expressions $G$ and $G'$ are
{\em structurally equivalent}, denoted by $G\approx G'$, if they can be reached from each other by applying the
inaction rules in a forward or a backward direction.
\end{definition}

Let $X$ be some set. We denote the Cartesian product $X\times X$ by $X^2$. Let ${\cal E}\subseteq X^2$ be an
equivalence relation on $X$. Then the {\em equivalence class} (with respect to ${\cal E}$) of an element $x\in X$ is
defined by $[x]_{\cal E}=\{y\in X\mid (x,y)\in{\cal E}\}$. The equivalence ${\cal E}$ partitions $X$ into the {\em set
of equivalence classes} $X/_{\cal E}=\{[x]_{\cal E}\mid x\in X\}$.

Let $G$ be a dynamic expression. Then $[G]_\approx =\{H\mid G\approx H\}$ is the equivalence class of $G$ with respect
to the structural equivalence, called the (corresponding) {\em state}. Next, $G$ is an {\em initial} dynamic
expression, denoted by $init(G)$, if $\exists E\in RegStatExpr\ G\in [\overline{E}]_\approx$. Further, $G$ is a {\em
final} dynamic expression, denoted by $final(G)$, if $\exists E\in RegStatExpr\ G\in [\underline{E}]_\approx$.

\begin{example}
Let $E$ be from Example \ref{overlinph.exm}. We have $init(\overline{E})$ and $[\overline{E}]_\approx =
\{\overline{(\{a\},\partial_1^{\bf A})}\cho (\{b\},\partial_2^{\bf B}),\\
\overline{(\{a\},\partial_1^{\bf A})^3}\cho (\{b\},\partial_2^{\bf B}),
\overline{(\{a\},\partial_1^{\bf A})}\cho (\{b\},\partial_2^{\bf B})^1,
\overline{(\{a\},\partial_1^{\bf A})^3}\cho (\{b\},\partial_2^{\bf B})^1,
(\{a\},\partial_1^{\bf A})\cho\overline{(\{b\},\partial_2^{\bf B})},
(\{a\},\partial_1^{\bf A})^3\cho\overline{(\{b\},\partial_2^{\bf B})},\\
(\{a\},\partial_1^{\bf A})\cho\overline{(\{b\},\partial_2^{\bf B})^1},
(\{a\},\partial_1^{\bf A})^3\cho\overline{(\{b\},\partial_2^{\bf B})^1},
\overline{(\{a\},\partial_1^{\bf A})\cho (\{b\},\partial_2^{\bf B})},
\overline{(\{a\},\partial_1^{\bf A})^3\cho (\{b\},\partial_2^{\bf B})},
\overline{(\{a\},\partial_1^{\bf A})\cho (\{b\},\partial_2^{\bf B})^1},\\
\overline{(\{a\},\partial_1^{\bf A})^3\cho (\{b\},\partial_2^{\bf B})^1}\}$.
Then $[\overline{E}]_\approx\cap OpRegDynExpr=
\{\overline{(\{a\},\partial_1^{\bf A})^3}\cho (\{b\},\partial_2^{\bf B}),
\overline{(\{a\},\partial_1^{\bf A})^3}\cho(\{b\},\partial_2^{\bf B})^1,\\
(\{a\},\partial_1^{\bf A})\cho\overline{(\{b\},\partial_2^{\bf B})^1},
(\{a\},\partial_1^{\bf A})^3\cho\overline{(\{b\},\partial_2^{\bf B})^1}\}$.
\label{streqclsph.exm}
\end{example}

The set of all subexpressions of a dynamic expression is defined like that of a regular static expression. Let $G$ be a
dynamic expression and $s=[G]_\approx$ be the (corresponding) state. The set of {\em all enabled timed multiactions of
$s$} is $EnaTime(s)=\{(\alpha ,\partial_l^{\bf A})\in{\cal TL}\mid\exists H\in s\cap OpRegDynExpr\ \overline{(\alpha
,\partial_l^{\bf A})^i}\in Sub(H),\ i\in\{1,\ldots ,ord({\bf A})\}\}$, i.e. it consists of all timed multiactions that,
being superscribed with their delay phases and overlined, are the subexpressions of some operative dynamic expression
from the state $s$. The set of {\em all newly enabled timed multiactions of $s$} is $EnaTimeNew(s)=\{(\alpha
,\partial_l^{\bf A})\in{\cal TL}\mid\exists H\in s\cap OpRegDynExpr\ \overline{(\alpha ,\partial_l^{\bf A})^{ord({\bf
A})}}\in Sub (H)\}$, i.e. it consists of all timed multiactions that, being superscribed with their {\em initial} delay
phases (the orders of the absorbing DTMCs defining the delays of those timed multiactions) and overlined, are the
subexpressions of some operative dynamic expression from the state $s$. Analogously, the set of {\em all enabled
immediate multiactions of $s$} is $EnaImm(s)=\{(\alpha ,\partial_l^\emptyset )\in{\cal IL}\mid\exists H\in s\cap
OpRegDynExpr\ \overline{(\alpha ,\partial_l^\emptyset )}\in Sub(H)\}$.

Thus, the set of {\em all enabled activities of $s$} is $Ena(s)=EnaTime(s)\cup EnaImm(s)$. Note that the activities,
resulted from syn\-chro\-ni\-za\-ti\-on, are not present explicitly in the syntax of the dynamic expressions.
Nevertheless, their enabledness status can be recovered by observing that of the pair of synchronized activities from
the syntax (they both should be enabled for enabling their synchronous product), even if they are affected by
restriction after the synchronization.

\begin{example}
Let $E$ be from Example \ref{overlinph.exm}. Then we have $EnaTime([\overline{E}]_\approx )=
EnaTimeNew([\overline{E}]_\approx )=\\
\{(\{a\},\partial_1^{\bf A}),(\{b\},\partial_2^{\bf B})\}$ and $EnaImm([\overline{E}]_\approx )=\emptyset$, hence,
$Ena([\overline{E}]_\approx )=\{(\{a\},\partial_1^{\bf A}),(\{b\},\partial_2^{\bf B})\}$.
\label{enablednsph.exm}
\end{example}

\begin{definition}
An operative dynamic expression $G$ is {\em saturated (with the delay phases)}, if each enabled timed multiaction of
$[G]_\approx$, being superscribed with their delay phases and possibly overlined, is the subexpression of $G$.
\end{definition}

Let $SatOpRegDynExpr$ denote the set of {\em all saturated operative dynamic expressions} of dtsdPBC.

\begin{proposition}
Any operative dynamic expression can be always transformed into the saturated one by applying the inaction rules in a
forward or a backward direction.
\label{saturation.pro}
\end{proposition}
{\em Proof.} Let $G$ be a dynamic expression, $(\alpha ,\partial_l^{\bf A})\in EnaTime([G]_\approx )$ with $ord({\bf
A})\geq 2$ and there exists $H\in [G]_\approx\cap OpRegDynExpr$ that contains a subexpression $\overline{(\alpha
,\partial_l^{\bf A})^i},\ i\in\{1,\ldots ,ord({\bf A})-1\}$. Then all operative dynamic expressions from
$[G]_\approx\cap OpRegDynExpr$ contain a subexpression $\overline{(\alpha ,\partial_l^{\bf A})^i}$ or $(\alpha
,\partial_l^{\bf A})^i$, i.e. the (possibly overlined) enabled timed multiaction $(\alpha ,\partial_l^{\bf A})$ with
the non-initial delay phase superscript $i$. Note that the delay phase superscript $i$ is the same for all such
structurally equivalent operative dynamic expressions. Indeed, all inaction rules, besides the first one, do not change
the values of delay phases, but only modify the overlines and underlines of dynamic expressions. The first inaction
rule just sets up the phase indicator of each overlined timed multiaction $\overline{(\alpha ,\partial_l^{\bf A})}$
with the initial value $i=ord({\bf A})$, equal to the initial delay phase of that timed multiaction, as follows:
$\overline{(\alpha ,\partial_l^{\bf A})^{ord({\bf A})}}$. Then the remaining inaction rules can shift out the overline
of that enabled timed multiaction before setting up its delay phase, which results in a non-overlined enabled timed
multiaction without delay phase superscript $(\alpha ,\partial_l^{\bf A})$. Thus, for $(\alpha ,\partial_l^{\bf A})\in
EnaTime([G]_\approx )$, it may happen that $\overline{(\alpha ,\partial_l^{\bf A})^{ord({\bf A})}}$ a subexpression of
some $H\in [G]_\approx\cap OpRegDynExpr$ and $(\alpha ,\partial_l^{\bf A})$ is a subexpression of a different $H'\in
[G]_\approx\cap OpRegDynExpr$.

Let now $G$ be an operative dynamic expression that is not saturated. By the arguments above, the saturation can be
violated only if $G$ contains as a subexpression at least one newly enabled timed multiaction $(\alpha ,\partial_l^{\bf
A})$ of $[G]_\approx$ that is not superscribed with the delay phase. By the definition of the new-enabling, there
exists $H\in [G]_\approx\cap OpRegDynExpr$ such that $\overline{(\alpha ,\partial_l^{\bf A})^{ord({\bf A})}}$ is a
subexpression of $H$. Since $G\approx H$, there is a sequence of the inaction rules applications (in a forward or a
backward direction) that transforms $G$ into $H$. Then the reverse sequence transforms $H$ into $G$. Let us remove from
that reverse sequence the following backward application of the first inaction rule: $\overline{(\alpha
,\partial_l^\theta )}\Leftarrow\overline{(\alpha ,\partial_l^{\bf A})^{ord({\bf A})}}$. Then such a reduced reverse
sequence will turn $H$ into $G_1\in [G]_\approx\cap OpRegDynExpr$, obtained from $G$ by replacing $(\alpha
,\partial_l^{\bf A})$ with $(\alpha ,\partial_l^{\bf A})^{ord({\bf A})}$. Let us start from $G_1$ and apply the above
procedure to the remaining non-superscribed with the delay phases newly enabled timed multiactions of $[G]_\approx
=[G_1]_\approx$. After repeated application of the mentioned procedure for all $n\geq 1$ non-superscribed newly enabled
timed multiactions of $G$, we shall get from it the saturated operative dynamic expression $G_n=\widetilde{G}\in
[G]_\approx\cap OpRegDynExpr$. \hfill $\eop$

Thus, any dynamic expression can be always transformed into a (not necessarily unique) saturated operative one by
(possibly reverse) applying the inaction rules.

\begin{example}
Let $E$ be from Example \ref{overlinph.exm}. We have $[\overline{E}]_\approx\cap SatOpRegDynExpr=
\{\overline{(\{a\},\partial_1^{\bf A})^3}\cho (\{b\},\partial_2^{\bf B})^1,\\
(\{a\},\partial_1^{\bf A})^3\cho\overline{(\{b\},\partial_2^{\bf B})^1}\}$. Consider the sequence of inaction rules,
applied (in a forward or a backward direction) while the following trans\-formation of a non-saturated $G\in
[\overline{E}]_\approx\cap OpRegDynExpr$ with the non-superscribed with the phase values enabled timed multiactions
$(\{a\},\partial_1^{\bf A})$ and $(\{b\},\partial_2^{\bf B})$. In this way, $G$ is transformed into a saturated $H\in
[\overline{E}]_\approx\cap OpRegDynExpr$, in which $(\{a\},\partial_1^{\bf A})$ and $(\{b\},\partial_2^{\bf B})$ are
superscribed:
$$\begin{array}{c}
G=(\{a\},\partial_1^{\bf A})\cho\overline{(\{b\},\partial_2^{\bf B})}\approx
(\{a\},\partial_1^{\bf A})\cho\overline{(\{b\},\partial_2^{\bf B})^1}\approx
\overline{(\{a\},\partial_1^{\bf A})\cho (\{b\},\partial_2^{\bf B})^1}\approx
\overline{(\{a\},\partial_1^{\bf A})}\cho (\{b\},\partial_2^{\bf B})^1\approx\\
\overline{(\{a\},\partial_1^{\bf A})^3}\cho (\{b\},\partial_2^{\bf B})^1=H.
\end{array}$$

The reduced reverse sequence of inaction rules induces the following transformations of $H$ that result in a saturated
and identically (to $G$) overlined $G_1=\widetilde{G}\in [\overline{E}]_\approx\cap OpRegDynExpr$, in which
$(\{a\},\partial_1^{\bf A})$ and $(\{b\},\partial_2^{\bf B})$ are superscribed:
$$\begin{array}{c}
H=\overline{(\{a\},\partial_1^{\bf A})^3}\cho (\{b\},\partial_2^{\bf B})^1\approx
\overline{(\{a\},\partial_1^{\bf A})^3\cho (\{b\},\partial_2^{\bf B})^1}\approx
(\{a\},\partial_1^{\bf A})^3\cho\overline{(\{b\},\partial_2^{\bf B})^1}=G_1=\widetilde{G}.
\end{array}$$
\label{saturateph.exm}
\end{example}

\subsection{Action and empty move rules}

The action rules are applied when some activities are executed. With these rules we capture the prioritization among
different types of multiactions. We also have the empty move rules which are used to capture a delay of one discrete
time unit when no timed or immediate multiactions are executable. In this case, the empty multiset of activities is
executed. The action and empty move rules will be used later to determine all executions from the structural
equivalence class of every dynamic expression (i.e. from the state of the corresponding process). This information
together with that about probabilities or weights of the activities to be executed from the current process state will
be used to calculate the probabilities of such executions resulting to the next process~states.

The action rules with probabilistic executions (definite executions or executions of immediate multiactions, called
immediate executions), describe dynamic expression transformations in the form of
$G\stackrel{\Gamma}{\rightarrow}\widetilde{G}$ ($G\stackrel{V}{\rightarrow}\widetilde{G}$ or
$G\stackrel{I}{\rightarrow}\widetilde{G}$, respectively) with non-empty multisets $\Gamma$ of probabilistic executions
($V$ of definite or $I$ of immediate executions, respectively). The rules represent possible state changes of the
specified processes when accomplishing some non-empty multisets of probabilistic executions (definite or immediate
executions, respectively).

The empty move rules describes dynamic expression transformations in the form of
$G\stackrel{\emptyset}{\rightarrow}\widetilde{G}$, called the {\em empty moves}, due to execution of the empty multiset
of activities at a discrete time tick while making the transition from positive to positive delay phases of the enabled
timed multiactions from the syntax of $G$. When no delay phases are changed within $G$ with the empty multiset
execution at the next moment, we have $\widetilde{G}=G$. For example, this is the case if $G$ contains no superscribed
timed multiactions, or all timed multiactions are superscribed with the positive delay phases from which the
transitions with probability $1$ exist to the zero phases while those superscribed timed multiactions are either
affected by restriction or not overlined. In such a case, the empty move from $G$ is in the form of
$G\stackrel{\emptyset}{\rightarrow}G$, called the {\em empty loop}.

Thus, an application of every action rule with timed multiactions or the (basic) empty move rule requires one discrete
time unit delay, i.e. accomplishing (possibly empty) multiset of executions of timed multiactions leading to the
dynamic expression transformation described by the rule is accomplished instantly after one time unit. An application
of every action rule with immediate executions does not take any time, i.e. the execution of a (non-empty) multiset of
immediate multiactions is accomplished instantly at the current moment of time.

Note that expressions of dtsdPBC can contain identical activities. To avoid technical difficulties, such as the proper
calculation of the state change probabilities for multiple transitions, we can always enumerate coinciding activities
from left to right in the syntax of expressions. The new activities, resulted from synchronization will be annotated
with concatenation of numberings of the activities they come from, hence, the numbering should have a tree structure to
reflect the effect of multiple synchronizations. We now define the numbering which encodes a binary tree with the
leaves labeled by natural numbers.

\begin{definition}
The {\em numbering} of expressions is defined as $\iota ::=\ n\mid (\iota )(\iota )$, where $n\in\nat$.
\end{definition}

Let $Num$ denote the set of {\em all numberings} of expressions.

The new activities resulting from synchronizations in different orders should be considered up to permutation of their
numbering. In this way, we shall recognize different instances of the same activity. If we compare the contents of
different numberings, i.e. the sets of natural numbers in them, we shall be able to identify the mentioned instances.

The {\em content} of a numbering $\iota\in Num$ is

$$Cont(\iota )=\left\{
\begin{array}{ll}
\{\iota\}, & \iota\in\nat ;\\
Cont(\iota_1)\cup Cont(\iota_2), & \iota =(\iota_1)(\iota_2).
\end{array}
\right.$$

After the enumeration, the multisets of activities from the expressions will become the proper sets. In the following,
we suppose that the identical activities are enumerated when needed to avoid ambiguity. This enumeration is considered
to be implicit.

\begin{definition}
Let $G\in OpRegDynExpr$. We now define the {\em set of all non-empty multisets of potential executions}, denoted by
$Can(G)$. Let $(\alpha ,\partial_l^{\bf A})\in{\cal PHL},\ {\cal A}_{ij}\ (1\leq i,j\leq ord({\bf A}))$ are the
elements of the transient TPM ${\bf A}$, ${\cal A}_i\ (1\leq i\leq ord({\bf A}))$ are the elements of vector ${\bf a}$
of the transition probabilities towards absorption, $E,F\in RegStatExpr,\ H\in OpRegDynExpr$ and $a\in Act$.
\begin{enumerate}

\item If $final(G)$ then $Can(G)=\emptyset$.

\item If $G=\overline{(\alpha ,\partial_l^{\bf A})^i}$ and $l\in\real_{>0},\ {\bf A}\neq\emptyset ,\ i\in\{1,\ldots
,ord({\bf A})\},\ {\cal A}_i=0$, then $Can(G)=\emptyset$.

\item If $G=\overline{(\alpha ,\partial_l^{\bf A})^i}$ and $l\in\real_{>0},\ {\bf A}\neq\emptyset ,\ i\in\{1,\ldots
,ord({\bf A})\},\ {\cal A}_i>0$, then $Can(G)=\{\{(\alpha ,\partial_l^{\bf A})^i\}\}$.

\item If $G=\overline{(\alpha ,\partial_l^\emptyset )}$ and $l\in\real_{>0}$, then $Can(G)=\{\{(\alpha
,\partial_l^\emptyset ))\}\}$.

\item If $\Upsilon\in Can(G)$ then $\Upsilon\in Can(G\circ E),\ \Upsilon\in Can(E\circ G)\ (\circ\in\{;,\cho\}),\
\Upsilon\in Can(G\| H),\ \Upsilon\in Can(H\| G),\\
f(\Upsilon )\in Can(G[f]),\ \Upsilon\in Can(G\rs a)\ (\mbox{when }a,\hat{a}\not\in{\cal A}(\Upsilon )),\
\Upsilon\in Can(G\sy a),\ \Upsilon\in Can([G*E*F]),\\
\Upsilon\in Can([E*G*F]),\ \Upsilon\in Can([E*F*G])$.

\item If $\Upsilon\in Can(G)$ and $\Xi\in Can(H)$ then $\Upsilon +\Xi\in Can(G\| H)$.

\item If $\Upsilon\in Can(G\sy a)$ and $(\alpha ,\partial_l^{\bf A})^i,(\beta ,\partial_m^{\bf B})^j\in\Upsilon$ are
different both probabilistic or both definite executions (i.e. ${\cal A}_i=1\ \Leftrightarrow\ {\cal B}_j=1$)
such that $a\in\alpha ,\ \hat{a}\in\beta ,\ l,m\in\real_{>0}$, then $\Upsilon -\{(\alpha ,\partial_l^{\bf A})^i,(\beta
,\partial_m^{\bf B})^j\}+\{(\alpha\oplus_a\beta ,\partial_{l+m}^{{\bf A}\sqcup{\bf B}})^{(i,j)}\}\in Can(G\sy a)$.

\item If $\Upsilon\in Can(G\sy a)$ and $(\alpha ,\partial_l^\emptyset ),(\beta ,\partial_m^\emptyset )\in\Upsilon$ are
different immediate executions such that $a\in\alpha ,\ \hat{a}\in\beta ,\ l,m\in\real_{>0}$, then $\Upsilon
-\{(\alpha ,\partial_l^\emptyset ),(\beta ,\partial_m^\emptyset )\}+\{(\alpha\oplus_a\beta ,\partial_{l+m}^\emptyset
)\}\in Can(G\sy a)$.

When we synchronize the same multiset of executions in different orders, we obtain several executions with the same
multiaction, weight, transient TPM (up to permutation in case of probabilistic executions) and delay phase (up to
permutation in case of probabilistic executions) parts, but with different numberings having the same content. Then
we only consider a single one of the resulting activities to avoid introducing redundant ones.

For example, the synchronization of probabilistic executions $(\alpha ,\partial_l^{\bf A})_1^i$ and $(\beta
,\partial_m^{\bf B})_2^j$ in different orders generates the probabilistic executions $(\alpha\oplus_a\beta
,\partial_{l+m}^{{\bf A}\sqcup{\bf B}})_{(1)(2)}^{(i,j)}$ and $(\beta\oplus_a\alpha ,\partial_{m+l}^{{\bf
B}\sqcup{\bf A}})_{(2)(1)}^{(j,i)}$. Since $Cont((1)(2))=\{1,2\}=Cont((2)(1))$, in both cases, only the first
execution (or, symmetrically, the second one) resulting from synchronization will appear in a multiset from
$Can(G\sy a)$. Note that ${\bf A}\sqcup{\bf B}$ and ${\bf B}\sqcup{\bf A}$ may differ, but they have exactly the
same entries and are {\em permutation equivalent}, since so are ${\bf A}\otimes{\bf B}$ and ${\bf B}\otimes{\bf
A}$.

\end{enumerate}
\end{definition}

Note that if $\Upsilon\in Can(G)$ then by definition of $Can(G),\ \forall\Xi\subseteq\Upsilon ,\ \Xi\neq\emptyset$, we
have $\Xi\in Can(G)$.

Let $G\in OpRegDynExpr$ and $Can(G)\neq\emptyset$. Obviously, if there are only probabilistic (definite or immediate,
respectively) executions in the multisets from $Can(G)$ then these probabilistic (definite or immediate) executions can
be accomplished from $G$. Otherwise, besides probabilistic ones, there are also definite and/or immediate executions in
the multisets from $Can(G)$. By the note above, there are non-empty multisets of definite and/or immediate executions
in $Can(G)$ as well, i.e. $\exists\Upsilon\in Can(G)\ \Upsilon\in\nat_{fin}^{{\cal DE}\cup{\cal
IL}}\setminus\{\emptyset\}$. In this case, no probabilistic executions can be accomplished from $G$, even if $Can(G)$
contains non-empty multisets of probabilistic executions, since definite and immediate executions have a priority over
probabilistic ones, and should be accomplished first. Further, if there are no probabilistic, but both definite and
immediate executions in the multisets from $Can(G)$, then, analogously, no definite executions can be accomplished from
$G$, since immediate executions have a priority over definite ones (besides that over probabilistic ones).

When there are only definite and, possibly, probabilistic executions in the multisets from $Can(G)$ then, from above,
only definite ones can be accomplished from $G$. Then just {\em maximal} non-empty multisets of definite executions can be
accomplished from $G$, since all non-conflicting definite executions cannot wait anymore and they should occur at the next
time moment with probability $1$. The next definition formalizes these requirements.

\begin{definition}
Let $G\in OpRegDynExpr$. The {\em set of all non-empty multisets of executions which can be accomplished from $G$} is

$$Now(G)=
\left\{
\begin{array}{ll}
Can(G)\cap\nat_{fin}^{\cal IL}, & Can(G)\cap\nat_{fin}^{\cal IL}\neq\emptyset ;\\
\{W\in Can(G)\cap\nat_{fin}^{\cal DE}\mid\forall V\in Can(G)\cap\nat_{fin}^{\cal DE}\ W\subseteq V\ \Rightarrow\
V\!=\!W\}, &
(Can(G)\cap\nat_{fin}^{\cal IL}=\emptyset )\wedge\\
 & (Can(G)\cap\nat_{fin}^{\cal DE}\neq\emptyset );\\
Can(G), & \mbox{otherwise}.
\end{array}
\right.$$

\end{definition}

Consider an operative dynamic expression $G\in OpRegDynExpr$. The expression $G$ is {\em p-tangible (probabilistically
tangible)}, denoted by $ptang(G)$, if $Now(G)\subseteq\nat_{fin}^{\cal PE}\setminus\{\emptyset\}$. In particular, we
have $ptang(G)$, if $Now(G)=\emptyset$. The expression $G$ is {\em d-tangible (definitely tangible)}, denoted by
$dtang(G)$, if $\emptyset\neq Now(G)\subseteq\nat_{fin}^{\cal DE}\setminus\{\emptyset\}$. The expression $G$ is {\em
tangible}, denoted by $tang(G)$, if $ptang(G)$ or $dtang(G)$, i.e. $Now(G)\subseteq (\nat_{fin}^{\cal
PE}\cup\nat_{fin}^{\cal DE})\setminus\{\emptyset\}$. Again, we particularly have $tang(G)$, if $Now(G)=\emptyset$.
Otherwise, the expression $G$ is {\em vanishing}, denoted by $vanish(G)$, and in this case $\emptyset\neq
Now(G)\subseteq\nat_{fin}^{\cal IL}\setminus\{\emptyset\}$. Note that the operative dynamic expressions from
$[G]_\approx$ may have different types in general.

\begin{example}
Let ${\bf A}$ be the transient TPM of the (discrete) uniform distribution $Unif(1,3)$ ($DPH({\bf A})=Unif(1,3)$) and
${\bf B}$ be the transient TPM of the geometric distribution $Geom(\frac{1}{3})$ ($DPH({\bf B})=Geom(\frac{1}{3})$),
i.e. ${\bf A}$ and ${\bf B}$ are from Example \ref{dphminmax.exm}.

Let $G=\overline{(\{a\},\partial_1^{\bf A})^1}\cho (\{b\},\partial_2^{\bf B})^1$ and $G'=(\{a\},\partial_1^{\bf
A})^1\cho\overline{(\{b\},\partial_2^{\bf B})^1}$. Note that $(\{a\},\partial_1^{\bf A})^1$ is a definite execution,
since ${\cal A}_1=1$, and $(\{b\},\partial_2^{\bf B})^1$ is a probabilistic execution, since ${\cal B}_1=\frac{1}{3}$.
Then $G\approx G'$, since $G\Leftarrow G''\Rightarrow G'$ for $G''=\overline{(\{a\},\partial_1^{\bf A})^1\cho
(\{b\},\partial_2^{\bf B})^1}$, but $Can(G)=Now(G)=\{\{(\{a\},\partial_1^{\bf A})^1\}\}$ and
$Can(G')=Now(G')=\{\{(\{b\},\partial_2^{\bf B})^1\}\}$. We have $dtang(G)$, but $ptang(G')$. To make the action rules
correct under structural equivalence, the executions like $(\{b\},\partial_2^{\bf B})^1$ from $G'$ must be disabled
using preconditions in the action rules, since definite executions have a priority over probabilistic ones, hence, the
choices between them are always resolved in favour of the former.

Let $H=\overline{(\{a\},\partial_1^{\bf A})^1}\|\overline{(\{b\},\partial_2^{\bf B})^1}$. Then
$Can(H)=\{\{(\{a\},\partial_1^{\bf A})^1\},\{(\{b\},\partial_2^{\bf B})^1\},\{\{a\},\partial_1^{\bf A})^1,
(\{b\},\partial_2^{\bf B})^1\}\}$ and $Now(H)=\{\{(\{a\},\partial_1^{\bf A})^1\}\}$. We have $dtang(H)$. To make the
action rules correct under structural equivalence, the executions like $(\{b\},\partial_2^{\bf B})^1$ (and all
multisets of executions including it) from $H$ must be disabled using preconditions in the action rules, since definite
executions have a priority over probabilistic ones, hence, the former are always accomplished first.

Let $K=\overline{(\{a\},\partial_1^{\bf A})^1}\cho (\{b\},\partial_2^\emptyset )$ and $K'=(\{a\},\partial_1^{\bf
A})^1\cho\overline{(\{b\},\partial_2^\emptyset )}$. Note that $(\{b\},\partial_2^\emptyset )$ is a vanishing execution.
Then $K\approx K'$, since $K\Leftarrow K''\Rightarrow K'$ for $K''=\overline{(\{a\},\partial_1^{\bf A})^1\cho
(\{b\},\partial_2^\emptyset )}$, but $Can(K)=Now(K)=\{\{(\{a\},\partial_1^{\bf A})^1\}\}$ and
$Can(K')=Now(K')=\{\{(\{b\},\partial_2^\emptyset )\}\}$. We have $dtang(K)$, but $vanish(K')$. To make the action rules
correct under structural equivalence, the executions like $\{(\{a\},\partial_1^{\bf A})^1\}$ from $K$ must be disabled
using preconditions in the action rules, since immediate executions have a priority over definite ones, hence, the
choices between them are always resolved in favour of the former.
\label{cannowph.exm}
\end{example}

\begin{definition}
Let $G\in RegDynExpr$. The {\em set of all non-empty multisets of executions which can be accomplished from
$[G]_\approx$} is

$$Now([G]_\approx )=
\left\{
\begin{array}{ll}
\bigcup\{Now(H)\mid H\in [G]_\approx ,\ vanish(H)\}, & \exists H\in [G]_\approx\cap OpRegDynExpr\ vanish(H);\\
\bigcup\{Now(H)\mid H\in [G]_\approx ,\ dtang(H)\}, & (\exists H\in [G]_\approx\cap OpRegDynExpr\ dtang(H))\wedge\\
 & (\forall H'\in [G]_\approx\cap OpRegDynExpr\ tang(H'));\\
\bigcup\{Now(H)\mid H\in [G]_\approx\}, & \mbox{otherwise}.
\end{array}
\right.$$

\end{definition}

We write $ptang([G]_\approx )$, if $Now([G]_\approx )\subseteq\nat_{fin}^{\cal PE}\setminus\{\emptyset\}$. In
particular, we have $ptang([G]_\approx )$, if $Now([G]_\approx )=\emptyset$. We write $dtang([G]_\approx )$, if
$\emptyset\neq Now([G]_\approx )\subseteq\nat_{fin}^{\cal DE}\setminus\{\emptyset\}$. We write $tang([G]_\approx )$, if
$ptang([G]_\approx )$ or $dtang([G]_\approx )$. Again, we have $tang([G]_\approx )$, if $Now([G]_\approx )=\emptyset$.
Otherwise, we write $vanish([G]_\approx )$, and in this case $\emptyset\neq Now([G]_\approx )\subseteq\nat_{fin}^{\cal
IL}\setminus\{\emptyset\}$.

\begin{example}
Let $G$ and $G'$ be from Example \ref{cannowph.exm}. Since $G\approx G'$, we get $[G]_\approx =[G']_\approx$. Then
$Now([G]_\approx )=\{\{(\{a\},\partial_1^{\bf A})^1\}\}$ and we have $dtang([G]_\approx )$.

Let $H$ be from Example \ref{cannowph.exm}. Then $Now([H]_\approx )=\{\{(\{a\},\partial_1^{\bf A})^1\}\}$ and we have
$dtang([H]_\approx )$.

Let $K$ and $K'$ be from Example \ref{cannowph.exm}. Since $K\approx K'$, we get $[K]_\approx =[K']_\approx$. Then
$Now([K]_\approx )=\{\{(\{b\},\partial_2^\emptyset )\}\}$ and we have $vanish([K]_\approx )$.
\label{nowstatph.exm}
\end{example}

In Table \ref{actrulesphm.tab}, we define action and empty move rules. In the table, $(\alpha ,\partial_l^{\bf
A})\in{\cal TL},\ {\cal A}_{ij}\ (1\leq i,j\leq ord({\bf A}))$ are the elements of transient TPM ${\bf A},\ {\cal A}_i\
(1\leq i\leq ord({\bf A}))$ are the elements of vector ${\bf a}$ of the transition probabilities towards absorption,
$(\alpha ,\partial_l^\emptyset ),(\beta ,\partial_m^\emptyset )\in{\cal IL}$. Further, $E,F\in RegStatExpr,\ G,H\in
SatOpRegDynExpr,\ \widetilde{G},\widetilde{H}\in RegDynExpr$ and $a\in Act$. Moreover, $\Gamma
,\Delta\in\nat_{fin}^{\cal PE}\setminus\{\emptyset\},\ \Gamma '\in\nat_{fin}^{\cal PE},\ I,J\in\nat_{fin}^{\cal
IL}\setminus\{\emptyset\},\ I'\in\nat_{fin}^{\cal IL},\ V,W\in\nat_{fin}^{\cal DE}\setminus\{\emptyset\},\
V'\in\nat_{fin}^{\cal DE}$ and $\Upsilon\in\nat_{fin}^{{\cal PE}\cup{\cal DE}}\setminus\{\emptyset\},\ \Upsilon
'\in\nat_{fin}^{{\cal PE}\cup{\cal DE}}$.

We use the following abbreviations in the names of the rules from the table: ``{\bf B}'' for ``{\bf B}asis case'',
``{\bf S}'' for ``{\bf S}equence'', ``{\bf C}'' for ``{\bf C}hoice'', ``{\bf P}'' for ``{\bf P}arallel'', ``{\bf L}''
for ``re{\bf L}abeling'', ``{\bf R}'' for ``{\bf R}estriction'', ``{\bf I}'' for ``{\bf I}teraton'' and ``{\bf Sy}''
for ``{\bf Sy}nchronization''. If we cannot merge the action rules with probabilistic, definite, immediate executions
or the empty move rules (with the empty multiset of activities) in one rule for a particular operation then we get the
coupled rules. In such cases, the names of the action rules with probabilistic executions have a suffix `{\bf p}',
those with definite executions have a suffix `{\bf d}', those with immediate executions have a suffix `{\bf i}', and
the empty move rules have a suffix `{\bf e}'. The combination of suffixes denotes the combined rules. For example, the
suffix `${\bf pd}$' marks the rules with probabilistic (`${\bf p}$') and definite (`${\bf d}$') executions. The upper
four rules in the table are the basis case rules {\bf Bpd}, {\bf Bi} and {\bf Be} that describe dynamic expressions
transformations as a result of accomplishing execution of either a single activity or no activities: a timed
multiaction (probabilistically or definitely), an immediate multiaction, or the empty multiset of activities. The other
action and empty move rules describe transformations of dynamic expressions, which are built using particular algebraic
operations. To make presentation more compact, the action rules with double conclusion are combined from two distinct
action rules with the same premises.

\begin{table}[h]
\caption{Action and empty move rules}
\label{actrulesphm.tab}
\begin{center}
$\hspace{-7mm}\small\begin{array}{|ll|}
\hline
\multicolumn{2}{|l|}{\rule{0mm}{5.5mm}
{\bf Bpd}\ \frac{{\cal A}_i>0}{\overline{(\alpha ,\partial_l^{\bf A})^i}\stackrel{\{(\alpha ,\partial_l^{\bf A})^i\}}
{\longrightarrow}\underline{(\alpha ,\partial_l^{\bf A})}}\hspace{42mm}
{\bf Bi}\ \overline{(\alpha ,\partial_l^\emptyset )}\stackrel{\{(\alpha ,\partial_l^\emptyset )\}}{\longrightarrow}
\underline{(\alpha ,\partial_l^\emptyset )}\hspace{42mm}
{\bf Be}\ \frac{{\cal A}_{ij}>0}{\overline{(\alpha ,\partial_l^{\bf A})^i}\stackrel{\emptyset}{\rightarrow}
\overline{(\alpha ,\partial_l^{\bf A})^j}}}\\[3mm]

{\bf S}\ \dfrac{G\stackrel{\Upsilon '}{\rightarrow}\widetilde{G}}{G;E\stackrel{\Upsilon
'}{\rightarrow}\widetilde{G};E,\ E;G\stackrel{\Upsilon '}{\rightarrow}E;\widetilde{G}} &
{\bf C1p}\ \dfrac{G\stackrel{\Gamma}{\rightarrow}\widetilde{G},\ \neg init(G)\vee (init(G)\wedge
ptang([\overline{E}]_\approx ))}
{G\cho E\stackrel{\Gamma}{\rightarrow}\widetilde{G}\cho\!\downharpoonleft\!\!E,\
E\cho G\stackrel{\Gamma}{\rightarrow}\downharpoonleft\!\!E\cho\widetilde{G}}\\[3mm]

{\bf C1d}\ \dfrac{G\stackrel{V}{\rightarrow}\widetilde{G},\ \neg init(G)\vee (init(G)\wedge
tang([\overline{E}]_\approx ))}
{G\cho E\stackrel{V}{\rightarrow}\widetilde{G}\cho\!\downharpoonleft\!\!E,\
E\cho G\stackrel{V}{\rightarrow}\downharpoonleft\!\!E\cho\widetilde{G}} &
{\bf C1i}\ \dfrac{G\stackrel{I}{\rightarrow}\widetilde{G}}{G\cho E\stackrel{I}{\rightarrow}
\widetilde{G}\cho\!\downharpoonleft\!\!E,\
E\cho G\stackrel{I}{\rightarrow}\downharpoonleft\!\!E\cho\widetilde{G}}\\[3mm]

{\bf C1e}\ \dfrac{G\stackrel{\emptyset}{\rightarrow}\widetilde{G},\ \neg init(G)}
{G\cho E\stackrel{\emptyset}{\rightarrow}\widetilde{G}\cho\!\downharpoonleft\!\!E,\
E\cho G\stackrel{\emptyset}{\rightarrow}\downharpoonleft\!\!E\cho\widetilde{G}} &
{\bf C2e}\ \dfrac{G\stackrel{\emptyset}{\rightarrow}\widetilde{G},\ init(G)\wedge ptang([\overline{E}]_\approx ),\
[\overline{E}]_\approx\ni H\stackrel{\emptyset}{\rightarrow}\widetilde{H}}
{G\cho E\stackrel{\emptyset}{\rightarrow}\widetilde{G}\cho\lfloor\widetilde{H}\rfloor,\
E\cho G\stackrel{\emptyset}{\rightarrow}\lfloor\widetilde{H}\rfloor\cho\widetilde{G}}\\[3mm]

{\bf P1p}\ \dfrac{G\stackrel{\Gamma}{\rightarrow}\widetilde{G},\ H\stackrel{\emptyset}{\rightarrow}\widetilde{H},\
ptang([H]_\approx )}
{G\| H\stackrel{\Gamma}{\rightarrow}\widetilde{G}\|\widetilde{H},\
H\| G\stackrel{\Gamma}{\rightarrow}\widetilde{H}\|\widetilde{G}} &
{\bf P1d}\ \dfrac{G\stackrel{V}{\rightarrow}\widetilde{G},\ H\stackrel{\emptyset}{\rightarrow}\widetilde{H},\
ptang([H]_\approx )}
{G\| H\stackrel{V}{\rightarrow}\widetilde{G}\|\widetilde{H},\
H\| G\stackrel{V}{\rightarrow}\widetilde{H}\|\widetilde{G}}\\[3mm]

{\bf P1i}\ \dfrac{G\stackrel{I}{\rightarrow}\widetilde{G}}{G\| H\stackrel{I}{\rightarrow}\widetilde{G}\| H,\
H\| G\stackrel{I}{\rightarrow}H\|\widetilde{G}} &
{\bf P2p}\ \dfrac{G\stackrel{\Gamma}{\rightarrow}\widetilde{G},\ H\stackrel{\Delta}{\rightarrow}\widetilde{H}}
{G\| H\stackrel{\Gamma +\Delta}{\longrightarrow}\widetilde{G}\|\widetilde{H}}\\[3mm]

{\bf P2d}\ \dfrac{G\stackrel{V}{\rightarrow}\widetilde{G},\ H\stackrel{W}{\rightarrow}\widetilde{H}}
{G\| H\stackrel{V+W}{\longrightarrow}\widetilde{G}\|\widetilde{H}} &
{\bf P2i}\ \dfrac{G\stackrel{I}{\rightarrow}\widetilde{G},\ H\stackrel{J}{\rightarrow}\widetilde{H}}
{G\| H\stackrel{I+J}{\longrightarrow}\widetilde{G}\|\widetilde{H}}\\[3mm]

{\bf P2e}\ \dfrac{G\stackrel{\emptyset}{\rightarrow}\widetilde{G},\ H\stackrel{\emptyset}{\rightarrow}\widetilde{H}}
{G\| H\stackrel{\emptyset}{\longrightarrow}\widetilde{G}\|\widetilde{H}} &
{\bf L}\ \dfrac{G\stackrel{\Upsilon '}{\rightarrow}\widetilde{G}}{G[f]\stackrel{f(\Upsilon ')}{\longrightarrow}
\widetilde{G}[f]}\\[3mm]

{\bf R}\ \dfrac{G\stackrel{\Upsilon '}{\rightarrow}\widetilde{G},\ a,\hat{a}\not\in{\cal A}(\Upsilon ')}
{G\rs a\stackrel{\Upsilon}{\rightarrow}\widetilde{G}\rs a} &
{\bf I1}\ \dfrac{G\stackrel{\Upsilon '}{\rightarrow}\widetilde{G}}
{[G*E*F]\stackrel{\Upsilon '}{\rightarrow}[\widetilde{G}*E*F]}\\[3mm]

{\bf I2p}\ \dfrac{G\stackrel{\Gamma}{\rightarrow}\widetilde{G},\ \neg init(G)\vee (init(G)\wedge
ptang([\overline{F}]_\approx ))}
{[E*G*F]\stackrel{\Gamma}{\rightarrow}[E*\widetilde{G}*\!\downharpoonleft\!\!F],\
[E*F*G]\stackrel{\Gamma}{\rightarrow}[E*\!\downharpoonleft\!\!F*\widetilde{G}]} &
{\bf I2d}\ \!\!\dfrac{G\stackrel{V}{\rightarrow}\widetilde{G},\ \neg init(G)\vee (init(G)\wedge
tang([\overline{F}]_\approx ))}
{[E*G*F]\stackrel{V}{\rightarrow}[E*\widetilde{G}*\!\downharpoonleft\!\!F],\
[E*F*G]\stackrel{V}{\rightarrow}[E*\!\downharpoonleft\!\!F*\widetilde{G}]}\\[3mm]

{\bf I2i}\ \!\!\dfrac{G\stackrel{I}{\rightarrow}\widetilde{G}}
{[E*G*F]\stackrel{I}{\rightarrow}[E*\widetilde{G}*\!\downharpoonleft\!\!F],\
[E*F*G]\stackrel{I}{\rightarrow}[E*\!\downharpoonleft\!\!F*\widetilde{G}]} &
{\bf I2e}\ \dfrac{G\stackrel{\emptyset}{\rightarrow}\widetilde{G},\ \neg init(G)}
{[E*G*F]\stackrel{\emptyset}{\rightarrow}[E*\widetilde{G}*\!\downharpoonleft\!\!F],\
[E*F*G]\stackrel{\emptyset}{\rightarrow}[E*\!\downharpoonleft\!\!F*\widetilde{G}]}\\[3mm]

{\bf I3e}\ \dfrac{G\stackrel{\emptyset}{\rightarrow}\widetilde{G},\ init(G)\wedge ptang([\overline{F}]_\approx ),\
[\overline{F}]_\approx\ni H\stackrel{\emptyset}{\rightarrow}\widetilde{H}}
{[E*G*F]\stackrel{\emptyset}{\rightarrow}[E*\widetilde{G}*\lfloor\widetilde{H}\rfloor],\
[E*F*G]\stackrel{\emptyset}{\rightarrow}[E*\lfloor\widetilde{H}\rfloor*\widetilde{G}]} &
{\bf Sy1}\ \dfrac{G\stackrel{\Upsilon '}{\rightarrow}\widetilde{G}}
{G\sy a\stackrel{\Upsilon '}{\rightarrow}\widetilde{G}\sy a}\\[3mm]

{\bf Sy2pd}\ \dfrac{G\sy a\xrightarrow{\Gamma '+\{(\alpha ,\partial_l^{\bf A})^i\}+\{(\beta ,\partial_m^{\bf
B})^j\}}\widetilde{G}\sy a,\ a\!\!\in\!\!\alpha ,\hat{a}\!\!\in\!\!\beta ,\ {\cal A}_i\!\!=\!\!1\!\Leftrightarrow\!
{\cal B}_j\!\!=\!\!1}{G\sy a\xrightarrow{\Gamma '+\{(\alpha\oplus_a\beta ,\partial_{l+m}^{{\bf A}\sqcup
{\bf B}})^{(i,j)}\}}\widetilde{G}\sy a} &
{\bf Sy2i}\ \dfrac{G\sy a\xrightarrow{I'+\{(\alpha ,\partial_l^\emptyset )\}+\{(\beta ,\partial_m^\emptyset
)\}}\widetilde{G}\sy a,\ a\in\alpha ,\ \hat{a}\in\beta}{G\sy a\xrightarrow{I'+\{(\alpha\oplus_a\beta
,\partial_{l+m}^\emptyset )\}}\widetilde{G}\sy a}\\[5mm]
\hline
\end{array}$
\end{center}
\end{table}

In Table \ref{actrulesphm.tab}, rules {\bf Bi}, {\bf S}, {\bf P1i}, {\bf L}, {\bf R}, {\bf I1}, {\bf Sy1} resemble
those of gsPBC, but the former correspond to execution of multisets of activities, not of single activities, as in the
latter, and our rules have simpler preconditions (if any), since all immediate multiactions in dtphPBC have the same
priority level, unlike those of gsPBC. The other rules {\bf Bpd}, {\bf Be}, {\bf C1p}, {\bf C1d}, {\bf C1i}, {\bf C1e},
{\bf C2e}, {\bf P1p}, {\bf P1d}, {\bf P2p}, {\bf P2d}, {\bf P2i}, {\bf P2e}, {\bf I2p}, {\bf I2d}, {\bf I2i}, {\bf
I2e}, {\bf I3e}, {\bf Sy2pd} and {\bf Sy2i} differ from those of gsPBC in various ways, to be explained below.

The preconditions in rules {\bf C1p}, {\bf C1e}, {\bf C2e}, {\bf P1p}, {\bf I2p}, {\bf I2e} and {\bf I3e} are needed to
ensure that (possibly empty) multisets of probabilistic executions are accomplished only from {\em p-tangible}
saturated operative dynamic expressions, such that all dynamic expressions structurally equivalent to them are
p-tangible as well. For example, assuming that $ptang([G]_\approx )$ in rule {\bf C1p}, if $init(G)$ then
$G\approx\overline{F}$ for some static expression $F$ and $G\cho E\approx\overline{F}\cho E\approx\overline{F\cho
E}\approx F\cho\overline{E}$. Hence, it should be guaranteed that $ptang([F\cho\overline{E}]_\approx )$, which holds
iff $ptang([\overline{E}]_\approx )$. The case $E\cho G$ is treated similarly. Assuming that $ptang([G]_\approx )$ in
rule {\bf P1p}, it should be guaranteed that $ptang([G\| H]_\approx )$ and $ptang([H\| G]_\approx )$, which holds iff
$ptang([H]_\approx )$. The precondition in rule {\bf I2p} is analogous to that in rule {\bf C1p}. Rules {\bf C1e}, {\bf
C2e}, {\bf I2e} and {\bf I3e} differ from {\bf C1p} and {\bf I2p} in that in the former four rules the empty multiset
of activities is executed from $G$ and the premises of rules {\bf C1e} and {\bf I2e} contain only the first parts of
the conditions from those of rules {\bf C1p} and {\bf I2p}, respectively. The second parts of those conditions are
included in the premises of rules {\bf C2e} and {\bf I3e}, which also have different conclusions signifying that at the
time step, the phases are changed in the {\em both} alternative subexpressions.

Analogously, the preconditions in rules {\bf C1d} and {\bf I2d} are needed to ensure that non-empty multisets of
definite executions are accomplished only from {\em d-tangible} saturated operative dynamic expressions, such that all
dynamic expressions structurally equivalent to them are tangible. This requirement (about tangible expressions) means
that only (possibly empty) multisets of probabilistic executions or non-empty multisets of definite executions, and no
immediate executions, can be accomplished from the subprocess that is composed {\em alternatively (via choice)} with the
subprocess $G$. Hence, the multiset $V$ of definite executions, accomplished from $G$, can also be accomplished from the
composition of $G$ and that alternative subprocess, since immediate executions cannot occur from the latter. Otherwise,
it would prevent the accomplishing of $V$ from $G$ in the composite process, by disregarding the alternative choice of
the branch specified by $G$, due to the zero delays and priority (captured by all action rules) of immediate executions
over all other execution types.

The precondition in rule {\bf P1d} is an exception from the above. It also ensures that non-empty multisets of definite
executions are accomplished only from {\em d-tangible} saturated operative dynamic expressions, such that all dynamic
expressions structurally equivalent to them are tangible, but all the expressions structurally equivalent to $H$
specifying parallel with $G$ subprocess should be p-tangible. This stricter requirement (about p-tangible, instead of
just tangible, expressions) means that only (possibly empty) multisets of probabilistic executions, and no immediate or
definite executions, can be accomplished from the subprocess $H$ that is composed {\em concurrently (in parallel)} with
the subprocess $G$. Hence, the multiset $V$ of definite executions, accomplished from $G$, is also a maximal (by the
inclusion relation) multiset that can be accomplished from the parallel composition of $G$ and $H$. The reason is that
only the change of the positive phases of delays (like in rule {\bf Be}) is actually possible in $H$ while
accomplishing $V$ from $G$, due to priority (captured by all action rules) of definite executions over probabilistic
ones. Moreover, the empty move from $H$ is explicitly given in the premise of rule {\bf P1d}. Thus, taking the rule
precondition $ptang([H]_\approx )$ instead of $tang([H]_\approx )$ preserves maximality of the steps consisting of
definite executions while applying parallel composition.

In rules {\bf P1p} and {\bf P1d}, the execution of the empty multiset of activities from the subprocess $H$ that is
composed in parallel with $G$ (from which probabilistic or definite executions are respectively accomplished at the
next time tick) is used to maintain the time progress uniformity in the composite expression, as well as to indicate in
the premises of those rules the resulting subprocess $\widetilde{H}$, needed to build the resulting composite process
in the conclusions of those rules. Although rules {\bf P1p} and {\bf P1d} can be merged, we have not done it, aiming to
emphasize the exceptional precondition $ptang([H]_\approx )$ (instead of $tang([H]_\approx )$) in rule {\bf P1d}.

In rules {\bf C1p}, {\bf C1d}, {\bf C1i} and {\bf C1e}, the delay phases discarding $\downharpoonleft\!\!E$, applied to
the static expression $E$ that is composed via choice with $G$ (from which probabilistic, definite or immediate
executions or the empty multiset of activities are respectively accomplished) signifies that the delay phases of the
non-chosen subexpression (branch) become irrelevant in the composite expression and thus may be removed. Analogously,
in rules {\bf I2p}, {\bf I2d}, {\bf I2i} and {\bf I2e}, the delay phases discarding $\downharpoonleft\!\!F$ is applied
to the static expression $F$ that is an alternative to $G$ (from which activities or the empty multiset of them are
executed), since the choice is always made between the body and termination subexpressions of the composite iteration
expression (between the second and third arguments of iteration).

In the conclusion of of rule {\bf C2e}, the lines discarding $\lfloor\widetilde{H}\rfloor$ is applied to the dynamic
expression $\widetilde{H}$ that is composed via the choice operation with $\widetilde{G}$, resulted from execution of
the empty multiset of activities from $G$. This technique describes the situation when the delay phases of the
subexpression (second alternative branch) $E$, which can be potentially chosen instead of $G$ (due to the overline
shift to the second alternative branch in the structurally equivalent to $G\cho E$ expression $\overline{F}\cho E$) by
proceeding with the subexpression $H\approx\overline{E}$, are changed at the time step in the {\em both} alternative
subexpressions $G$ and $E$. If $H\stackrel{\emptyset}{\rightarrow}\widetilde{H}$ then these alternatives are
respectively transformed into the subexpressions $\widetilde{G}$ and $\lfloor\widetilde{H}\rfloor$ of the resulting
composite expression. Rule {\bf I3e} is explained similarly, by taking $F$ instead of $E$.

Rule {\bf Bpd} differs from the more standard one {\bf Bi} that resembles rule {\bf B} in gsPBC. The reason is that
{\bf Bpd} has the premise with the conditions on the probability of changing the current {\em positive phase} to the
{\em zero phase}. If the probability is less than $1$ then a probabilistic execution is accomplished while such a
change. If the probability is equal to $1$ then a definite execution is accomplished at that change.

Rule {\bf Be} corresponds to one discrete time unit delay (passage of one unit of time) while executing no activities
and therefore it has no analogues among the rules of gsPBC that adapts the continuous time model. This rule describes
the execution of the empty multiset of activities while making a transition between two {\em positive phases} of the
timed multiaction delay. Since with such a change of phases the zero phase is not entered, no activities are executed
(no execution is accomplished).

Rules {\bf P2p}, {\bf P2d}, {\bf P2i} and {\bf P2e} have no similar rules in gsPBC, since interleaving semantics of
that algebra allows no simultaneous execution of activities or their empty multisets. On the other hand, {\bf P2p},
{\bf P2d}, {\bf P2i} and {\bf P2e} have in PBC the analogous rule {\bf PAR} that is used to construct step semantics of
the calculus, but the former four rules correspond to execution of (possibly empty) multisets of activities, unlike
that of multisets of multiactions in the latter rule. Rules {\bf P2p}, {\bf P2d}, {\bf P2i} and {\bf P2e} cannot be
merged, since otherwise simultaneous accomplishing of different types of executions and execution of the empty multiset
of activities would be allowed.

Rules {\bf Sy2pd} and {\bf Sy2i} differ from the corresponding synchronization rules in gsPBC, since the probability or
the weight of synchronization in the former rules and the rate or the weight of synchronization in the latter rules are
calculated in two distinct ways. Rules {\bf Sy2pd} and {\bf Sy2i} cannot be merged, since otherwise synchronization
among probabilistic, definite and immediate executions would be allowed.

Rule {\bf Sy2pd} establishes that the synchronization of two probabilistic executions is made by taking the product of
their probabilities $\rho ={\cal A}_i$ and $\chi ={\cal B}_j$, since we are considering that both must occur for the
synchronization to happen. This corresponds, in some sense, to the probability of the independent event intersection,
but the real situation is more complex, since these probabilistic executions can also be accomplished in parallel.
Nevertheless, when scoping (the combined operation consisting of synchronization followed by restriction over the same
action \cite{BDK01}) is applied over a parallel accomplishment, we get as final result just the simple product of the
probabilities, since no normalization is needed there. Multiplication is an associative and commutative binary
operation that is distributive over addition, i.e. it fulfills all practical conditions imposed on the synchronization
operator in \cite{Hil94b}. Further, if both arguments of multiplication are from $(0;1)$ then the result belongs to the
same interval, hence, multiplication naturally maintains probabilistic compositionality in our model. Our approach is
similar to the multiplication of rates of the synchronized actions in MTIPP \cite{HR94} in the case when the rates are
less than $1$. Moreover, for the probabilities $\rho$ and $\chi$ of two probabilistic executions to be synchronized we
have $\rho\cdot\chi <\min\{\rho ,\chi\}$, i.e. multiplication meets the performance requirement stating that the
probability of the resulting synchronized probabilistic execution should be less than the probabilities of the two ones
to be synchronized. While performance evaluation, it is usually supposed that the execution of two components together
require more system resources and time than the execution of each single one. This resembles the {\em bounded capacity}
assumption from \cite{Hil94b}. Thus, multiplication is easy to handle with and it satisfies the algebraic,
probabilistic, time and performance requirements. Therefore, we have chosen the product of the probabilities for the
synchronization. See also \cite{BKLL95,BrHe01,Bri03,SS24} for a discussion about binary operations producing the rates
of synchronization in the continuous time setting.

In rules {\bf Sy2pd} and {\bf Sy2i}, we sum the weights $l$ and $m$ of two synchronized definite (immediate)
executions, since the weights can be interpreted as the rewards \cite{Ros96}, thus, we collect the rewards. Moreover,
we express that the synchronized accomplishment of definite (immediate) executions has more importance than
accomplishing each single one. The weights of definite and immediate executions can also be seen as bonus rewards
associated with transitions \cite{BBr01}. The rewards are summed during synchronized accomplishment of definite
(immediate) executions, since in that case all the synchronized executions can be seen as participated in the
accomplishment. We prefer to collect more rewards, thus, the transitions providing greater rewards will have a
preference and they will be accomplished with a greater probability. In particular, since accomplishment of immediate
executions takes no time, we prefer to collect in a step (parallel accomplishment) as many synchronized immediate
executions as possible to get more significant progress in behaviour. Under behavioural progress we understand an
advance in accomplishing executions, which does not always imply a progress in time, as in the case when the executions
are immediate ones. This aspect is used while evaluating performance via analysis of the embedded discrete time Markov
chains (EDTMCs) of expressions. Since every state change in EDTMC takes one unit of (its local) time, greater advance
in operation of the EDTMC allows one to calculate quicker many performance indices. As for definite executions, only
the maximal multisets of them, accomplishable from a state, occur with a time tick. The reason is that each definite
execution has a probability $1$ to occur in the next moment, when there exist no conflicting definite executions.
Hence, all definite executions that are accomplishable together from a state must participate in a step from that
state. Since there may exist different such maximal multisets of definite executions, a probabilistic choice among all
possible steps is made, imposed by the weights of those executions. Thus, the steps of definite executions always
produce maximal overall weights, but they are mainly used to calculate the probabilities of alternative maximal steps
rather than the cumulative bonus rewards.

We do not have self-synchronization, i.e. synchronization of an activity with itself, since all the (enumerated)
activities executed together are considered to be different. This allows us to avoid rather cumbersome and unexpected
behaviour, as well as many technical difficulties \cite{BDK01}.

Notice that the phase indicators of all enabled timed multiactions that lose their enabledness when a state change
occurs become inactive (turned off) and their values become irrelevant while the phase indicators of all those
preserving their enabledness stay active with their (possibly new, due to the phase change) values. Hence, we adapt the
{\em enabling memory} policy \cite{MBCDF95,AHR00,Bal01,Bal07} when the process states are changed and the enabledness
of deterministic multiactions is possibly modified (remember that immediate multiactions may be seen as those with the
phase indicators displaying a single value $0$, so we do not need to store their values). Then the phases of delays of
timed multiactions are taken as the enabling memory variables.

Similar to \cite{Kou00}, we are mainly interested in the dynamic expressions, inferred by applying the inaction rules
(also in the reverse direction) and action rules from the overlined static expressions, such that no superscribed (i.e.
superscribed with the delay phases) timed multiaction is a subexpression of them. The reason is to ensure that time
proceeds uniformly and only enabled timed multiactions are superscribed. We call such dynamic expressions~reachable.

\begin{definition}
A dynamic expression $G$ is {\em reachable}, if there exists a static expression $E$ without delay phase superscripts,
such that $\overline{E}\approx G$ or $\overline{E}\approx G_0\stackrel{\Upsilon_1}{\rightarrow}H_1\approx
G_1\stackrel{\Upsilon_2}{\rightarrow}\ldots\stackrel{\Upsilon_n}{\rightarrow}H_n\approx G$ for some $\Upsilon_1,\ldots
,\Upsilon_n\in\nat_{fin}^{\cal PHL}$.
\end{definition}

\begin{example}
Let ${\bf A}$ be the transient TPM of the (discrete) uniform distribution $Unif(1,3)$ ($DPH({\bf A})=Unif(1,3)$) and
${\bf B}$ be the transient TPM of the geometric distribution $Geom(\frac{1}{3})$ ($DPH({\bf B})=Geom(\frac{1}{3})$),
i.e. ${\bf A}$ and ${\bf B}$ are from Example \ref{dphminmax.exm}.

A dynamic expression $G=\overline{(\{a\},\partial_1^{\bf A})^3}\cho\underline{\{b\},\partial_2^{\bf B})}$ is
``illegal'' and $G'=\overline{(\{a\},\partial_1^{\bf A})^2}\cho\underline{(\{b\},\partial_2^{\bf B})}$ is ``le\-gal'',
since $G'$ is obtained after one phase change from the overlined static expression without phase value superscripts
$\overline{E}=\overline{(\{a\},\partial_1^{\bf A})\cho (\{b\},\partial_2^{\bf B})}$.

A dynamic expression $H=\overline{(\{a\},\partial_1^{\bf A})};(\{b\},\partial_2^{\bf B})^1$ is ``illegal'' and
$H'=\overline{(\{a\},\partial_1^{\bf A})};(\{b\},\partial_2^{\bf B})$ is ``legal'', since the timed mul\-ti\-ac\-ti\-on
$(\{b\},\partial_2^{\bf B})$ is not enabled in $[H]_\approx$ and its phase indicator cannot start before the executions
of the timed multiaction $(\{a\},\partial_1^{\bf A})$ are accomplished.
\label{reachabph.exm}
\end{example}

Enabledness of the superscribed timed multiactions is considered in the following proposition.

\begin{proposition}
Let $G$ be a reachable dynamic expression. Then only timed multiactions from\\
$EnaTime([G]_\approx )$ are superscribed in $G$.
\label{enawmstmph.pro}
\end{proposition}
{\em Proof.} By the definition of reachability, there exists $E\in StatExpr$ without superscribed timed multiactions, such
that $G$ is derived from $\overline{E}$ by applying the (possibly reversed) inaction rules and action rules:
$\overline{E}\approx G$ or $\overline{E}\approx G_0\stackrel{\Upsilon_1}{\rightarrow}H_1\approx
G_1\stackrel{\Upsilon_2}{\rightarrow}\ldots\stackrel{\Upsilon_n}{\rightarrow}H_n\approx G_n=G$ for some
$\Upsilon_1,\ldots ,\Upsilon_n\in\nat_{fin}^{\cal PHL}$.

In that derivation, only the first {\em inaction} rule can add delay phase superscripts to the timed multiactions from
${\cal TL}(E)={\cal TL}(\overline{E})={\cal TL}(G_i)={\cal TL}(H_j)={\cal TL}(G)\ (0\leq i\leq n,\ 1\leq j\leq n)$ that
are overlined. The other inaction rules (also reversed) can just ``shift'' the up\-per bars from / to those
superscribed timed multiactions while preserving the enabledness of all timed mul\-ti\-ac\-ti\-ons from ${\cal TL}(G)$.
Thus, just the timed multiactions from $EnaTime([G_i]_\approx )$ {\em become} superscribed in $G_i\ (0\leq i\leq n)$
(in the subexpressions like $\overline{(\alpha ,\partial_l^\theta )^\theta}$ or $(\alpha ,\partial_l^\theta )^\theta$,
corresponding to the newly enabled timed~multiactions).

Further, in the derivation, the {\em action} rules cannot add new delay phase superscripts to the timed multiactions
from ${\cal TL}(G)$. Instead, the action rules can only change the phases (rule {\bf Be}) or discard them by making
such timed multiactions non-enabled (disabled). Such ``dis\-ab\-ling'' action rules correspond either to the executing
an overlined superscribed timed mul\-ti\-ac\-ti\-on (rule {\bf Bpd}) or to the choice of some alternative process
branch (rules {\bf C1p}, {\bf C1d}, {\bf C1i}, {\bf I2p}, {\bf I2d} and {\bf I2i}). In the both cases, all the disabled
timed multiactions loose their delay phase superscripts. Thus, only the timed mul\-ti\-ac\-ti\-ons from
$EnaTime([H_j]_\approx )=EnaTime([G_j]_\approx )$ {\em remain} superscribed in $H_j$ and {\em become} superscribed in
$G_j$, ob\-tai\-ned from $H_j\ (1\leq j\leq n)$ by applying the (possibly reversed) inaction~rules.

Hence, $\overline{E}$ does not contain superscribed timed multiactions and at the end of the derivation of $G$ from it, only
the timed multiactions from $EnaTime([G_n]_\approx )$ {\em remain} superscribed in $H_n$ and {\em become} superscribed in
$G_n=G$. Therefore, only timed mul\-ti\-ac\-ti\-ons from $EnaTime([G]_\approx )$ are superscribed in $G$. \hfill $\eop$

In Table \ref{rulescompphm.tab}, inaction rules, action rules (with probabilistic or definite, or immediate executions)
and empty move rules are compared according to the three questions about their application: whether it changes the
current state, whether it leads to a time progress, and whether it results in execution of some activities. Positive
answers to the questions are denoted by the plus sign while negative ones are specified by the minus sign. If both
positive and negative answers can be given to some of the questions in different cases then the plus-minus sign is
written. Notice that the process states are considered up to structural equivalence of the corresponding expressions,
and time progress is not regarded as a state change.

\begin{table}
\caption{Comparison of inaction, action and empty move rules}
\label{rulescompphm.tab}
\begin{center}
\begin{tabular}{|c||c|c|c|}
\hline
Rules & State change & Time progress & Activities execution\\
\hline\hline
Inaction rules & $-$ & $-$ & $-$\\
\hline
Action rules & $\pm$ & $+$ & $+$ \\
(probabilistic or definite executions) & & & \\
\hline
Action rules & $\pm$ & $-$ & $+$ \\
(immediate executions) & & & \\
\hline
Empty move rule & $\pm$ & $+$ & $-$ \\
\hline
\end{tabular}
\end{center}
\end{table}

\subsection{Transition systems}

We now construct labeled probabilistic transition systems associated with dynamic expressions. The transition systems
are used to define the operational semantics of dynamic expressions.

In addition to the structural equivalence laws for dynamic expressions, obtained by reflexive, symmetric and transitive
closure of the inaction rules, we use the laws, inherited from the context of PBC \cite{BDK01}. Those extra laws
include associativity of sequential, choice and parallel compositions; commutativity of choice and parallel ones;
distributivity of basic relabeling over sequential, choice and parallel ones; collapse of successive relabelings;
neutrality of the identical relabeling.

We need in the next auxiliary notion. Let $E$ be a static expression. The {\em alphabet} of $E$ is defined in a
structural way. Let $(\alpha ,\kappa )\in{\cal PHL},\ (\alpha ,\partial_l^{\bf A})\in{\cal TL},\ i\in\{1,\ldots
,ord({\bf A})\},\ E,F,K\in RegStatExpr$ and $a\in Act$. Then
\begin{enumerate}

\item ${\cal A}((\alpha ,\kappa ))={\cal A}((\alpha ,\partial_l^{\bf A})^i)={\cal A}(\alpha )$;

\item ${\cal A}(E\circ F)={\cal A}(E)\cup{\cal A}(F)\ (\circ\in\{;,\cho ,\|\})$;

\item ${\cal A}(E[f])=f({\cal A}(E))$;

\item ${\cal A}(E\rs a)={\cal A}(E)\setminus\{a,\hat{a}\}$;

\item ${\cal A}(E\sy a)={\cal A}(E)$;

\item ${\cal A}([E*F*K])={\cal A}(E)\cup{\cal A}(F)\cup{\cal A}(K)$.

\end{enumerate}
Let $G$ be a dynamic expression. The {\em alphabet} of $G$ is defined as ${\cal A}(G)={\cal A}(\lfloor G\rfloor )$.

In Table \ref{selawsrssyphm.tab}, we introduce additional structural equivalence laws for dynamic expressions with
restriction and synchronization. Most of the laws have been borrowed from PBC \cite{BDK01}, excepting the last three
laws being new and applied when the restricted action and its coaction are absent or the synchronization pair is
missing. In this table, $E,F\in RegStatExpr,\ G,H\in RegDynExpr,\ a,b\in Act$ and
$\circ\in\{\!\!\rs\!\!,\!\!\sy\!\!\}$.

\begin{table}[h]
\caption{Structural equivalence laws for dynamic expressions with restriction and synchronization}
\label{selawsrssyphm.tab}
\begin{center}
$\begin{array}{|ll|}
\hline
\rule{0mm}{4mm}
G\circ a\circ b\approx G\circ b\circ a & G\circ a\circ a\approx G\circ a\\[1mm]
G\circ\hat{a}\approx G\circ a & (G;E)\circ a\approx (G\circ a);(E\circ a)\\[1mm]
(E;G)\circ a\approx (E\circ a);(G\circ a) & (G\cho E)\circ a\approx (G\circ a)\cho (E\circ a)\\[1mm]
(E\cho G)\circ a\approx (E\circ a)\cho (G\circ a) & (G\| H)\sy a\approx ((G\sy a)\| (H\sy a))\sy a\\[1mm]
(G\| H)\rs a\approx (G\rs a)\| (H\rs a) & G\sy a\rs b\approx G\rs b\sy a\ (a\not\in\{b,\hat{b}\})\\[1mm]
(G\circ a)[f]\approx (G[f])\circ f(a) & [G*E*F]\circ a\approx [(G\circ a)*(E\circ a)*(F\circ a)]\\[1mm]
[E*G*F]\circ a\approx [(E\circ a)*(G\circ a)*(F\circ a)] &
[E*F*G]\circ a\approx [(E\circ a)*(F\circ a)*(G\circ a)]\\[1mm]
G\rs a\approx G,\ a,\hat{a}\not\in{\cal A}(G) &
G\sy a\approx G,\ (a\not\in{\cal A}(G))\vee (\hat{a}\not\in{\cal A}(G))\\[1mm]
\multicolumn{2}{|l|}{(G\| H)\sy a\approx (G\sy a)\| (H\sy a),\ (a,\hat{a}\not\in{\cal A}(G))\vee
(a,\hat{a}\not\in{\cal A}(H))\vee (a\not\in{\cal A}(G)\cup {\cal A}(H))\vee
(\hat{a}\not\in{\cal A}(G)\cup{\cal A}(H))}\\
\hline
\end{array}$
\end{center}
\end{table}

The goal of the extra structural equivalence laws above is to move the restriction and synchronization operators so to
act on the dynamic expressions of minimal length (ideally, so that the restrictions affect single activities, and the
synchronizations affect only the composed activities with the synchronization pair).

Let $G$ be a dynamic expression and $s=[G]_\approx$. The set of {\em all multisets of executions in $s$} is defined as
$Exec(s)=\{\Upsilon\mid\exists H\in s\ \exists\widetilde{H}\ H\stackrel{\Upsilon}{\rightarrow}\widetilde{H}\}$. Here
$H\stackrel{\Upsilon}{\rightarrow}\widetilde{H}$ is an inference by the rules from Table \ref{actrulesphm.tab}.

It can be proved by induction on the structure of expressions that $\Upsilon\in Exec(s)\setminus\{\emptyset\}$ implies
$\exists H\in s\ \Upsilon\in Now(H)$. The reverse statement does not hold in general, since the preconditions in the
action rules disable the executions with the lower-priority types from every $H\in s$.

\begin{example}
Let $K,K'$ be from Example \ref{cannowph.exm} and $s=[K]_\approx =[K']_\approx$. We have
$Now(K)=\{\{(\{a\},\partial_1^{\bf A})^1\}\}$ and $Now(K')=\{\{(\{b\},\partial_2^\emptyset )\}\}$. Since no rule can be
applied to $K$ while the only rules applicable to $K'$ are {\bf C1i} and {\bf Bi}, we get
$Exec(s)=\{\{(\{b\},\partial_2^\emptyset )\}\}$. Then, for $K\in s$ and $\Upsilon =\{(\{a\},\partial_1^{\bf A})^1\}\in
Now(K)$, we obtain $\Upsilon\not\in Exec(s)$.
\label{nowexph.exm}
\end{example}

The state $s$ is {\em p-tangible (probabilistically tangible)}, denoted by $ptang(s)$, if
$Exec(s)\subseteq\nat_{fin}^{\cal PE}$. For a p-tangible state $s$ we always have $\emptyset\in Exec(s)$ by applying
rules {\bf Bpd} and {\bf Be} for each timed multiaction $(\alpha ,\partial_l^{\bf A})$ from $EnaTime(s)$, superscribed
with delay phase $i\ (1<i<ord({\bf A}))$ and overlined, since in those rules ${\cal A}_i<1$ implies ${\cal A}_{ij}>0$
for some phase $j\ (1<j<ord({\bf A}))$. Hence, we may have $Exec(s)=\{\emptyset\}$, when no probabilistic executions
are accomplished and ${\cal A}_i=0$, which implies $\sum_{j=1}^{ord({\bf A})}{\cal A}_{ij}=1$. The state $s$ is {\em
d-tangible (definitely tangible)}, denoted by $dtang(s)$, if $Exec(s)\subseteq\nat_{fin}^{\cal
DE}\setminus\{\emptyset\}$. The state $s$ is {\em tangible}, denoted by $tang(s)$, if $ptang(s)$ or $dtang(s)$, i.e.
$Exec(s)\subseteq\nat_{fin}^{\cal PE}\cup\nat_{fin}^{\cal DE}$. Again, for a tangible state $s$ we may have
$\emptyset\in Exec(s)$ and $Exec(s)=\{\emptyset\}$. Otherwise, the state $s$ is {\em vanishing}, denoted by
$vanish(s)$, and in this case $Exec(s)\subseteq\nat_{fin}^{\cal IL}\setminus\{\emptyset\}$.

Since for every $H\in s,\ Now(H)$ containing the multisets of executions with the lower-priority types is not included
into $Exec(s)$, and the types of states are determined from the highest-priority types of the possible executions, the
state type definitions based on $Now(H),\ H\in s$, and on $Exec(s)$ are consistent.

Note that if $\Upsilon\in Exec(s)$ and $\Upsilon\in\nat_{fin}^{\cal PE}\cup\nat_{fin}^{\cal IL}$ then by rules {\bf
P2p}, {\bf P2i}, {\bf Sy2pd}, {\bf Sy2i} and definition of $Exec(s)\ \forall\Xi\subseteq\Upsilon ,\ \Xi\neq\emptyset$,
we have $\Xi\in Exec(s)$, i.e. $2^\Upsilon\setminus\{\emptyset\}\subseteq Exec(s)$.

Since the inaction rules only set the initial delay phases or distribute and move upper and lower bars along the syntax
of dynamic expressions, all $H\in s$ have the same phase-free underlying static expression $F$. Process expressions
always have a finite length, hence, the number of all (enumerated) activities and the number of all operations in the
syntax of $F$ are finite as well. The action rules {\bf Sy2pd} and {\bf Sy2i} are the only ones that generate new
activities. They result from the handshake synchronization of actions and their conjugates belonging to the multiaction
parts of the first and second constituent execution, respectively. Since we have a finite number of operators $\sy$ in
$F$ and all the multiaction parts of the activities are finite multisets, the number of the new synchronized executions
is also finite. The action rules contribute to $Exec(s)$ (in addition to the empty set) only the sets consisting of
both activities from $F$ being superscribed and the new activities, produced by {\bf Sy2pd} and {\bf Sy2i}, being
superscribed. Since we have a finite number $n$ of all such activities, we get $|Exec(s)|\leq 2^n<\infty$. Thus,
summation and multiplication by elements from the finite set $Exec(s)$ are well-defined. Similar reasoning can be used
to demonstrate that for all dynamic expressions $H$ (not just for those from $s$), $Now(H)$ is a finite~set.

\begin{definition}
The {\em derivation set} of a dynamic expression $G$, denoted by $DR(G)$, is the minimal set such that
\begin{itemize}

\item $[G]_\approx\in DR(G)$;

\item if $[H]_\approx\in DR(G)$ and $\exists\Upsilon\ H\stackrel{\Upsilon}{\rightarrow}\widetilde{H}$ then
$[\widetilde{H}]_\approx\in DR(G)$.

\end{itemize}
\end{definition}

The set of {\em all p-tangible states from $DR(G)$} is denoted by $DR_{PT}(G)$, and the set of {\em all d-tangible
states from $DR(G)$} is denoted by $DR_{DT}(G)$. The set of {\em all tangible states from $DR(G)$} is denoted by
$DR_T(G)=DR_{PT}(G)\cup DR_{DT}(G)$. The set of {\em all vanishing states from $DR(G)$} is denoted by $DR_V(G)$.
Obviously, $DR(G)=DR_T(G)\cup DR_V(G)=DR_{PT}(G)\cup DR_{DT}(G)\cup DR_V(G)$.

Let now $G$ be a dynamic expression and $s,\tilde{s}\in DR(G)$.

Let $\Upsilon\in Exec(s)\setminus\{\emptyset\}$ and for some $H\in s,\ \widetilde{H}\in\tilde{s}$ we have
$H\stackrel{\Upsilon}{\rightarrow}\widetilde{H}$. The {\em probability that the multiset of probabilistic executions
$\Upsilon$ is ready to accomplish in $s$ and move to $\tilde{s}$} or the {\em weight of the multiset of definite
(immediate) executions $\Upsilon$ which is ready to accomplish in $s$ and move to $\tilde{s}$} is

$$PF(\Upsilon ,s,\tilde{s})\!\!=\!\!
\left\{\!\!\!\!
\begin{array}{ll}
\displaystyle\prod_{(\alpha ,\partial_l^{\bf A})^i\in\Upsilon}{\cal A}_i\cdot
\displaystyle\prod_{\{\{(\beta ,\partial_m^{\cal B})^j\}\in Exec(s)\mid
(\beta ,\partial_m^{\cal B})^j\not\in\Upsilon ,\
\Upsilon +\{(\beta ,\partial_m^{\cal B})^j\}\not\in Exec(s)\}}(1-{\cal B}_j)\cdot & \\
\displaystyle\prod_{\{\{(\beta ,\partial_m^{\cal B})^j\}\in Exec(s)\mid
(\beta ,\partial_m^{\cal B})^j\not\in\Upsilon ,\
\Upsilon +\{(\beta ,\partial_m^{\cal B})^j\}\in Exec(s),\ \exists\widetilde{H}\in\tilde{s}\
\overline{(\beta ,\partial_m^{\bf B})^k}\in Sub(\widetilde{H})\}}{\cal B}_{jk}\cdot & \\
\displaystyle\prod_{\{\{(\beta ,\partial_m^{\cal B})^j\}\not\in Exec(s)\mid\exists H\in s\
\overline{(\beta ,\partial_m^{\bf B})^j}\in Sub(H),\ \exists\widetilde{H}\in\tilde{s}\
\overline{(\beta ,\partial_m^{\bf B})^k}\in Sub(\widetilde{H})\}}\!\!{\cal B}_{jk},
& \!\!\!s\!\in\!DR_{PT}(G);\\
\displaystyle\sum_{(\alpha ,\partial_l^{\bf A})^i\in\Upsilon}l, & \!\!\!s\!\in\!DR_{DT}(G);\\
\displaystyle\sum_{(\alpha ,\partial_l^\emptyset )\in\Upsilon}l, & \!\!\!s\!\in\!DR_V(G).
\end{array}
\right.$$

Let $\Upsilon =\emptyset ,\ s\in DR_{PT}(G)$ and for some $H\in s,\ \widetilde{H}\in\tilde{s}$ we have
$H\stackrel{\emptyset}{\rightarrow}\widetilde{H}$. We define

$$PF(\emptyset ,s,\tilde{s})=
\displaystyle\prod_{\{(\beta ,\partial_m^{\cal B})\in EnaTime(s)\mid\exists H\in s\
\overline{(\beta ,\partial_m^{\bf B})^j}\in Sub(H),\ \exists\widetilde{H}\in\tilde{s}\
\overline{(\beta ,\partial_m^{\bf B})^k}\in Sub(\widetilde{H})\}}{\cal B}_{jk}.$$

If $s\in DR_{PT}(G)$ and $Exec(s)\neq\{\emptyset\}$ then $PF(\Upsilon ,s,\tilde{s})$ can be interpreted as a {\em
joint} probability of independent events (in a probability sense, i.e. the probability of intersection of these events
is equal to the product of their probabilities). Each such an event consists in the decision of a particular enabled
timed multiaction to accomplish at the next time step either the corresponding probabilistic execution or the empty
multiset of activities. Thus, every enabled timed multiaction decides probabilistically (using the absorbing DTMC
defining its delay) and independently (from others), if it wants to accomplish the derived probabilistic execution or
the empty multiset of activities in $s$ and move to $\tilde{s}$. If $\Upsilon$ is a multiset of all probabilistic
executions whose accomplishments have been decided in this way and $\Upsilon\in Exec(s)$ then $\Upsilon$ is ready for
accomplishing in $s$ and moving to $\tilde{s}$. The multiplication in the definition is used because it reflects the
probability of the independent events intersection. The {\em first product} for the case $s\in DR_{PT}(G)$ in the above
definition of $PF(\Upsilon ,s,\tilde{s})$ involves the {\em positive} probabilities of all probabilistic executions
{\em included} in $\Upsilon$ being the probabilities to change the current phase of the corresponding enabled timed
mutiaction to the zero phase by accomplishing a probabilistic execution from $\Upsilon$. The {\em second product}
involves the {\em negative} probabilities of all possible in the current phase probabilistic executions {\em
non-included} in $\Upsilon$ and {\em conflicting} with those from $\Upsilon$, i.e. those executions are {\em not}
possible together with $\Upsilon$. When the current phase of each corresponding enabled timed multiaction is changed to
the zero phase due to accomplishing a probabilistic execution from $\Upsilon$, it prevents accomplishing a conflicting
probabilistic execution {\em not} from $\Upsilon$. The {\em third product} involves the execution probabilities of the
{\em empty multisets} of activities when the possible in the current phase probabilistic executions are {\em
non-included} in $\Upsilon$ and are {\em parallel} with those from $\Upsilon$, i.e. those executions are possible
together with $\Upsilon$. When the current phase of each corresponding enabled timed multiaction is changed to the zero
phase due to accomplishing a probabilistic execution from $\Upsilon$, it does not prevent accomplishing a parallel
probabilistic execution {\em not} from $\Upsilon$. It has been decided not to accomplish that probabilistic execution
(by non-including it in $\Upsilon$), and to execute the empty multiset of activities instead. The {\em fourth product}
involves the execution probabilities of the {\em empty multisets} of activities when the probabilistic executions {\em
cannot be accomplished} in the current phase of the corresponding enabled timed multiactions. Since for all timed
multiactions that can accomplish probabilistic executions in particular phases the probabilities of the positive phase
changes from those phases are greater than $0$, for $s\in DR_{PT}(G)$ we have $PF(\emptyset ,s,\tilde{s})\in (0;1]$.
Hence, we always execute at the next time moment the empty multiset of activities in $s$ and move to $\tilde{s}$ with a
certain positive probability.

If $s\in DR_{DT}(G)\cup DR_V(G)$ then $PF(\Upsilon ,s,\tilde{s})$ could be interpreted as the {\em overall
(cumulative)} weight of the definite (immediate, respectively) executions from $\Upsilon$, i.e. the sum of all their
weights. The summation here is used since the weights can be seen as the rewards which are collected \cite{Ros96}. This
means that parallel accomplishment of the definite (immediate) executions has more importance than accomplishing each
single one. The weights of definite (immediate) executions can also be interpreted as bonus rewards of transitions
\cite{BBr01}. The rewards are summed when definite (immediate) executions are accomplished in parallel, because all of
them participated in the accomplishment. In particular, since immediate executions take no time, we prefer to collect
in a step (parallel accomplishment) as many parallel immediate executions as possible to get more progress in
behaviour. Concerning definite executions, only the maximal multisets of them possible from a state occur in the next
moment. Therefore, the steps of definite executions produce maximal overall weights, which are used to calculate
probabilities of alternative maximal steps rather than the cumulative bonuses. Note that this reasoning is the same as
that used to define the weight of synchronized timed (immediate) executions in rules {\bf Sy2pd} and {\bf Sy2i}.

Note that the definition of $PF(\Upsilon ,s,\tilde{s})$ (as well as the definitions of other probability functions
which we shall present) is based on the enumeration of activities which is considered implicit.

Let $\Upsilon\in Exec(s)$. Accomplishing the same multiset of executions in $s$ may result in the states different from
$\tilde{s}$ due to various changes of the positive phases related to executing the empty multisets of activities
accompanying that (non-empty) multiset. The {\em probability that the multiset of probabilistic executions $\Upsilon$
is ready for accomplishing in $s$} or the {\em weight of the multiset of definite (immediate) executions $\Upsilon$
which is ready for accomplishing in $s$} is $PF(\Upsilon ,s)=\sum_{\tilde{s}\in DR(G)}PF(\Upsilon ,s,\tilde{s})$.
Moreover, some other multisets of executions may be ready for accomplishing in $s$, besides $\Upsilon$. Hence, a kind
of conditioning or normalization is needed to calculate the execution probability. The {\em probability to move from
$s$ to $\tilde{s}$ by accomplishing the multiset of executions $\Upsilon$} is

$$PT(\Upsilon ,s,\tilde{s})=\frac{PF(\Upsilon ,s,\tilde{s})}{\sum_{\Xi\in Exec(s)}\sum_{\tilde{s}\in DR(G)}PF(\Xi
,s,\tilde{s})}.$$

If $s\in DR_{PT}(G)$ then $PT(\Upsilon ,s,\tilde{s})$ can be interpreted as the {\em conditional} probability to
accomplish $\Upsilon$ in $s$ calculated using the conditional probability formula in the form of ${\sf
P}(Z|W)=\frac{{\sf P}(Z\cap W)}{{\sf P}(W)}$. The event $Z$ consists in the exclusive accomplishing $\Upsilon$ in $s$
and moving to $\tilde{s}$, hence, ${\sf P}(Z)=PF(\Upsilon ,s,\tilde{s})$. The event $W$ consists in the exclusive
accomplishing of any set (including the empty one) $\Xi\in Exec(s)$ in $s$ (and moving to any $\tilde{s}\in DR(G)$).
Thus, $W=\cup_j Z_j$, where $\forall j,\ Z_j$ are mutually exclusive events (in a probability sense, i.e. intersection
of these events is the empty event) and $\exists i,\ Z=Z_i$. We have ${\sf P}(W)=\sum_j{\sf P}(Z_j)=\sum_{\Xi\in
Exec(s)}PF(\Xi ,s)=\sum_{\Xi\in Exec(s)}\sum_{\tilde{s}\in DR(G)}PF(\Xi ,s,\tilde{s})$, because summation reflects the
probability of the mutually exclusive event union. Since $Z\cap W=Z_i\cap (\cup_j Z_j)=Z_i=Z$, we have ${\sf
P}(Z|W)=\frac{{\sf P}(Z)}{{\sf P}(W)}=\frac{PF(\Upsilon ,s,\tilde{s})}{\sum_{\Xi\in Exec(s)}\sum_{\tilde{s}\in
DR(G)}PF(\Xi ,s,\tilde{s})}$. One can also treat $PT(\Upsilon ,s,\tilde{s})$ and $PF(\Upsilon ,s,\tilde{s}$ as the {\em
actual} and {\em potential} probabilities to accomplish $\Upsilon$ in $s$ and move to $\tilde{s}$, respectively. Note
that for $s\in DR_{PT}(G)$ we have $PT(\emptyset ,s,\tilde{s})\in (0;1]$. Hence, there is a non-zero probability to
execute the empty multiset of activities in $s$ and move to $\tilde{s}$ at the next time moment.

If $s\in DR_{DT}(G)\cup DR_V(G)$ then $PT(\Upsilon ,s,\tilde{s})$ can be interpreted as the weight of the set of
definite (immediate, respectively) executions $\Upsilon$ which is ready to accomplish in $s$ and move to $\tilde{s}$
{\em normalized} by the weights of {\em all} the sets that can be accomplished in $s$ (and moved to any $\tilde{s}\in
DR(G)$). This approach is analogous to that used in the EMPA definition of the probabilities of immediate actions
executable from the same process state \cite{BGo98} (inspired by way in which the probabilities of conflicting
immediate transitions in GSPNs are calculated \cite{Bal07}). The only difference is that we have a step semantics and,
for every set of definite (immediate) executions accomplished in parallel, we should use its cumulative weight. To get
the analogy with EMPA possessing interleaving semantics, we should interpret the weights of immediate actions of EMPA
as the cumulative weights of the sets of definite (immediate) executions of dtsdPBC.

The advantage of our two-stage approach to definition of the probability to execute a set of activities is that the
probability formula $PT(\Upsilon ,s,\tilde{s})$ is valid both for (sets of) timed (i.e. probabilistic or definite) and
immediate executions. It al\-lows one to unify the notation used later while constructing the operational semantics.

Note that the sum of outgoing probabilities for the expressions belonging to the derivations of $G$ is equal to $1$.
More formally, $\forall s\in DR(G)\ \sum_{\Upsilon\in Exec(s)}\sum_{\tilde{s}\in DR(G)}PT(\Upsilon ,s,\tilde{s})=1$.
This, obviously, follows from the definition of $PT(\Upsilon ,s,\tilde{s})$, and guarantees that it defines a
probability distribution.

The {\em probability to move from $s$ to $\tilde{s}$ by accomplishing any multiset of executions} is

$$PM(s,\tilde{s})=\sum_{\{\Upsilon\mid\exists H\in s\ \exists\widetilde{H}\in\tilde{s}\
H\stackrel{\Upsilon}{\rightarrow}\widetilde{H}\}}PT(\Upsilon ,s,\tilde{s}).$$

The summation in the definition above reflects the probability of the mutually exclusive event union, since
$\sum_{\{\Upsilon\mid\exists H\in s,\ \exists\widetilde{H}\in\tilde{s},\
H\stackrel{\Upsilon}{\rightarrow}\widetilde{H}\}}PT(\Upsilon ,s,\tilde{s})=\frac{1}{\sum_{\Xi\in
Exec(s)}\sum_{\tilde{s}\in DR(G)}PF(\Xi ,s,\tilde{s})}\cdot\sum_{\{\Upsilon\mid\exists H\in s,\
\exists\widetilde{H}\in\tilde{s},\ H\stackrel{\Upsilon}{\rightarrow}\widetilde{H}\}}PF(\Upsilon ,s,\tilde{s})$, where
for each $\Upsilon ,\ PF(\Upsilon ,s,\tilde{s})$ is the probability of the exclusive execution of $\Upsilon$ in $s$
and moving to $\tilde{s}$. Note that $\forall s\in DR(G)\ \sum_{\{\tilde{s}\mid\exists H\in s\
\exists\widetilde{H}\in\tilde{s}\ \exists\Upsilon\
H\stackrel{\Upsilon}{\rightarrow}\widetilde{H}\}}PM(s,\tilde{s})=\sum_{\{\tilde{s}\mid\exists H\in s\
\exists\widetilde{H}\in\tilde{s}\ \exists\Upsilon\ H\stackrel{\Upsilon}{\rightarrow}\widetilde{H}\}}
\sum_{\{\Upsilon\mid\exists H\in s\ \exists\widetilde{H}\in\tilde{s}\
H\stackrel{\Upsilon}{\rightarrow}\widetilde{H}\}}PT(\Upsilon ,s,\tilde{s})$,\\
which, by associativity and commutativity of addition, equals $\sum_{\Upsilon\in Exec(s)}\sum_{\tilde{s}\in
DR(G)}PT(\Upsilon ,s,\tilde{s})=1$.

\begin{example}
Let $E=(\{a\},\partial_1^{\bf A})\cho (\{b\},\partial_2^{\bf B})$, where ${\bf A}$ is the transient TPM of the
(discrete) uniform distribution $Unif(1,3)$ ($DPH({\bf A})=Unif(1,3)$) and ${\bf B}$ is the transient TPM of the
geometric distribution $Geom(\frac{1}{3})$ ($DPH({\bf B})=Geom(\frac{1}{3})$), i.e. ${\bf A}$ and ${\bf B}$ are from
Example \ref{dphminmax.exm}. $DR(\overline{E})$ consists of the equivalence classes
$$\begin{array}{l}
s_1=[\overline{(\{a\},\partial_1^{\bf A})^3}\cho (\{b\},\partial_2^{\bf B})^1]_\approx =
[(\{a\},\partial_1^{\bf A})^3\cho\overline{(\{b\},\partial_2^{\bf B})^1}]_\approx ,\\[1mm]
s_2=[\overline{(\{a\},\partial_1^{\bf A})^2}\cho (\{b\},\partial_2^{\bf B})^1]_\approx =
[(\{a\},\partial_1^{\bf A})^2\cho\overline{(\{b\},\partial_2^{\bf B})^1}]_\approx ,\\[1mm]
s_3=[\overline{(\{a\},\partial_1^{\bf A})^1}\cho (\{b\},\partial_2^{\bf B})^1]_\approx =
[(\{a\},\partial_1^{\bf A})^1\cho\overline{(\{b\},\partial_2^{\bf B})^1}]_\approx ,\\[1mm]
s_4=[\underline{(\{a\},\partial_1^{\bf A})\cho (\{b\},\partial_2^{\bf B})}]_\approx .
\end{array}$$
We have $DR_{PT}(\overline{E})=\{s_1,s_2,s_4\},\ DR_{DT}(\overline{E})=\{s_3\}$ and $DR_V(\overline{E})=\emptyset$. The
probabilities to accomplish the mul\-ti\-sets of exe\-cu\-ti\-ons in $s_1$
and move to some state are calculated as follows. Since $Exec(s_1)=\{\emptyset ,\{(\{a\},\partial_1^{\bf A})^3\},\\
\{(\{b\},\partial_2^{\bf B})^1\}\}$, we get $PF(\{(\{a\},\partial_1^{\bf
A})^3\},s_1,s_4)\!=\!\frac{1}{3}\cdot\frac{2}{3}\!=\! \frac{2}{9},\ PF(\{(\{b\},\partial_2^{\bf
B})^1\},s_1,s_4)\!=\!\frac{2}{3}\cdot\frac{1}{3}=\frac{2}{9}$ and $PF(\emptyset
,s_1,s_2)\!=\!\frac{2}{3}\cdot\frac{2}{3}\!=\!\frac{4}{9}$. Then $\sum_{\Xi\in Exec(s_1)}\sum_{\tilde{s}\in
DR(\overline{E})}PF(\Xi ,s_1,\tilde{s})\!=\!\frac{2}{9}+\frac{2}{9}+\frac{4}{9}\!=\!\frac{8}{9}$. Thus,
$PT(\{(\{a\},\partial_1^{\bf A})^3\},s_1,s_4)=\frac{\frac{2}{9}}{\frac{8}{9}}=\frac{1}{4},\
PT(\{((\{b\},\partial_2^{\bf B})^1\},s_1,s_4)=\frac{\frac{2}{9}}{\frac{8}{9}}=\frac{1}{4}$ and $PT(\emptyset
,s_1,s_2)=PM(s_1,s_2)=\frac{\frac{4}{9}}{\frac{8}{9}}=\frac{1}{2}$. Finally, $PM(s_1,s_4)=PT(\{\{a\},\partial_1^{\bf
A})^3\},s_1,s_4)+PT(\{(\{b\},\partial_2^{\bf B})^1\},s_1,s_4)=\frac{1}{4}+\frac{1}{4}=\frac{1}{2}$. In Table
\ref{trprobpht.tab}, the calculation of the probability functions $PF(\Upsilon ,s_1,\tilde{s}),\ PT(\Upsilon
,s_1,\tilde{s}),\ PM(s_1,\tilde{s})$ is explained, where $\Upsilon\in Exec(s_1),\ \tilde{s}\in\{s_2,s_4\}$ (the value
of $\tilde{s}$ is depicted in the parentheses near the value of $PM(s_1,\tilde{s})$) and $\Sigma =\sum_{\Xi\in
Exec(s_1)}\sum_{\tilde{s}\in DR(\overline{E})}PX(\Xi ,s_1,\tilde{s}),\ PX\in\{PF,PT,PM\}$.

Let $E'=(\{a\},\partial_l^\emptyset )\cho (\{b\},\partial_m^\emptyset )$, where $l,m\in\real_{>0}$. $DR(\overline{E'})$
consists of the equivalence classes
$$\begin{array}{l}
s_1'=[\overline{(\{a\},\partial_l^\emptyset )}\cho (\{b\},\partial_m^\emptyset )]_\approx =
[(\{a\},\partial_l^\emptyset )\cho\overline{(\{b\},\partial_m^\emptyset )}]_\approx,\\
s_2'=[\underline{(\{a\},\partial_l^\emptyset )\cho (\{b\},\partial_m^\emptyset )}]_\approx .
\end{array}$$
We have $DR_{PT}(\overline{E'})=\{s_2'\},\ DR_{DT}(\overline{E'})=\emptyset$ and $DR_V(\overline{E'})=\{s_1'\}$. The
probabilities to accomplish the mul\-ti\-sets of exe\-cu\-ti\-ons in $s_1'$ and move to some state are calculated as
follows. Since $Exec(s_1')\!=\!\{\{(\{a\},\partial_l^\emptyset )\},\{(\{b\},\partial_m^\emptyset )\}\}$, we get
$PF(\{(\{a\},\partial_l^\emptyset )\},s_1',s_2')\!=\!l$ and $PF(\{(\{b\},\partial_m^\emptyset )\},s_1',s_2')\!=\!m$.
Then $\sum_{\Xi\in Exec(s_1')}\sum_{\tilde{s}\in DR(\overline{E'})}PF(\Xi ,s_1',\tilde{s})\!=\!l+m$. Thus,
$PT(\{(\{a\},\partial_l^\emptyset )\},s_1',s_2')=\frac{l}{l+m}$ and $PT(\{(\{a\},\partial_l^\emptyset
)\},s_1',s_2')=\frac{m}{l+m}$. Finally, $PM(s_1',s_2')=PT(\{(\{a\},\partial_l^\emptyset
)\},s_1',s_2')+PT(\{(\{b\},\partial_m^\emptyset )\},s_1',s_2')=\frac{l}{l+m}+\frac{m}{l+m}=1$. In Table
\ref{trprobphi.tab}, the calculation of the pro\-ba\-bi\-li\-ty func\-ti\-ons $PF(\Upsilon ',s_1',\tilde{s}'),\
PT(\Upsilon ',s_1',\tilde{s}'),\ PM(s_1',\tilde{s}')$ is explained, where $\Upsilon '\in Exec(s_1'),\
\tilde{s}'\in\{s_2'\}$ (the value of $\tilde{s}'$ is de\-pic\-ted in the parentheses near the va\-lue of
$PM(s_1',\tilde{s}')$) and $\Sigma '=\sum_{\Xi\in Exec(s_1')}\sum_{\tilde{s}'\in DR(\overline{E'})}PX(\Xi
,s_1',\tilde{s}'),\\
PX\in\{PF,PT,PM\}$.
\label{trprobph.exm}
\end{example}

\begin{table}
\caption{Calculation of the probability functions $PF,\ PT,\ PM$ for $s_1\in DR(\overline{E})$ and
$E=(\{a\},\partial_1^{\bf A})\cho (\{b\},\partial_2^{\bf B})$, where $DPH({\bf A})=Unif(1,3)$ and
$DPH({\bf B})=Geom(\frac{1}{3})$}
\label{trprobpht.tab}
\begin{center}
$\begin{array}{|c||c|c|c|c|}
\hline
s_1\backslash\Upsilon ,\tilde{s}& \emptyset ,s_2 & \{(\{a\},\partial_1^{\bf A})^3\},s_4 &
\{(\{b\},\partial_2^{\bf B})^1\},s_4 & \Sigma\\
\hline\hline
PF & \frac{4}{9} & \frac{2}{9} & \frac{2}{9} & \frac{8}{9}\\
\hline
PT & \frac{1}{2} & \frac{1}{4} & \frac{1}{4} & 1\\
\hline
PM & \frac{1}{2}\ (s_2) & \multicolumn{2}{c|}{\frac{1}{2}\ (s_4)} & 1\\
\hline
\end{array}$
\end{center}
\end{table}

\begin{table}
\caption{Calculation of the probability functions $PF,\ PT,\ PM$ for $s_1'\in DR(\overline{E}')$ and
$E'=(\{a\},\partial_l^\emptyset )\cho (\{b\},\partial_m^\emptyset )$}
\label{trprobphi.tab}
\begin{center}
$\begin{array}{|c||c|c|c|}
\hline
s_1'\backslash\Upsilon ',\tilde{s}'& \{(\{a\},\partial_l^\emptyset )\},s_2' & \{(\{b\},\partial_m^\emptyset )\},s_2' &
\Sigma '\\
\hline\hline
PF & l & m & l+m\\
\hline
PT & \frac{l}{l+m} & \frac{m}{l+m} & 1\\
\hline
PM & \multicolumn{2}{c|}{1\ (s_2')} & 1\\
\hline
\end{array}$
\end{center}
\end{table}

Let $G$ be a dynamic expression, $s\in DR(G)$ and $\Upsilon\in Exec(s)$. The delay phases discarding
$\downharpoonleft\!\!\Upsilon$ of $\Upsilon$ is obtained by removing all delay phase superscripts from the executions
contained in the multisets belonging to $\Upsilon$. Then we get $\downharpoonleft\!\!\Upsilon\in\nat_{fin}^{\cal TL}$.

\begin{definition}
Let $G$ be a dynamic expression. The {\em (labeled probabilistic) transition system} of $G$ is a quadruple
$TS(G)=(S_G,L_G,{\cal T}_G,s_G)$, where
\begin{itemize}

\item the set of {\em states} is $S_G=DR(G)$;

\item the set of {\em labels} is $L_G=\nat_{fin}^{\cal PHL}\times (0;1]$;

\item the set of {\em transitions} is ${\cal T}_G=\{(s,(\Upsilon ,PT(\Upsilon ,s,\tilde{s})),\tilde{s})\mid
s,\tilde{s}\in DR(G),\ \exists H\in s\ \exists\widetilde{H}\in\tilde{s}\
H\stackrel{\Upsilon}{\rightarrow}\widetilde{H}\}\cup\\
\{(s,(\emptyset ,1),s)\mid s\in DR(G),\ Exec(s)=\emptyset\}$;

\item the {\em initial state} is $s_G=[G]_\approx$.

\end{itemize}
\end{definition}

\begin{example}
Let $E$ be from Example \ref{overlinph.exm}. The next inferences by rules {\bf C2e} and {\bf Be} are possible from the
elements of $[\overline{E}]_\approx$:
$$\begin{array}{l}
\overline{(\{a\},\partial_1^{\bf A})\cho (\{b\},\partial_2^{\bf B})}\approx
\overline{(\{a\},\partial_1^{\bf A})^3}\cho (\{b\},\partial_2^{\bf B})^1\stackrel{\emptyset}{\rightarrow}
\overline{(\{a\},\partial_1^{\bf A})^2}\cho (\{b\},\partial_2^{\bf B})^1,\\[1mm]
\overline{(\{a\},\partial_1^{\bf A})\cho (\{b\},\partial_2^{\bf B})}\approx
(\{a\},\partial_1^{\bf A})^3\cho\overline{(\{b\},\partial_2^{\bf B})^1}\stackrel{\emptyset}{\rightarrow}
(\{a\},\partial_1^{\bf A})^2\cho\overline{(\{b\},\partial_2^{\bf B})^1}.
\end{array}$$

The both inferences suggest the empty move transition $[\overline{E}]_\approx\stackrel{\emptyset}{\rightarrow}
[\overline{(\{a\},\partial_1^{\bf A})^2\cho (\{b\},\partial_2^{\bf B})^1}]_\approx\neq [\overline{E}]_\approx$. The
intuition is that the phase indicator of the enabled timed multiaction $(\{a\},\partial_1^{\bf A})$ is changed in the
both cases, whenever it is overlined or not.
\label{determph.exm}
\end{example}

In the above definition, ${\cal T}_G$ consists of two distinct sets. The first set collects the {\em standard}
transitions, de\-fi\-ned by the rules $H\stackrel{\Upsilon}{\rightarrow}\widetilde{H}$ from Table \ref{actrulesphm.tab}
and the probability function $PT(\Upsilon ,s,\tilde{s})$. The second set collects the {\em ad\-di\-ti\-o\-nal}
self-loop empty set transitions with probability $1$, associated with the terminal states in which no multiset of
executions exists, even the empty (multi)set of executions. The reason is that from them there is no transition that is
defined by the rules from Table \ref{actrulesphm.tab}. Those extra self-loops should be added ``artificially'' to the
transition system in order to guarantee that the Markov chain to be extracted from it has no states with zero exit
probabilities (they should be equal to $1$ by the definition of Markov chains). The situation when no transition exists
from a state may happen when all multisets of executions from a d-tangible (or vanishing) state contain as executions
only the superscribed timed (or immediate) multiactions, later affected by restriction. Since the state is d-tangible
or vanishing, all multisets of executions from it contain no empty set as the element that would remain after applying
the restriction. Particularly, for $s=[\overline{(\{a\},\partial_1^1)}\rs a]_\approx$ we have $Exec(s)=\emptyset$, and
for $s'=[\overline{(\{a\},\partial_1^\emptyset )}\rs a]_\approx$ we have $Exec(s')=\emptyset$, hence, the extra
self-loops with the empty set and probability $1$ should be associated with $s$ and $s'$.

The definition of $TS(G)$ is correct, i.e. for every state, the sum of the probabilities of all the transitions
starting from it is $1$. This is guaranteed by the note after the definition of $PT(\Upsilon ,s,\tilde{s})$, and by the
extra self-loops with the empty set and probability $1$ in the terminal states without executions. Thus, we have
defined a {\em generative} model of probabilistic processes, according to the classification from \cite{GSS95}. The
reason is that the sum of the probabilities of the transitions with all possible labels should be equal to $1$, not
only of those with the same labels (up to enumeration of activities they include) as in the {\em reactive} models, and
we do not have a nested probabilistic choice as in the {\em stratified} models.

The transition system $TS(G)$ associated with a dynamic expression $G$ describes all the steps (parallel
accomplishments) that occur at discrete time moments with some (one-step) probability and consist of multisets of
executions. Every step consisting of probabilistic (definite, respectively) executions or the empty step (i.e. that
consisting of the empty multiset of activities) occurs instantly after one discrete time unit delay. Each step
consisting of immediate executions occurs instantly without any delay. The step can change the current state to a
different one. The states are the structural equivalence classes of dynamic expressions obtained by application of
action rules starting from the expressions belonging to $[G]_\approx$. A transition $(s,(\Upsilon ,{\cal
P}),\tilde{s})\in{\cal T}_G$ will be written as $s\stackrel{\Upsilon}{\rightarrow}_{\cal P}\tilde{s}$. It is
interpreted as follows: the probability to change from state $s$ to $\tilde{s}$ as a result of accomplishing $\Upsilon$
is ${\cal P}$.

The empty move from $s$ in the form of $s\stackrel{\emptyset}{\rightarrow}_{\cal P}s$ is called the {\em empty loop}.
For d-tangible and vanishing states $\Upsilon$ cannot be the empty multiset, since we must accomplish some immediate
(definite, respectively) executions from them at the current (next, respectively) time moment.

The step probabilities belong to the interval $(0;1]$, being $1$ in the case when the only transition from an
p-tangible state $s$ is the empty move one $s\stackrel{\emptyset}{\rightarrow}_1 \tilde{s}$, or if there is just a
single transition from a d-tangible or a vanishing state.

We write $s\stackrel{\Upsilon}{\rightarrow}\tilde{s}$ if $\exists{\cal P}\ s\stackrel{\Upsilon}{\rightarrow}_{\cal
P}\tilde{s}$ and $s\rightarrow\tilde{s}$ if $\exists\Upsilon\ s\stackrel{\Upsilon}{\rightarrow}\tilde{s}$.

The first equivalence we are going to introduce is isomorphism which is a coincidence of systems up to renaming of
their components or states.

\begin{definition}
Let $G,G'$ be dynamic expressions and $TS(G)\!=\!(S_G,L_G,{\cal T}_G,s_G),TS(G')\!=\!(S_{G'},L_{G'},{\cal
T}_{G'},s_{G'})$ be their transition systems. A mapping $\beta :S_G\rightarrow S_{G'}$ is an {\em isomorphism} between
$TS(G)$ and $TS(G')$, denoted by $\beta :TS(G)\simeq TS(G')$, if
\begin{enumerate}

\item $\beta$ is a bijection such that $\beta (s_G)=s_{G'}$;

\item $\forall s,\tilde{s}\in S_G\ \forall\Upsilon\ s\stackrel{\Upsilon}{\rightarrow}_{\cal P}\tilde{s}\
\Leftrightarrow\ \beta (s)\stackrel{\Upsilon}{\rightarrow}_{\cal P}\beta (\tilde{s})$.

\end{enumerate}
Two transition systems $TS(G)$ and $TS(G')$ are {\em isomorphic}, denoted by $TS(G)\!\simeq\!TS(G')$, if
$\exists\beta\!:\!TS(G)\!\simeq\!TS(G')$.
\end{definition}

Transition systems of static expressions can be defined as well. For $E\in RegStatExpr$, let $TS(E)=TS(\overline{E})$.

\begin{definition}
Two dynamic expressions $G$ and $G'$ are {\em equivalent with respect to transition systems}, denoted by $G=_{ts}G'$,
if $TS(G)\simeq TS(G')$.
\end{definition}

\subsection{Examples of transition systems}

We now present a series of examples that demonstrate how to construct the transition systems of the dynamic expressions
that include various compositions of probabilistic, definite and immediate executions. In the transition systems, the
p-tangible and d-tangible states are depicted in ordinary and double ovals, respectively, and the vanishing ones are
depicted in boxes. To simplify the graphical representation, the singleton multisets of executions are written without
outer braces.

\begin{example}
Let $E_1=(\{a\},\partial_1^{\bf A}),\ E_2=(\{b\},\partial_2^{\bf B})$ and $E=E_1\cho E_2=(\{a\},\partial_1^{\bf A})\cho
(\{b\},\partial_2^{\bf B})$. Here ${\bf A}$ is the transient TPM of the (discrete) uniform distribution $Unif(1,2)$
with the value border parameters $1,2$ and support $\{1,2\}$ ($DPH({\bf A})=Unif(1,2)$). Next, ${\bf B}$ is the
transient TPM of the (discrete) uniform distribution $Unif(2,3)$ and with the value border parameters $2,3$ and support
$\{2,3\}$ ($DPH({\bf B})=Unif(2,3)$). The transient TPMs ${\bf A},\ {\bf B}$, and the corresponding vectors of the
transition pro\-ba\-bi\-li\-ti\-es towards absorption ${\bf a},\ {\bf b}$ are

$${\bf A}=\left(\begin{array}{cc}
0 & \frac{1}{2}\\
0 & 0
\end{array}\right),\
{\bf a}=\left(\begin{array}{c}
\frac{1}{2}\\
1
\end{array}\right),\
{\bf B}=\left(\begin{array}{ccc}
0 & 1 & 0\\
0 & 0 & \frac{1}{2}\\
0 & 0 & 0
\end{array}\right),\
{\bf b}=\left(\begin{array}{c}
0\\
\frac{1}{2}\\
1
\end{array}\right).$$

$DR(\overline{E_1})$ consists of the equivalence classes

$$\begin{array}{lll}
s_{11}=[\overline{(\{a\},\partial_1^{\bf A})^2}]_\approx , &
s_{12}=[\overline{(\{a\},\partial_1^{\bf A})^1}]_\approx , &
s_{13}=[\underline{(\{a\},\partial_1^{\bf A})}]_\approx .
\end{array}$$

We have $DR_{PT}(\overline{E_1})=\{s_{11},s_{13}\},\ DR_{DT}(\overline{E_1})=\{s_{12}\}$ and
$DR_V(\overline{E_1})=\emptyset$.

$DR(\overline{E_2})$ consists of the equivalence classes

$$\begin{array}{llll}
s_{21}=[\overline{(\{b\},\partial_2^{\bf B})^3}]_\approx , &
s_{21}=[\overline{(\{b\},\partial_2^{\bf B})^2}]_\approx , &
s_{21}=[\overline{(\{b\},\partial_2^{\bf B})^1}]_\approx , &
s_{22}=[\underline{(\{b\},\partial_2^{\bf B})}]_\approx .
\end{array}$$

We have $DR_{PT}(\overline{E_2})=\{s_{21},s_{22},s_{24}\},\ DR_{DT}(\overline{E_2})=\{s_{23}\}$ and
$DR_V(\overline{E_2})=\emptyset$.

$DR(\overline{E})$ consists of the equivalence classes

$$\begin{array}{l}
s_1=[\overline{(\{a\},\partial_1^{\bf A})^2}\cho (\{b\},\partial_2^{\bf B})^3]_\approx =
[(\{a\},\partial_1^{\bf A})^2\cho\overline{(\{b\},\partial_2^{\bf B})^3}]_\approx ,\\[1mm]
s_2=[\overline{(\{a\},\partial_1^{\bf A})^1}\cho (\{b\},\partial_2^{\bf B})^2]_\approx =
[(\{a\},\partial_1^{\bf A})^1\cho\overline{(\{b\},\partial_2^{\bf B})^2}]_\approx ,\\[1mm]
s_3=[\underline{(\{a\},\partial_1^{\bf A})\cho (\{b\},\partial_2^{\bf B})}]_\approx .
\end{array}$$

We have $DR_{PT}(\overline{E})=\{s_1,s_3\},\ DR_{DT}(\overline{E})=\{s_2\}$ and $DR_V(\overline{E})=\emptyset$. In
Figure \ref{tschoufuf.fig}, the transition systems $TS(\overline{E_1}),\ TS(\overline{E_2})$ and $TS(\overline{E})$ are
shown. Note that $TS(\overline{E})\simeq TS(\overline{E_1})$.

This example demonstrates a choice between two uniformly delayed timed multiactions with different (one time unit
shifted) delay ranges (discrete intervals). It shows that the uniformly delayed multiaction $(\{a\},\partial_1^{\bf
A})$ with a less delay range $\{1,2\}$ always accomplish its executions first, hence, the choice is resolved in favour
of it in any case and an absorbing state is then reached, so that the uniformly delayed multiaction
$(\{b\},\partial_2^{\bf B})$ with a greater delay range $\{2,3\}$ can never accomplish its executions.
\label{tschoufuf.exm}
\end{example}

\begin{figure}
\begin{center}
\includegraphics{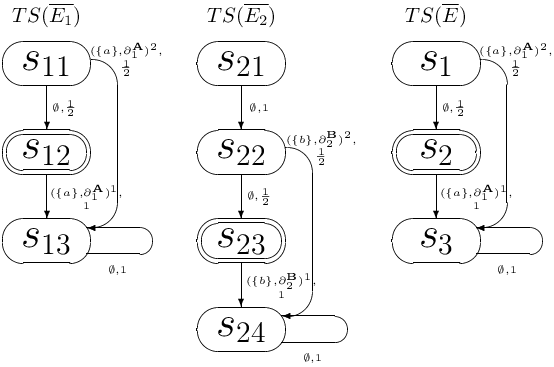}
\end{center}
\caption{The transition systems of $\overline{E_1},\ \overline{E_2}$ and $\overline{E}$ for $E_1=(\{a\},\partial_1^{\bf
A}),\ E_2=(\{b\},\partial_2^{\bf B})$ and $E=E_1\cho E_2$, where $DPH({\bf A})=Unif(1,2)$ and $DPH({\bf B})=Unif(2,3)$}
\label{tschoufuf.fig}
\end{figure}

\begin{example}
Let $E_1=(\{a\},\partial_1^{\bf A}),\ E_2=(\{b\},\partial_2^{\bf B})$ and $E=E_1\| E_2=(\{a\},\partial_1^{\bf A})\|
(\{b\},\partial_2^{\bf B})$. Here ${\bf A}$ is the
transient TPM of the (discrete) uniform distribution $Unif(1,2)$ ($DPH({\bf A})=Unif(1,2)$),
and ${\bf B}$ is the
transient TPM of the (discrete) uniform distribution $Unif(2,3)$ ($DPH({\bf B})=Unif(2,3)$), i.e. $E_1$ and $E_2$ are
from Example \ref{tschoufuf.exm}.

$DR(\overline{E})$ consists of the equivalence classes

$$\begin{array}{lll}
s_1=[\overline{(\{a\},\partial_1^{\bf A})^2}\|\overline{(\{b\},\partial_2^{\bf B})^3}]_\approx , &
s_2=[\overline{(\{a\},\partial_1^{\bf A})^1}\|\overline{(\{b\},\partial_2^{\bf B})^2}]_\approx , &
s_3=[\underline{(\{a\},\partial_1^{\bf A})}\|\overline{(\{b\},\partial_2^{\bf B})^1}]_\approx ,\\[1mm]
s_4=[\underline{(\{a\},\partial_1^{\bf A})}\|\overline{(\{b\},\partial_2^{\bf B})^2}]_\approx , &
s_5=[\underline{(\{a\},\partial_1^{\bf A})\| (\{b\},\partial_2^{\bf B})}]_\approx . &
\end{array}$$

We have $DR_{PT}(\overline{E})=\{s_1,s_4,s_5\},\ DR_{DT}(\overline{E})=\{s_2,s_3\}$ and $DR_V(\overline{E})=\emptyset$.
In Figure \ref{tsparufuf.fig}, the transition systems $TS(\overline{E_1}),\ TS(\overline{E_2})$ and $TS(\overline{E})$
are shown.

This example demonstrates a parallel composition of two uniformly delayed timed multiactions with different (one time
unit shifted) delay ranges (discrete intervals). It shows that from the two parallel timed multiactions, that
$(\{a\},\partial_1^{\bf A})$ with a less delay range $\{1,2\}$ accomplishes its executions first in any case. Next, the
executions accomplishment of the uniformly delayed multiaction $(\{b\},\partial_2^{\bf B})$ with a greater delay range
$\{2,3\}$ leads to an ab\-sor\-bing state. Thus, in spite of parallelism of those two timed multiactions, they
accomplish executions se\-quen\-ti\-al\-ly in fact, in the increasing order of their (different) delay ranges.
\label{tsparufuf.exm}
\end{example}

\begin{figure}
\begin{center}
\includegraphics{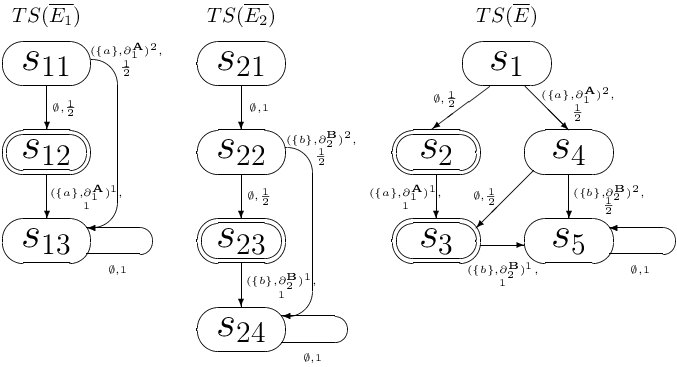}
\end{center}
\caption{The transition systems of $\overline{E_1},\ \overline{E_2}$ and $\overline{E}$ for $E_1=(\{a\},\partial_1^{\bf
A}),\ E_2=(\{b\},\partial_2^{\bf B})$ and $E=E_1\| E_2$, where $DPH({\bf A})=Unif(1,2)$ and $DPH({\bf B})=Unif(2,3)$}
\label{tsparufuf.fig}
\end{figure}

\begin{example}
Let $E_1=(\{a\},\partial_1^{\bf A}),\ E_2=(\{b\},\partial_2^{\bf B})$ and $E=E_1\cho E_2=(\{a\},\partial_1^{\bf A})\cho
(\{b\},\partial_2^{\bf B})$. Here ${\bf A}$ is the transient TPM of the (discrete) uniform distribution $Unif(1,3)$
($DPH({\bf A})=Unif(1,3)$) and ${\bf B}$ is the transient TPM of the geometric distribution $Geom(\frac{1}{3})$
($DPH({\bf B})=Geom(\frac{1}{3})$), i.e. ${\bf A}$ and ${\bf B}$ are from Example \ref{dphminmax.exm}.

$DR(\overline{E_1})$ consists of the equivalence classes

$$\begin{array}{llll}
s_{11}=[\overline{(\{a\},\partial_1^{\bf A})^3}]_\approx , &
s_{12}=[\overline{(\{a\},\partial_1^{\bf A})^2}]_\approx , &
s_{13}=[\overline{(\{a\},\partial_1^{\bf A})^1}]_\approx , &
s_{14}=[\underline{(\{a\},\partial_1^{\bf A})}]_\approx .
\end{array}$$

We have $DR_{PT}(\overline{E_1})=\{s_{11},s_{12},s_{14}\},\ DR_{DT}(\overline{E_1})=\{s_{13}\}$ and
$DR_V(\overline{E_1})=\emptyset$.

$DR(\overline{E_2})$ consists of the equivalence classes

$$\begin{array}{ll}
s_{21}=[\overline{(\{b\},\partial_2^{\bf B})^1}]_\approx , &
s_{22}=[\underline{(\{b\},\partial_2^{\bf B})}]_\approx .
\end{array}$$

We have $DR_{PT}(\overline{E_2})=\{s_{21},s_{22}\},\ DR_{DT}(\overline{E_2})=\emptyset =DR_V(\overline{E_2})$.

$DR(\overline{E})$ consists of the equivalence classes

$$\begin{array}{l}
s_1=[\overline{(\{a\},\partial_1^{\bf A})^3}\cho (\{b\},\partial_2^{\bf B})^1]_\approx =
[(\{a\},\partial_1^{\bf A})^3\cho\overline{(\{b\},\partial_2^{\bf B})^1}]_\approx ,\\[1mm]
s_2=[\overline{(\{a\},\partial_1^{\bf A})^2}\cho (\{b\},\partial_2^{\bf B})^1]_\approx =
[(\{a\},\partial_1^{\bf A})^2\cho\overline{(\{b\},\partial_2^{\bf B})^1}]_\approx ,\\[1mm]
s_3=[\overline{(\{a\},\partial_1^{\bf A})^1}\cho (\{b\},\partial_2^{\bf B})^1]_\approx =
[(\{a\},\partial_1^{\bf A})^1\cho\overline{(\{b\},\partial_2^{\bf B})^1}]_\approx ,\\[1mm]
s_4=[\underline{(\{a\},\partial_1^{\bf A})\cho (\{b\},\partial_2^{\bf B})}]_\approx .
\end{array}$$

We have $DR_{PT}(\overline{E})=\{s_1,s_2,s_4\},\ DR_{DT}(\overline{E})=\{s_3\}$ and $DR_V(\overline{E})=\emptyset$. In
Figure \ref{tschoufge.fig}, the transition systems $TS(\overline{E_1}),\ TS(\overline{E_2})$ and $TS(\overline{E})$ are
shown.

This example demonstrates a choice between a uniformly delayed timed multiaction and a geometrically delayed one. It
shows that the geometrically delayed multiaction $(\{b\},\partial_2^{\bf B})$ can accomplish the probabilistic
execution $(\{b\},\partial_2^{\bf B})^1$ until the phase of the uniformly delayed multiaction $(\{a\},\partial_1^{\bf
A})$ beco\-mes $1$, after which only the definite execution $(\{a\},\partial_1^{\bf A})^1$ can be accomplished by
$(\{a\},\partial_1^{\bf A})$ in the next moment, leading to an absorbing state. Thus, in our setting, a uniformly
delayed multiaction that cannot accomplish definite executions (but only probabilistic executions or the empty set of
activities) in the next moment and whose phase indicator is still active may be interrupted (preempted) by a
geometrically delayed multiaction that accomplishes its probabilistic execution instead.
\label{tschoufge.exm}
\end{example}

\begin{figure}
\begin{center}
\includegraphics{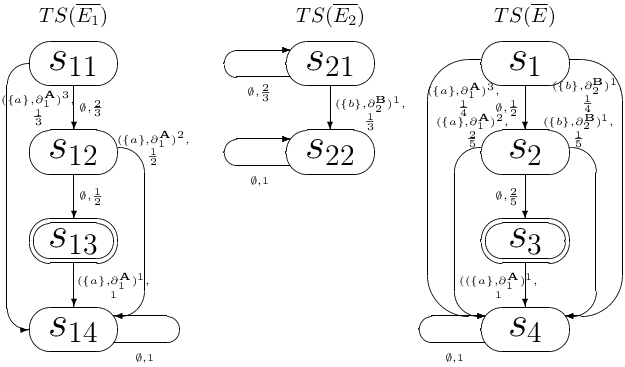}
\end{center}
\caption{The transition systems of $\overline{E_1},\ \overline{E_2}$ and $\overline{E}$ for $E_1=(\{a\},\partial_1^{\bf
A}),\ E_2=(\{b\},\partial_2^{\bf B})$ and $E=E_1\cho E_2$, where $DPH({\bf A})=Unif(1,3)$ and
$DPH({\bf B})=Geom(\frac{1}{3})$}
\label{tschoufge.fig}
\end{figure}

\begin{example}
Let $E_1=(\{a\},\partial_1^{\bf A}),\ E_2=(\{b\},\partial_2^{\bf B})$ and $E=E_1\|E_2=(\{a\},\partial_1^{\bf A})\|
(\{b\},\partial_2^{\bf B})$. Here ${\bf A}$ is the transient TPM of the (discrete) uniform distribution $Unif(1,3)$
($DPH({\bf A})=Unif(1,3)$) and ${\bf B}$ is the transient TPM of the geometric distribution $Geom(\frac{1}{3})$
($DPH({\bf B})=Geom(\frac{1}{3})$), i.e. $E_1$ and $E_2$ are from Example \ref{tschoufge.exm}.

$DR(\overline{E})$ consists of the equivalence classes

$$\begin{array}{lll}
s_1=[\overline{(\{a\},\partial_1^{\bf A})^3}\|\overline{(\{b\},\partial_2^{\bf B})^1}]_\approx , &
s_2=[\overline{(\{a\},\partial_1^{\bf A})^2}\|\overline{(\{b\},\partial_2^{\bf B})^1}]_\approx , &
s_3=[\overline{(\{a\},\partial_1^{\bf A})^1}\|\overline{(\{b\},\partial_2^{\bf B})^1}]_\approx ,\\[1mm]
s_4=[\underline{(\{a\},\partial_1^{\bf A})}\|\overline{(\{b\},\partial_2^{\bf B})^1}]_\approx , &
s_5=[\overline{(\{a\},\partial_1^{\bf A})^2}\|\underline{(\{b\},\partial_2^{\bf B})}]_\approx , &
s_6=[\overline{(\{a\},\partial_1^{\bf A})^1}\|\underline{(\{b\},\partial_2^{\bf B})}]_\approx ,\\[1mm]
s_7=[\underline{(\{a\},\partial_1^{\bf A})\| (\{b\},\partial_2^{\bf B})}]_\approx . & &
\end{array}$$

We have $DR_{PT}(\overline{E})=\{s_1,s_2,s_4,s_5,s_7\},\ DR_{DT}(\overline{E})=\{s_3,s_6\}$ and
$DR_V(\overline{E})=\emptyset$. In Figure \ref{tsparufge.fig}, the transition systems $TS(\overline{E_1}),\
TS(\overline{E_2})$ and $TS(\overline{E})$ are shown.

This example demonstrates a parallel composition of a uniformly delayed timed multiaction and a geometrically delayed
one. It shows that the geometrically delayed multiaction $(\{b\},\partial_2^{\bf B})$ can accomplish the probabilistic
execution $(\{b\},\partial_2^{\bf B})^1$ until the phase of the uniformly delayed multiaction $(\{a\},\partial_1^{\bf
A})$ beco\-mes $1$, after which only the definite execution $(\{a\},\partial_1^{\bf A})^1$ can be accomplished by
$(\{a\},\partial_1^{\bf A})$ in the next moment. The accomplishment of the latter execution leads to an absorbing state
either directly or indirectly, via executing a possible empty loop, followed (via sequential composition) by
accomplishing from $(\{b\},\partial_2^{\bf B})$ the probabilistic execution $(\{b\},\partial_2^{\bf B})^1$ that has not
been accomplished in the preceding states.
\label{tsparufge.exm}
\end{example}

\begin{figure}
\begin{center}
\includegraphics{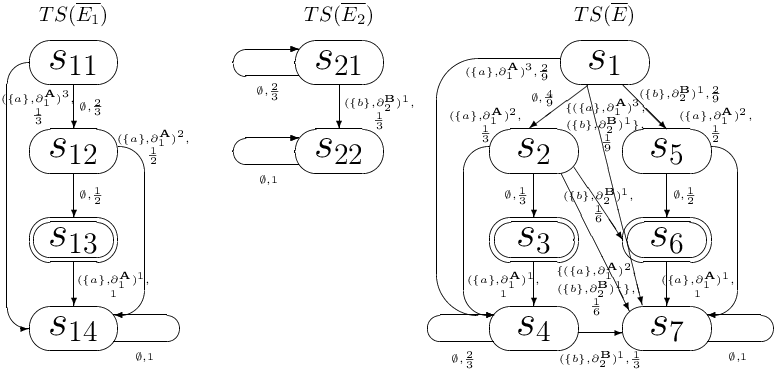}
\end{center}
\caption{The transition systems of $\overline{E_1},\ \overline{E_2}$ and $\overline{E}$ for $E_1=(\{a\},\partial_1^{\bf
A}),\ E_2=(\{b\},\partial_2^{\bf B})$ and $E=E_1\|E_2$, where $DPH({\bf A})=Unif(1,3)$ and
$DPH({\bf B})=Geom(\frac{1}{3})$}
\label{tsparufge.fig}
\end{figure}

\begin{example}
Let $E_1=(\{a\},\partial_1^{\bf A}),\ E_2=(\{\hat{a}\},\partial_2^{\bf B})$ and $E=(E_1\|E_2)\sy a=
((\{a\},\partial_1^{\bf A})\| (\{\hat{a}\},\partial_2^{\bf B}))\sy a$. Here ${\bf A}$ is the transient TPM of the
(discrete) uniform distribution $Unif(1,3)$ ($DPH({\bf A})=Unif(1,3)$), i.e. $E_1$ is from Example \ref{tschoufge.exm}.
Next, ${\bf B}$ is the transient TPM of the (discrete) uniform distribution $Unif(2,4)$ with the value border
parameters $2,4$ and support $\{2,3,4\}$ ($DPH({\bf B})=Unif(2,4)$). The transient TPM ${\bf B}$ and the corresponding
vector of the transition pro\-ba\-bi\-li\-ti\-es towards absorption ${\bf b}$ are

$${\bf B}=\left(\begin{array}{cccc}
0 & 1 & 0 & 0\\
0 & 0 & \frac{2}{3} & 0\\
0 & 0 & 0 & \frac{1}{2}\\
0 & 0 & 0 & 0
\end{array}\right),\
{\bf b}=\left(\begin{array}{c}
0\\
\frac{1}{3}\\
\frac{1}{2}\\
1
\end{array}\right).$$

$DR(\overline{E_2})$ consists of the equivalence classes

$$\begin{array}{lllll}
s_{21}=[\overline{(\{\hat{a}\},\partial_2^{\bf B})^4}]_\approx , &
s_{22}=[\overline{(\{\hat{a}\},\partial_2^{\bf B})^3}]_\approx , &
s_{23}=[\overline{(\{\hat{a}\},\partial_2^{\bf B})^2}]_\approx , &
s_{24}=[\overline{(\{\hat{a}\},\partial_2^{\bf B})^1}]_\approx , &
s_{25}=[\underline{(\{\hat{a}\},\partial_2^{\bf B})}]_\approx .
\end{array}$$

We have $DR_{PT}(\overline{E_2})=\{s_{21},s_{22},s_{23},s_{25}\},\ DR_{DT}(\overline{E_2})=\{s_{24}\}$ and
$DR_V(\overline{E_2})=\emptyset$.

$DR(\overline{E})$ consists of the equivalence classes

$$\begin{array}{ll}
s_1=[(\overline{(\{a\},\partial_1^{\bf A})^3}\|\overline{(\{\hat{a}\},\partial_2^{\bf B})^4})\sy a]_\approx , &
s_2=[(\overline{(\{a\},\partial_1^{\bf A})^2}\|\overline{(\{\hat{a}\},\partial_2^{\bf B})^3})\sy a]_\approx ,\\[1mm]
s_3=[(\overline{(\{a\},\partial_1^{\bf A})^1}\|\overline{(\{\hat{a}\},\partial_2^{\bf B})^2})\sy a]_\approx , &
s_4=[(\underline{(\{a\},\partial_1^{\bf A})}\|\overline{(\{\hat{a}\},\partial_2^{\bf B})^1})\sy a]_\approx ,\\[1mm]
s_5=[(\overline{(\{a\},\partial_1^{\bf A})^1}\|\underline{(\{\hat{a}\},\partial_2^{\bf B})})\sy a]_\approx , &
s_6=[(\underline{(\{a\},\partial_1^{\bf A})}\|\overline{(\{\hat{a}\},\partial_2^{\bf B})^3})\sy a]_\approx ,\\[1mm]
s_7=[(\underline{(\{a\},\partial_1^{\bf A})}\|\overline{(\{\hat{a}\},\partial_2^{\bf B})^2})\sy a]_\approx , &
s_8=[\underline{((\{a\},\partial_1^{\bf A})\| (\{\hat{a}\},\partial_2^{\bf B}))\sy a}]_\approx .
\end{array}$$

We have $DR_{PT}(\overline{E})=\{s_1,s_2,s_4,s_5,s_7\},\ DR_{DT}(\overline{E})=\{s_3,s_6\}$ and
$DR_V(\overline{E})=\emptyset$. In Figure \ref{tssynufuf.fig}, the transition systems $TS(\overline{E_1}),\
TS(\overline{E_2})$ and $TS(\overline{E})$ are shown.

This example demonstrates a parallel composition of two timed multiactions $(\{a\},\partial_1^{\bf A})$ and
$(\{\hat{a}\},\partial_2^{\bf B})$, whose multiaction parts are singleton multisets with an action $a$ and its
conjugate $\hat{a}$, respectively. The resulting composition is synchronized by that action. From the initial state,
only the empty multiset of activities is executed that decrements by one the values of the timers. That evolution
follows by the execution of a new timed multiaction $(\emptyset ,\partial_3^{{\bf A}\sqcup{\bf B}})$ with the empty
multiaction part, resulted from synchronization of the two timed multiactions, which leads to an absorbing state.

Note that the delay phases of the two timed multiactions and that of the new timed multiaction $(\emptyset
,\partial_3^{{\bf A}\sqcup{\bf B}})$ (being their synchronous product) coincide until all of them remain enabled with
the time progress. Thus, it is very useful that the expression syntax preserves such two enabled synchronized timed
multiactions, possibly removed by restriction from the behaviour, since their timer values suggest that of their
synchronous product, which is possibly not explicit in the syntax. Thus, the delay phases of those two possibly
``virtual'' enabled timed multiactions cannot just be marked as undefined in the syntax, provided that one keeps track
of the delay phase of their synchronous product being possibly only implicit in the syntax.

If both synchronized timed multiactions lose their enabledness with the time progress then their synchronous product
$(\emptyset ,\partial_3^{{\bf A}\sqcup{\bf B}})$ also loses its enabledness and all of them obviously loose their delay
phase annotations. It may happen that one of the synchronized timed multiactions loses its enabledness (for example,
when a conflicting timed multiaction is executed) while the other one keeps its enabledness. Then their synchronous
product also loses its enabledness, together with its delay phase annotation. In such a case, the delay phase of the
enabled synchronized timed multiaction does not suggest anymore that of the synchronous product. That ``saved'' delay
phase merely changes with every time tick unless it becomes equal to the phase $i$ such that ${\cal C}_i>1\ ({\bf
C}={\bf A}\sqcup{\bf B})$, in which either the enabled synchronized timed multiaction accomplishes its execution or it
cannot accomplish its execution by some reason (for example, when affected by restriction) and then the delay phase $i$
remains unchanged with the time progress.
\label{tssynufuf.exm}
\end{example}

\begin{figure}
\begin{center}
\includegraphics{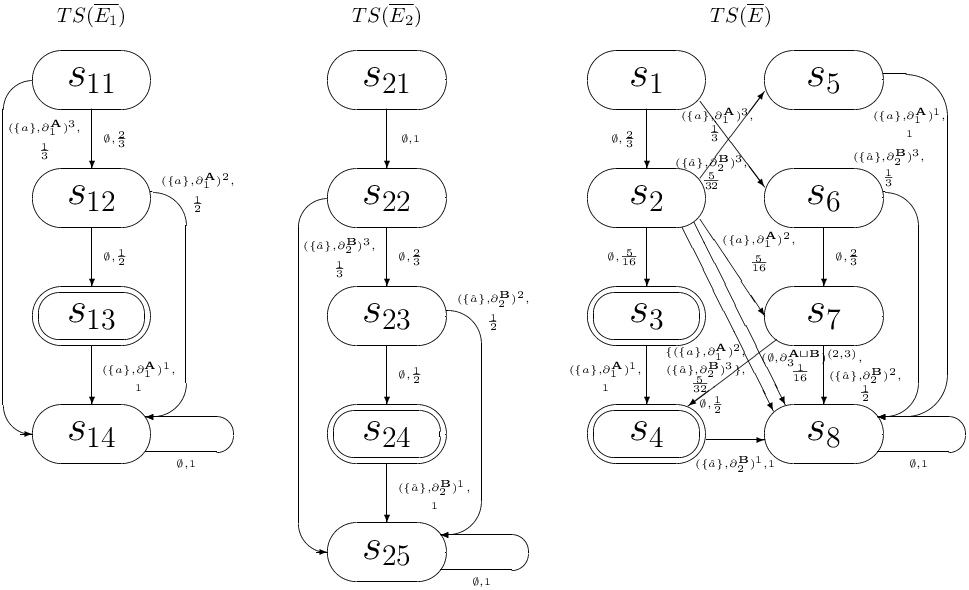}
\end{center}
\caption{The transition systems of $\overline{E_1},\ \overline{E_2}$ and $\overline{E}$ for $E_1=(\{a\},\partial_1^{\bf
A}),\ E_2=(\{\hat{a}\},\partial_2^{\bf B})$ and $E=(E_1\|E_2)\sy a$, where $DPH({\bf A})=Unif(1,3)$ and
$DPH({\bf B})=Unif(2,4)$}
\label{tssynufuf.fig}
\end{figure}

\begin{example}
Consider the expression ${\sf Stop}=(\{g\},\partial_1^{\frac{1}{2}})\rs g$ specifying a special process that is only
able to perform empty loops with probability $1$ and never terminates. We could actually use any arbitrary action from
${\cal A}$ and any probability (one-element transient TPM) belonging to the interval $(0;1)$ in the definition of ${\sf
Stop}$. Note that ${\sf Stop}$ is analogous to the one used in the examples within sPBC. The latter is a continuous
time stochastic analogue of the {\sf stop} process proposed in \cite{BDK01}. {\sf Stop} is a discrete time stochastic
analogue of the {\sf stop}.

Let $E_1=(\{a\},\partial_j^{\bf A}),\ E_2=(\{b\},\partial_k^{\bf B}),\ E_3=(\{c\},\partial_l^\emptyset ),\
E_4=(\{d\},\partial_p^{\bf D}),\ E_5=(\{e\},\partial_m^\emptyset ),\ E_6=(\{f\},\partial_q^{\bf F})$ and
$E\!\!=\!\![E_1*(E_2;((E_3;E_4)\cho (E_5;E_6)))*{\sf Stop}]\!=\![(\{a\},\partial_j^{\bf A})*((\{b\},\partial_k^{\bf
B});(((\{c\},\partial_l^\emptyset );(\{d\},\partial_p^{\bf C}))\cho ((\{e\},\partial_m^\emptyset
);(\{f\},\partial_q^{\bf D}))))*{\sf Stop}]$, where $j,k,l,m,p,q\in\real_{>0}$.

Here ${\bf A}$ is the transient TPM of the geometric distribution $Geom(\rho )$ with the probability parameter $\rho\in
(0;1)$ and support $\nat_{\geq 1}$ ($DPH({\bf A})=Geom(\rho )$). Next, ${\bf B}$ is the transient TPM of the (discrete)
uniform distribution $Unif(1,2)$ ($DPH({\bf B})=Unif(1,2)$), i.e. ${\bf B}$ is from Example \ref{tschoufuf.exm}. At
last, ${\bf D}$ is the transient TPM of the geometric distribution $Geom(\theta )$ ($DPH({\bf D})=Geom(\theta )$) and
${\bf F}$ is the transient TPM of the geometric distribution $Geom(\phi )$ ($DPH({\bf F})=Geom(\phi )$). Remember that
$\emptyset$ is the transient TPM of the Dirac distribution $Dirac(0)$ ($DPH(\emptyset )=Dirac(0)$, being the delay of
immediate multiactions $(\{c\},\partial_l^\emptyset )$ and $(\{e\},\partial_m^\emptyset )$.

The transient TPMs ${\bf A},\ {\bf D},\ {\bf F}$, and the corresponding vectors of the transition
pro\-ba\-bi\-li\-ti\-es towards absorption ${\bf a},\ {\bf d},\ {\bf f}$ are
$${\bf A}=1-\rho ,\ {\bf a}=\rho ,\ {\bf D}=1-\theta ,\ {\bf d}=\theta ,\ {\bf F}=1-\phi ,\ {\bf f}=\phi .$$

$DR(\overline{E_1})$ consists of the equivalence classes
$$\begin{array}{ll}
s_{11}=[\overline{(\{a\},\partial_j^{\bf A})^1}]_\approx , &
s_{12}=[\underline{(\{a\},\partial_j^{\bf A})}]_\approx .
\end{array}$$

We have $DR_{PT}(\overline{E_1})=\{s_{11},s_{12}\},\ DR_{DT}(\overline{E_1})=\emptyset =DR_V(\overline{E_1})$.

$DR(\overline{E_3})$ consists of the equivalence classes
$$\begin{array}{ll}
s_{31}=[\overline{(\{c\},\partial_l^\emptyset )}]_\approx , &
s_{32}=[\underline{(\{c\},\partial_l^\emptyset )}]_\approx .
\end{array}$$

We have $DR_{PT}(\overline{E_3})=\{s_{32}\},\ DR_{DT}(\overline{E_3})=\emptyset$ and $DR_V(\overline{E_3})=\{s_{31}\}$.

$DR(\overline{E_4})$ consists of the equivalence classes
$$\begin{array}{ll}
s_{41}=[\overline{(\{d\},\partial_p^{\bf D})^1}]_\approx , &
s_{42}=[\underline{(\{d\},\partial_p^{\bf D})}]_\approx .
\end{array}$$

We have $DR_{PT}(\overline{E_4})=\{s_{41},s_{42}\},\ DR_{DT}(\overline{E_4})=\emptyset =DR_V(\overline{E_4})$.

$DR(\overline{E_5})$ consists of the equivalence classes
$$\begin{array}{ll}
s_{51}=[\overline{(\{e\},\partial_m^\emptyset )}]_\approx , &
s_{52}=[\underline{(\{e\},\partial_m^\emptyset )}]_\approx .
\end{array}$$

We have $DR_{PT}(\overline{E_5})=\{s_{52}\},\ DR_{DT}(\overline{E_5})=\emptyset$ and $DR_V(\overline{E_5})=\{s_{51}\}$.

$DR(\overline{E_6})$ consists of the equivalence classes
$$\begin{array}{ll}
s_{61}=[\overline{(\{f\},\partial_q^{\bf F})^1}]_\approx , &
s_{62}=[\underline{(\{f\},\partial_q^{\bf F})}]_\approx .
\end{array}$$

We have $DR_{PT}(\overline{E_6})=\{s_{61},s_{62}\},\ DR_{DT}(\overline{E_6})=\emptyset =DR_V(\overline{E_6})$.

$DR(\overline{E})$ consists of the equivalence classes
$$\begin{array}{l}
s_1=[[\overline{(\{a\},\partial_j^{\bf A})^1}*((\{b\},\partial_k^{\bf B});(((\{c\},\partial_l^\emptyset );
(\{d\},\partial_p^{\bf D}))\cho ((\{e\},\partial_m^\emptyset );(\{f\},\partial_q^{\bf F}))))*
{\sf Stop}]]_\approx ,\\[1mm]
s_2=[[(\{a\},\partial_j^{\bf A})*(\overline{(\{b\},\partial_k^{\bf B})^2};(((\{c\},\partial_l^\emptyset );
(\{d\},\partial_p^{\bf D}))\cho ((\{e\},\partial_m^\emptyset );(\{f\},\partial_q^{\bf F}))))*
{\sf Stop}]]_\approx ,\\[1mm]
s_3=[[(\{a\},\partial_j^{\bf A})*(\overline{(\{b\},\partial_k^{\bf B})^1};(((\{c\},\partial_l^\emptyset );
(\{d\},\partial_p^{\bf D}))\cho ((\{e\},\partial_m^\emptyset );(\{f\},\partial_q^{\bf F}))))*
{\sf Stop}]]_\approx ,\\[1mm]
s_4=[[(\{a\},\partial_j^{\bf A})*((\{b\},\partial_k^{\bf B});((\overline{(\{c\},\partial_l^\emptyset )};
(\{d\},\partial_p^{\bf D}))\cho ((\{e\},\partial_m^\emptyset );(\{f\},\partial_q^{\bf F}))))*
{\sf Stop}]]_\approx =\\
\hspace{8mm}[[(\{a\},\partial_j^{\bf A})*((\{b\},\partial_k^{\bf B});(((\{c\},\partial_l^\emptyset );
(\{d\},\partial_p^{\bf D}))\cho (\overline{(\{e\},\partial_m^\emptyset )};(\{f\},\partial_q^{\bf F}))))*
{\sf Stop}]]_\approx ,\\[1mm]
s_5=[[(\{a\},\partial_j^{\bf A})*((\{b\},\partial_k^{\bf B});(((\{c\},\partial_l^\emptyset );
\overline{(\{d\},\partial_p^{\bf D})^1})\cho ((\{e\},\partial_m^\emptyset );(\{f\},\partial_q^{\bf F}))))*
{\sf Stop}]]_\approx ,\\[1mm]
s_6=[[(\{a\},\partial_j^{\bf A})*((\{b\},\partial_k^{\bf B});(((\{c\},\partial_l^\emptyset );
(\{d\},\partial_p^{\bf D}))\cho ((\{e\},\partial_m^\emptyset );\overline{(\{f\},\partial_q^{\bf F})^1})))*
{\sf Stop}]]_\approx .
\end{array}$$

We have $DR_{PT}(\overline{E})=\{s_1,s_2,s_5,s_6\},\ DR_{DT}(\overline{E})=\{s_3\}$ and $DR_V(\overline{E})=\{s_4\}$.
In Figure \ref{tsitchogeufdz.fig}, the transition system $TS(\overline{E})$ is presented.

This example demonstrates an infinite iteration loop. The loop is preceded with the iteration initiation, modeled by a
timed multiaction $(\{a\},\partial_j^{\bf A})$. The iteration body that corresponds to the loop consists of a timed
multiaction $(\{b\},\partial_k^{\bf B})$, followed (via sequential composition) by the probabilistic choice, modeled
via two conflicting immediate multiactions $(\{c\},\partial_l^\emptyset )$ and $(\{e\},\partial_m^\emptyset )$,
followed by different timed multiactions $(\{d\},\partial_p^{\bf D})$ and $(\{f\},\partial_q^{\bf F})$. The iteration
termination ${\sf Stop}$ demonstrates an empty behaviour, assuring that the iteration does not reach its final state
after any number of repeated executions of its body.
\label{tsitchogeufdz.exm}
\end{example}

\begin{figure}
\begin{center}
\includegraphics{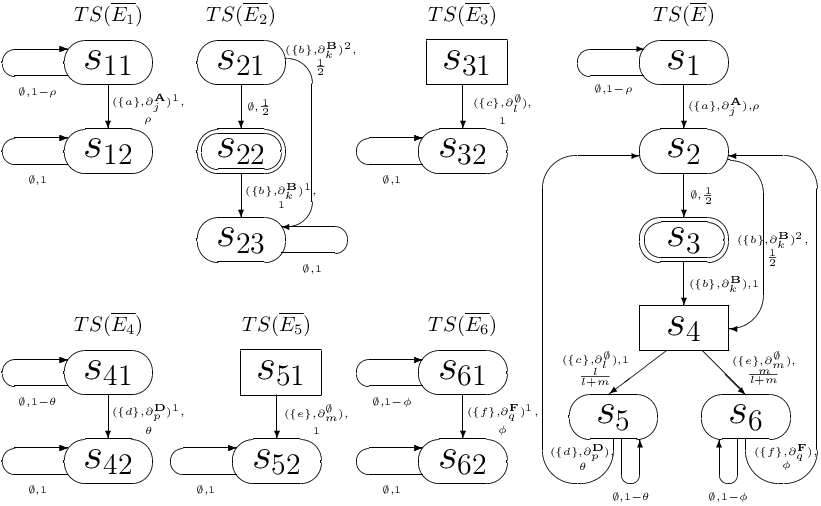}
\end{center}
\caption{The transition systems of $\overline{E_1},\ \overline{E_2},\ \overline{E_3},\ \overline{E_4},\
\overline{E_5},\ \overline{E_6}$ and $\overline{E}$ for $E_1=(\{a\},\partial_j^{\bf A}),\ E_2=(\{b\},\partial_k^{\bf
B}),\newline
E_3=(\{c\},\partial_l^\emptyset ),\ E_4=(\{d\},\partial_p^{\bf D}),\ E_5=(\{e\},\partial_m^\emptyset ),\
E_6=(\{f\},\partial_q^{\bf F})$ and $E=[E_1*(E_2;((E_3;E_4)\cho (E_5;E_6)))* {\sf Stop}]$, where $DPH({\bf
A})=Geom(\rho ),\ DPH({\bf B})=Unif(1,2),\ DPH({\bf D})=Geom(\theta )$ and $DPH({\bf F})=Geom(\phi )$}
\label{tsitchogeufdz.fig}
\end{figure}

Note that, due to the time constraints and since definite multiactions may be preempted by geometric ones, some simple
dynamic expressions can have complex transition systems, as in Examples \ref{tsparufuf.exm}--\ref{tssynufuf.exm}.

\section{Conclusion}
\label{conclusion.sec}

We now sum up the work presented in the paper and explain research perspectives along this research topic.

\subsection{Results obtained}

In this paper, we have proposed a discrete time phased extension dtphPBC of Petri Box Calculus (PBC)
\cite{BDH92,BKo95,BDK01,BD24}, by associating phase type stochastic delays with the PBC multiactions, giving rise to
the enriched syntax with phased multiactions. The resulting process algebra dtphPBC has a parallel step operational
semantics, based on labeled probabilistic transition systems. The syntax and semantics of dtphPBC have been illustrated
with a series of examples, describing probabilistic, definite, immediate executions and empty moves in the transition
systems, constructed for various dynamic expressions, combined from timed (positively phased) and immediate (zero
phased) multiactions with different operations of the calculus.

\subsection{Originality of the model}

The merit of our framework is twofold. First, one can specify in it concurrent composition and synchro\-nization of
(multi)actions, whereas this is not possible in classical Markov chains. As argued in \cite{Tri16}, (stochastic) PNs
represent the systems structure more concisely and can be an intermediate formalism for their more intuitive
translation into Markov chains. Second, algebraic formulas represent processes in a more compact way than PNs and allow
one to apply syntactic transformations and comparisons. Process algebras are compositional by definition and their
operations naturally correspond to operators of programming languages. Hence, it is much easier to construct a complex
model in the algebraic setting than in PNs. The complexity of PNs generated for practical models in the literature
demonstrates that it is not straightforward to construct such PNs directly from the system specifications.

dtphPBC is well suited for the discrete time applications, whose discrete states change with a global time tick, such
as business processes, neural and transportation networks, computer and communication systems, timed web services
\cite{VC17}, as well as for those, in which the distributed architecture or the concurrency level should be preserved
while modeling and analysis, such as genetic regulatory and cellular signalling networks (featuring maximal
parallelism) in biology \cite{DFFHK07,DFFK07,DFFK08,BFFK09,BL16} (remember that, in step semantics, we have additional
transitions due to concurrent executions).

In \cite{GHL07}, biological networks were jointly modeled by (standard, qualitative) PNs, CTSPNs and continuous PNs
(CPNs), to demonstrate their complementarity that makes necessary adding deterministic time to stochastic models, as
well as combining stochastic and continuous (deterministic) aspects into one model (such as stochastic rates of
reactions and continuous amounts of species). Like dtsdPBC, the process algebra dtphPBC is capable to model and analyze
parallel systems with fixed durations of the typical activities (loading, processing, transfer, repair, low-level
events, message delivery) and stochastic durations of the randomly occurring activities (arrival, departure, failure,
packet loss, message collision), including industrial, manufacturing, queueing, computing and network systems.

Moreover, since discrete phase type (DPH) distribution can be seen as a combination of multiple interrelated geometric
distributions, the DPH distributions are a flexible mathematical framework for analysis of complex random behaviours.
Remember that DPH distributions can approximate arbitrarily closely any non-negative discrete distributions, including
non-Markovian and finite-support ones. Therefore, DPH distributions are widely used in the distribution fitting (DF)
method that searches for optimal distribution (its type and parameters) that most closely approximates the data
produced by a given random process. The DF procedure measures the difference between the obtained theoretical (fitting)
and experimental (sample data) distributions, in order to identify the distribution type and fix the estimated
parameters for the maximal likehood \cite{BHST03,Ste09,HPV15,TBo17,LST19,HV23,Cio23,HHPTV24}. Thus, dtphPBC featuring
DPH distributed delays of activities would be appropriate to work with the non-Markovian stochastic systems from
various application areas like telecommunications, queueing theory, reliability engineering, insurance risk, inventory
management, population genetics, epidemiological models, molecular and quantum dynamics \cite{Cio23,HRGBF24}.

It is known that the attempts to combine time restrictions, parallelism and compositionality usually lead to many
technical difficulties, so that the formal models possessing all the mentioned properties have almost not been proposed
in the literature, in spite of the investigations in the related areas (for example, discrete time, generally
distributed delays, non-interleaving functional semantics in the SPA framework). To solve the mentioned problem, some
new (not existing in dtsdPBC) notions and constructions have been introduced in dtphPBC, such as phased multiactions,
delay phase indicators of timed multiactions, saturation with the phases, phase indicators discarding operation,
phases-respecting $Can$ and $Now$ functions, p-tangible and d-tangible dynamic expressions and states, action rules
respecting timed multiactions (with probabilistic and definite executions) and empty moves, transition systems with $3$
types of states (p-tangible, d-tangible and vanishing), and $4$ types of transitions (generated by probabilistic,
definite, immediate executions and empty~moves).

\subsection{Research perspectives}

We intend to construct the de\-no\-ta\-ti\-o\-nal se\-man\-tics of dtphPBC on a subclass of labeled discrete time
phased Petri nets (LDTPHPNs), an extension of DTSPNs \cite{Mol81,Mol85} with transition labeling and phased (phase type
delayed) transitions, called dtph-boxes. For evaluating performance in dtphPBC, the underlying stochastic process of
the process expressions will be extracted and analyzed, which will be a semi-Markov chain (SMC). The alternative
solution methods will be developed, based on the corresponding discrete time Markov chain (DTMC) and its reduction
(RDTMC) by eliminating vanishing (zero residence time) states.

We also plan to define behavioural equivalences for dtsdPBC, such as the step bisimulation one, aiming to reduce
behaviour of the algebraic processes by quotienting their transition systems and Markov chains. Such a reduction should
simplify the functional (qualitative) and performance (quantitative) analysis. We plan to construct some application
examples demonstrating expressiveness of the calculus and application of the behavioural analysis and performance
evaluation, both simplified using quotienting by step stochastic bisimulation. In future we can propose a congruence
relation for dtsdPBC, i.e. the equivalence that withstands application of all operations of the algebra. The first
possible candidate is a stronger version of step stochastic bisimulation equivalence, defined via transition systems
equipped with two extra transitions {\sf skip} and {\sf redo}, like those from sPBC \cite{MVCF08}.

Moreover, recursion operation could be added to dtsdPBC to increase further specification power of the algebra. It
would be very interesting to implement the class of DTSDPNs, to be able to specify them and then model their behaviour
by constructing the reahability graphs. Note that even DTSPNs of M.K. Molloy \cite{Mol81,Mol85} have never been
implemented. Mostly interleaving and continuous time variants of stochastic or timed PNs have been implemented so far.


\begin{thebibliography}{777}

\bibitem{AHR00} {\sc van der Aalst W.M.P., van Hee K.M., Reijers H.A.}
{\em Analysis of discrete-time stochastic Petri nets.}
{\sl Statistica Neerlandica} {\bf 54(2)}, p. 237--255, The Netherlands Society for Statistics and Operations Research
(VVSOR), Wiley Blackwell Publishers, July 2000,
{\tt http://tmitwww.tm.tue.nl/staff/hreijers/}\\
{\tt H.A. Reijers Bestanden/Statistica.pdf}, {\tt http://onlinelibrary.wiley.com/doi/epdf/10.1111/}\\
{\tt 1467-9574.00139}.

\bibitem{Bal01} {\sc Balbo G.}
{\em Introduction to stochastic Petri nets.}
{\sl Lecture Notes in Computer Science} {\bf 2090}, p. 84--155, 2001.

\bibitem{Bal07} {\sc Balbo G.}
{\em Introduction to generalized stochastic Petri nets.}
{\sl Lecture Notes in Computer Science} {\bf 4486}, p. 83--131, 2007.

\bibitem{BL16} {\sc Bartocci E., Li\'o P.}
{\em Computational modeling, formal analysis, and tools for systems biology.}
{\sl PLoS Com\-putational Biology} {\bf 12(1)}, p. 1--22 (e1004591), Public Library of Science, Cambridge, UK, January
2016.

\bibitem{BK85} {\sc Bergstra J.A., Klop J.W.}
{\em Algebra of communicating processes with abstraction.}
{\sl Theoretical Computer Science} {\bf 37}, p. 77--121, 1985.

\bibitem{BBr01} {\sc Bernardo M., Bravetti M.}
{\em Reward based congruences: can we aggregate more?}
{\sl Lecture Notes in Computer Science} {\bf 2165}, p. 136--151, 2001,
{\tt http://www.cs.unibo.it/\symbol{126}bravetti/papers/papm01b.ps}.

\bibitem{BGo98} {\sc Bernardo M., Gorrieri R.}
{\em A tutorial on EMPA: a theory of concurrent processes with nondeterminism, priorities, probabilities and time.}
{\sl Theoretical Computer Science} {\bf 202(1--2)}, p. 1--54, July 1998,\\
{\tt http://www.sti.uniurb.it/bernardo/documents/tcs202.pdf}.

\bibitem{BD24} {\sc Best E., Devillers R.}
{\em Petri net primer: a compendium on the core model, analysis, and synthesis.}
$1^{st}$ edition, {\sl Computer Science Foundations and Applied Logic (CSFAL) series}, 545 p., Springer International
Publishing, Cham, Switzerland / Birkh\"auser, Germany, January 2024 (ISBN 978-3-031-48277-9).

\bibitem{BDH92} {\sc Best E., Devillers R., Hall J.G.}
{\em The box calculus: a new causal algebra with multi-label communi\-cation.}
{\sl Lecture Notes in Computer Science} {\bf 609}, p. 21--69, 1992.

\bibitem{BDK01} {\sc Best E., Devillers R., Koutny M.}
{\em Petri net algebra.}
{\sl EATCS Monographs on Theoretical Computer Science}, 378 p., Springer Verlag, 2001 (ISBN 3-540-67398-9).

\bibitem{BKo95} {\sc Best E., Koutny M.}
{\em A refined view of the box algebra.}
{\sl Lecture Notes in Computer Science} {\bf 935}, p. 1--20, 1995,
{\tt http://parsys.informatik.uni-oldenburg.de/\symbol{126}best/publications/pn95.ps.gz}.

\bibitem{BHKS23} {\sc Bies A., Hermanns H., K\"ohl M.A., Schmidt A.}
{\em Matching distributions under structural constraints.}
{\sl Lecture Notes in Computer Science} {\bf 14287}, p. 221--237, 2023.

\bibitem{BHST03} {\sc Bobbio A., Horvath A., Scarpa M., Telek M.}
{\em Acyclic discrete phase type distributions: properties and a parameter estimation algorithm.}
{\sl Performance Evaluation} {\bf 54(1)}, p. 1--32, September 2003.

\bibitem{BLT90} {\sc Bolognesi T., Lucidi F., Trigila S.}
{\em From timed Petri nets to timed LOTOS.}
{\sl Proceedings of the IFIP WG6.1 $10^{th}$ International Symposium on Protocol Specification, Testing and
Verification - 90 (PSTV'90)} (L. Logrippo, R.L. Probert, H. Ural, eds.), Ottawa, Ontario, Canada, June 1990, p.
395--408, Elsevier Science Publishers (North-Holland), Amsterdam, The Netherlands, 1990.

\bibitem{BFFK09} {\sc Bonzanni N., Feenstra K.A., Fokkink W., Krepska E.}
{\em What can formal methods bring to systems biology?}
{\sl Lecture Notes in Computer Science} {\bf 5850}, p. 16--22, 2009,
{\tt https://www.cs.vu.nl/\symbol{126}wanf/pubs/}\\
{\tt position-fm4sb.pdf}.

\bibitem{Bri03} {\sc Brinksma E.}
{\em Compositional theories of qualitative and quantitative behaviour.}
{\sl Lecture Notes in Computer Science} {\bf 2679}, p. 37--42, 2003.

\bibitem{BrHe01} {\sc Brinksma E., Hermanns H.}
{\em Process algebra and Markov chains.}
{\sl Lecture Notes in Computer Science} {\bf 2090}, p. 183--231, 2001.

\bibitem{BKLL95} {\sc Brinksma E., Katoen J.-P., Langerak R., Latella D.}
{\em A stochastic causality-based process algebra.}
{\sl The Computer Journal} {\bf 38(7)}, p. 552--565, Oxford University Press, Oxford, UK, 1995,\\
{\tt http://eprints.eemcs.utwente.nl/6387/01/552.pdf}.

\bibitem{Buc95} {\sc Buchholz P.}
{\em A notion of equivalence for stochastic Petri nets.}
{\sl Lecture Notes in Computer Science} {\bf 935}, p. 161--180, 1995.

\bibitem{Buc98} {\sc Buchholz P.}
{\em Iterative decomposition and aggregation of labeled GSPNs.}
{\sl Lecture Notes in Computer Science} {\bf 1420}, p. 226--245, 1998.

\bibitem{Cia95} {\sc Ciardo G.}
{\em Discrete-time Markovian stochastic Petri nets.}
{\sl Computations with Markov Chains: Proceedings of $2^{nd}$ International Workshop on the Numerical Solution of
Markov Chains - 95 (NSMC'95)} (W.J. Stewart, ed.), Raleigh, NC, USA, January 1995, Chapter {\bf 20}, p. 339--358,
Kluwer Academic Publishers, Boston, MA, USA, February 1995,
{\tt http://www.cs.ucr.edu/\symbol{126}ciardo/pubs/1995NSMC-Discrete.pdf},\\
{\tt http://www.cs.iastate.edu/\symbol{126}ciardo/pubs/1995NSMC-Discrete.pdf}.

\bibitem{Cio23} {\sc Ciobanu G.}
{\em Analyzing non-Markovian systems by using a stochastic process calculus and a probabilistic model checker.}
{\sl Mathematics} {\bf 11}, Article 302, 17 p., Multidisciplinary Digital Publishing Institute (MDPI), Basel,
Swit\-zer\-land, January 2023.

\bibitem{CR14} {\sc Ciobanu G., Rotaru A.S.}
{\em PHASE: a stochastic formalism for phase-type distributions.}
{\sl Lecture Notes in Computer Science} {\bf 8829}, p. 91--106, 2014.

\bibitem{CR15} {\sc Ciobanu G., Rotaru A.S.}
{\em Phase-type approximations for non-Markovian systems: a case study.}
{\sl Lecture Notes in Computer Science} {\bf 8938}, p. 323--334, 2015.

\bibitem{Cum85} {\sc Cumani A.}
{\em ESP --- a package for the evaluation of stochastic Petri nets with phase-type distributed transition times.}
{\sl Proceedings of $1^{st}$ International Workshop on Timed Petri Nets - 85}, Turin, Italy, July 1985, p. 144--151,
IEEE Computer Society Press, July 1985.

\bibitem{DFFHK07} {\sc Danos V., Feret J., Fontana W., Harmer R., Krivine J.}
{\em Rule-based modelling of cellular signalling.}
{\sl Lecture Notes in Computer Science} {\bf 4703}, p. 17--41, 2007.

\bibitem{DFFK07} {\sc Danos V., Feret J., Fontana W., Krivine J.}
{\em Scalable simulation of cellular signaling networks.}
{\sl Lecture Notes in Computer Science} {\bf 4807}, p. 139--157, 2007.

\bibitem{DFFK08} {\sc Danos V., Feret J., Fontana W., Krivine J.}
{\em Abstract interpretation of cellular signalling networks.}
{\sl Lecture Notes in Computer Science} {\bf 4905}, p. 83--97, 2008.

\bibitem{ERKN99} {\sc El-Rayes A., Kwiatkowska M., Norman G.}
{\em Solving infinite stochastic process algebra models through matrix-geometric methods.}
{\sl Proceedings of $7^{th}$ International Workshop on Process Algebra and Performance Modelling - 99 (PAPM'99)} (J.
Hillston, M. Silva, eds.), Zaragoza, Spain, 1999, p. 41--62, Prensas Universitarias de Zaragoza, Spain, September 1999,
{\tt http://qav.comlab.ox.ac.uk/papers/papm99.pdf}.

\bibitem{GHL07} {\sc Gilbert D., Heiner M., Lehrack S.}
{\em A unifying framework for modelling and analysing biochemical pathways using Petri nets.}
{\sl Lecture Notes in Computer Science} {\bf 4695}, p. 200--216, 2007.

\bibitem{GSS95} {\sc van Glabbeek R.J., Smolka S.A., Steffen B.}
{\em Reactive, generative, and stratified models of pro\-ba\-bi\-lis\-tic processes.}
{\sl Information and Computation} {\bf 121(1)}, p. 59--80, August 1995,\\
{\tt http://boole.stanford.edu/pub/prob.ps.gz}.

\bibitem{Hani93} {\sc Hanish H.M.}
{\em Analysis of place/transition nets with timed-arcs and its application to batch process control.}
{\sl Lecture Notes in Computer Science} {\bf 691}, p. 282--299, 1993.

\bibitem{HK00} {\sc Hermanns H., Katoen J.-P.}
{\em Automated compositional Markov chain generation for a plain-old telephone system.}
{\sl Science of Computer Programming} {\bf 36(1)}, p. 97--127, January 2000.

\bibitem{HR94} {\sc Hermanns H., Rettelbach M.}
{\em Syntax, semantics, equivalences and axioms for MTIPP.}
{\sl Proceedings of $2^{nd}$ Workshop on Process Algebras and Performance Modelling},
Regensberg / Erlangen (Herzog U., Rettelbach M., eds.),
{\sl Arbeitsberichte des IMMD} {\bf 27}, p. 71--88, University of Erlangen, Germany, November 1994,
{\tt http://ftp.informatik.uni-erlangen.de/local/inf7/papers/Hermanns/syntax\symbol{95}semantics\symbol{95}}\\
{\tt equivalences\symbol{95}axioms\symbol{95}for\symbol{95}MTIPP.ps.gz}.

\bibitem{Hil94a} {\sc Hillston J.}
{\em A compositional approach to performance modelling.}
{\sl Ph.D. thesis}, 226 p., Department of Computer Science, University of Edinburgh, UK, April 1994.

\bibitem{Hil94b} {\sc Hillston J.}
{\em The nature of synchronisation.}
{\sl Proceedings of the $2^{nd}$ International Workshop on Process Algebra and Performance Modelling - 94 (PAPM'94)},
Regensberg / Erlangen (U. Herzog, M. Rettelbach, eds.),
{\sl Arbeitsberichte des IMMD} {\bf 27}, p. 51--70, University of Erlangen, Germany, November 1994,\\
{\tt http://www.dcs.ed.ac.uk/pepa/synchronisation.pdf}.

\bibitem{Hil96} {\sc Hillston J.}
{\em A compositional approach to performance modelling.}
158 p., Cambridge University Press, UK, 1996 (ISBN 978-0-521-67353-2),
{\tt http://www.dcs.ed.ac.uk/pepa/book.pdf}.

\bibitem{Hoa85} {\sc Hoare C.A.R.}
{\em Communicating sequential processes.}
Prentice-Hall, Englewood Cliffs, NJ, USA, 256 p., July 1985 (ISBN 978-0-13-153271-7),
{\tt http://www.usingcsp.com/cspbook.pdf}.

\bibitem{HRGBF24} {\sc Hobolth A., Rivas-Gonzalez I., Bladt M., Futschik A.}
{\em Phase-type distributions in mathematical population genetics: an emerging framework.}
{\sl Theoretical Population Biology} {\bf 157}, p. 14--32, Elsevier, June 2024.

\bibitem{HHPTV24} {\sc Horv\'ath A., Horv\'ath I., Paolieri M., Telek M., Vicario E.}
{\em Approximation of cumulative distribution functions by Bernstein phase-type distributions.}
{\sl Lecture Notes in Computer Science} {\bf 14996}, p. 90--106, 2024.

\bibitem{HPV15} {\sc Horv\'ath A., Paolieri M., Vicario E.}
{\em Approximating distributions and transient probabilities by matrix exponential distributions and functions.}
{\sl Quantitative Assessments of Distributed Systems: Methodologies and Techniques} (D. Bruneo, S. Distefano,
eds.), Chapter {\bf 5}, p. 107--127, Performability Engineering Series, Scrivener Publishing LLC, Beverly, MA, USA,
April 2015.

\bibitem{HPST00} {\sc Horv\'ath A., Puliafito A., Scarpa M., Telek M.}
{\em Analysis and evaluation of non-Markovian stochastic Petri nets.}
{\sl Lecture Notes in Computer Science} {\bf 1786}, p. 171--187, March 2000,\\
{\tt http://webspn.hit.bme.hu/\symbol{126}telek/cikkek/horv00b.ps.gz}.

\bibitem{HV23} {\sc Horv\'ath A., Vicario E.}
{\em Construction of phase type distributions by Bernstein exponentials.}
{\sl Lecture Notes in Computer Science} {\bf 14231}, p. 201--215, 2023.

\bibitem{JC02} {\sc Jones R.L., Ciardo G.}
{\em On phased delay stochastic Petri nets.}
{\sl Proceedings of $9^{th}$ International Workshop on Petri Nets and Performance Models - 01 (PNPM'01)}, Aachen,
Germany, September 2001, p. 165--174, IEEE Computer Society Press, 2001.

\bibitem{Kat96} {\sc Katoen J.-P.}
{\em Quantinative and qualitative extensions of event structures.}
{\sl Ph. D. thesis}, {\sl CTIT Ph. D.-thesis series} {\bf 96-09}, 303 p., Centre for Telematics and Information
Technology, University of Twente, Enschede, The Netherlands, 1996.

\bibitem{Kou00} {\sc Koutny M.}
{\em A compositional model of time Petri nets.}
{\sl Lecture Notes in Computer Science} {\bf 1825}, p. 303--322, 2000.

\bibitem{LST19} {\sc Lakatos L., Szeidl L., Telek M.}
{\em Introduction to queueing systems with telecommunication applications.}
$2^{nd}$ edition, 559 p., Springer Nature, Cham, Switzerland, 2019 (ISBN 978-3-030-15141-6).

\bibitem{MVCC04} {\sc Maci\`a S.H., Valero R.V., Cazorla L.D.C., Cuartero G.F.}
{\em Introducing the iteration in sPBC.}
{\sl Lecture Notes in Computer Science} {\bf 3235}, p. 292--308, October 2004,
{\tt http://www.info-ab.uclm.es/}\\
{\tt retics/publications/2004/forte04.pdf}.

\bibitem{MVCF08} {\sc Maci\`a S.H., Valero R.V., Cuartero G.F., de Frutos E.D.}
{\em A congruence relation for sPBC.}
{\sl Formal Methods in System Design} {\bf 32(2)}, p. 85--128, Springer, The Netherlands, April 2008.

\bibitem{MVCR08} {\sc Maci\`a S.H., Valero R.V., Cuartero G.F., Ruiz M.C.}
{\em sPBC: a Markovian extension of Petri box calculus with immediate multiactions.}
{\sl Fundamenta Informaticae} {\bf 87(3--4)}, p. 367--406, IOS Press, Amsterdam, The Netherlands, 2008.

\bibitem{MVCRT16} {\sc Maci\`a S.H., Valero R.V., Cuartero G.F., Ruiz M.C., Tarasyuk I.V.}
{\em Modelling a video conference system with sPBC.}
{\sl Applied Mathematics and Information Sciences} {\bf 10(2)}, p. 475--493, Natural Sciences Publishing, New York, NY,
USA, March 2016 (ISSN 1935-0090), DOI: 10.18576/amis/100210,\\
{\tt http://itar.iis.nsk.su/files/itar/pages/amis16.pdf}.

\bibitem{MVF01} {\sc Maci\`a S.H., Valero R.V., de Frutos E.D.}
{\em sPBC: a Markovian extension of finite Petri box calculus.}
{\sl Proceedings of $9^{th}$ IEEE International Workshop on Petri Nets and Performance Models - 01 (PNPM'01)}, Aachen,
Germany, 2001, p. 207--216, IEEE Computer Society Press, September 2001,\\
{\tt http://www.info-ab.uclm.es/retics/publications/2001/pnpm01.ps}.

\bibitem{MABV12} {\sc Markovski J., D'Argenio P.R., Baeten J.C.M., de Vink E.P.}
{\em Reconciling real and stochastic time: the need for probabilistic refinement.}
{\sl Formal Aspects of Computing} {\bf 24(4--6)}, p. 497--518, July 2012.

\bibitem{MVi08} {\sc Markovski J., de Vink E.P.}
{\em Extending timed process algebra with discrete stochastic time.}
{\sl Lecture Notes of Computer Science} {\bf 5140}, p. 268--283, 2008.

\bibitem{MVi09} {\sc Markovski J., de Vink E.P.}
{\em Performance evaluation of distributed systems based on a discrete real- and stochastic-time process algebra.}
{\sl Fundamenta Informaticae} {\bf 95(1)}, p. 157--186, IOS Press, Amsterdam, The Netherlands, October 2009.

\bibitem{MF00} {\sc Marroqu\'{\i}n A.O., de Frutos E.D.}
{\em TPBC: timed Petri box calculus.}
{\sl Technical Report}, Departamento de Sistemas Infofm\'aticos y Programaci\'on, Universidad Complutense de Madrid,
Madrid, Spain, 2000 (in Spanish).

\bibitem{MF01} {\sc Marroqu\'{\i}n A.O., de Frutos E.D.}
{\em Extending the Petri box calculus with time.}
{\sl Lecture Notes in Computer Science} {\bf 2075}, p. 303--322, 2001.

\bibitem{Mar90} {\sc Marsan M.A.}
{\em Stochastic Petri nets: an elementary introduction.}
{\sl Lecture Notes in Computer Science} {\bf 424}, p. 1--29, 1990.

\bibitem{MBCDF95} {\sc Marsan M.A., Balbo G., Conte G., Donatelli S., Franceschinis G.}
{\em Modelling with generalized stochastic Petri nets.}
{\sl Wiley Series in Parallel Computing}, John Wiley and Sons, 316 p., 1995 (ISBN 0-471-93059-8),
{\tt http://www.di.unito.it/\symbol{126}greatspn/GSPN-Wiley}.

\bibitem{MFa76} {\sc Merlin Ph.M., Farber D.J.}
{\em Recoverability of communication protocols: implications of a theoretical study.}
{\sl IEEE Transactions on Communications} {\bf 24(9)}, p. 1036--1043, IEEE Computer Society Press, September 1976.

\bibitem{Mil89} {\sc Milner R.A.J.}
{\em Communication and concurrency.}
260 p., Prentice-Hall, Upper Saddle River, NJ, USA, 1989 (ISBN 0-13-115007-3).

\bibitem{Mol81} {\sc Molloy M.K.}
{\em On the integration of the throughput and delay measures in distributed processing models.}
{\sl Ph.D. thesis}, {\sl Report} {\bf CSD-810-921}, 108 p., University of California, Los Angeles, CA, USA, 1981.

\bibitem{Mol82} {\sc Molloy M.K.}
{\em Performance analysis using stochastic Petri nets.}
{\sl IEEE Transactions on Computing} {\bf 31(9)}, p. 913--917, IEEE Computer Society Press, September 1982.

\bibitem{Mol85} {\sc Molloy M.K.}
{\em Discrete time stochastic Petri nets.}
{\sl IEEE Transactions on Software Engineering} {\bf 11(4)}, p. 417--423, IEEE Computer Society Press, April 1985.

\bibitem{Neu81} {\sc Neuts M.F.}
{\em Matrix-geometric solutions in stochastic models: an algorithmic approach.}
{\sl Johns Hopkins Se\-ri\-es in the Ma\-the\-ma\-ti\-cal Sciences} {\bf 2}, 352 p., Johns Hopkins University Press,
Baltimore, MD, USA, July 1981 (ISBN 978-0-801-82560-6).

\bibitem{Nia05} {\sc Niaouris A.}
{\em An algebra of Petri nets with arc-based time restrictions.}
{\sl Lecture Notes in Computer Science} {\bf 3407}, p. 447--462, 2005.

\bibitem{NK05} {\sc Niaouris A., Koutny M.}
{\em An algebra of timed-arc Petri nets.}
{\sl Technical Report} {\bf CS-TR-895}, 60 p., School of Computer Science, University of Newcastle upon Tyne, UK, March
2005,\\
{\tt http://www.cs.ncl.ac.uk/publications/trs/papers/895.pdf}.

\bibitem{Pul09} {\sc Pulungan M.R.}
{\em Reduction of acyclic phase-type representations.}
{\sl Ph.D. thesis}, 167 p., Fakult\"at f\"ur Mathematik und Informatik (MI), Universit\"at des Saarlandes,
Saarbr\"ucken, Germany, May 2009.

\bibitem{PH15} {\sc Pulungan M.R., Hermanns H.}
{\em A construction and minimization service for continuous probability distributions.}
{\sl International Journal on Software Tools for Technology Transfer} {\bf 17(1)}, p. 77--90, Springer, February 2015.

\bibitem{Ram73} {\sc Ramchandani C.}
{\em Performance evaluation of asynchronous concurrent systems by timed Petri nets.}
{\sl Ph.D. thesis}, Department of Electrical Engineering, Massachusetts Institute of Technology, Cambridge,
Massachusetts, USA, July 1973.

\bibitem{Ros96} {\sc Ross S.M.}
{\em Stochastic processes.}
$2^{nd}$ edition, John Wiley and Sons, 528 p., New York, USA, April 1996.

\bibitem{SS24} {\sc Soltanieh A., Siegle M.}
{\em Rate liting for stochastic process algebra by transition context augmentation.}
{\sl ACM Transactions on Modeling and Computer Simulation (TOMACS)}, 29 p., ACM Press, 2024, DOI: 10.1145/3656582.

\bibitem{Ste09} {\sc Stewart W.J.}
{\em Probability, Markov chains, queues, and simulation. The mathematical basis of performance modeling.}
758 p., Princeton University Press, Princeton, NJ, USA, 2009 (ISBN 978-0-691-14062-9).

\bibitem{Tar05} {\sc Tarasyuk I.V.}
{\em Discrete time stochastic Petri box calculus.}
{\sl Berichte aus dem Department f\"ur Informatik} {\bf 3/05}, 25 p., Carl von Ossietzky Universit\"at Oldenburg,
Germany, November 2005 (ISSN 1867-9218),\\
{\tt http://itar.iis.nsk.su/files/itar/pages/dtspbcib\symbol{95}cov.pdf}.

\bibitem{Tar06} {\sc Tarasyuk I.V.}
{\em Iteration in discrete time stochastic Petri box calculus.}
{\sl Bulletin of the Novosibirsk Computing Center, Series Computer Science, IIS Special Issue} {\bf 24}, p. 129--148,
NCC Publisher, Novosibirsk, 2006 (ISSN 1680-6972),
{\tt http://itar.iis.nsk.su/files/itar/pages/dtsitncc.pdf}.

\bibitem{Tar07a} {\sc Tarasyuk I.V.}
{\em Stochastic Petri box calculus with discrete time.}
{\sl Fundamenta Informaticae} {\bf 76(1--2)}, p. 189--218, IOS Press, Amsterdam, The Netherlands, February 2007
(ISSN 0169-2968),\\
{\tt http://itar.iis.nsk.su/files/itar/pages/dtspbcfi.pdf}.

\bibitem{Tar14} {\sc Tarasyuk I.V.}
{\em Equivalence relations for modular performance evaluation in dtsPBC.}
{\sl Mathematical Structures in Computer Science} {\bf 24(1)}, p. 78--154 (e240103), Cambridge University Press,
Cambridge, UK, February 2014 (ISSN 0960-1295), DOI: 10.1017/S0960129513000029,
{\tt http://itar.iis.nsk.su/files/}\\
{\tt itar/pages/dtsdphms.pdf}.

\bibitem{Tar19} {\sc Tarasyuk I.V.}
{\em Discrete time stochastic and deterministic Petri box calculus.}
{\sl CoRR} {\bf abs/1905.00456} ({\sl arXiv}:{\bf 1905.00456}), 57 p., Computing Research Repository, Cornell
University Library, Ithaca, NY, USA, May 2019,
{\tt http://itar.iis.nsk.su/files/itar/pages/dtsdpbcarxiv.pdf},\\
{\tt http://arxiv.org/pdf/1905.00456.pdf}.

\bibitem{Tar20a} {\sc Tarasyuk I.V.}
{\em Stochastic bisimulation and performance evaluation in discrete time stochastic and deter\-ministic Petri box
calculus dtsdPBC.}
{\sl HAL Open Archives} {\bf hal-02573419v1}, 113 p., France, May 2020,\\
{\tt http://itar.iis.nsk.su/files/itar/pages/dtsdpbchal.pdf},
{\tt https://hal.science/}\\
{\tt hal-02573419v1/file/dtsdpbchal.pdf}.

\bibitem{Tar20b} {\sc Tarasyuk I.V.}
{\em Discrete time stochastic and deterministic Petri box calculus dtsdPBC.}
{\sl Siberian Electronic Mathema\-tical Reports} {\bf 17}, p. 1598--1679, S.L. Sobolev Institute of Mathematics,
Novosibirsk, October 2020 (ISSN 1813-3304), DOI: 10.33048/semi.2020.17.112,
{\tt http://itar.iis.nsk.su/files/itar/pages/}\\
{\tt dtsdpbcodesemr.pdf}, {\tt http://semr.math.nsc.ru/v17/p1598-1679.pdf}.

\bibitem{Tar21} {\sc Tarasyuk I.V.}
{\em Performance evaluation in stochastic process algebra dtsdPBC.}
{\sl Siberian Electronic Mathe\-matical Reports} {\bf 18(2)}, p. 1105--1145, S.L. Sobolev Institute of Mathematics,
Novosibirsk, October 2021 (ISSN 1813-3304), DOI: 10.33048/semi.2021.18.085,
{\tt http://itar.iis.nsk.su/files/itar/pages/}\\
{\tt dtsdpbcpevsemr.pdf}, {\tt http://semr.math.nsc.ru/v18/n2/p1105-1145.pdf}.

\bibitem{Tar23} {\sc Tarasyuk I.V.}
{\em Performance preserving equivalence for stochastic process algebra dtsdPBC.}
{\sl Siberian Elec\-tro\-nic Ma\-the\-ma\-tical Reports} {\bf 20(2)}, p. 646--699, S.L. Sobolev Institute of
Mathematics, Novosibirsk, July 2023 (ISSN 1813-3304), DOI: 10.33048/semi.2023.20.039,
{\tt http://itar.iis.nsk.su/files/itar/pages/}\\
{\tt dtsdpbceqsemr.pdf}, {\tt http://semr.math.nsc.ru/v20/n2/p646-699.pdf}.

\bibitem{Tar24} {\sc Tarasyuk I.V.}
{\em Embedding and elimination for performance analysis in stochastic process algebra dtsdPBC.}
{\sl International Journal of Parallel, Emergent and Distributed Systems} {\bf 39(6)}, p. 619--652, Tailor and Francis
Gro\-up, In\-for\-ma PLC, London, UK, November 2024 (ISSN 1744-5779), DOI: 10.1080/\\
17445760.2024.2417873, {\tt http://itar.iis.nsk.su/files/itar/pages/dtsdpbcerepeds.pdf}.

\bibitem{Tar25a} {\sc Tarasyuk I.V.}
{\em Performance analysis of the shared memory system in stochastic process algebra dtsdPBC.}
{\sl International Journal of Parallel, Emergent and Distributed Systems} {\bf 40(4)}, p. 373--423, Tailor and Francis
Gro\-up, In\-for\-ma PLC, London, UK, July 2025 (ISSN 1744-5779), DOI: 10.1080/17445760.2025.2493128,\\
{\tt http://itar.iis.nsk.su/files/itar/pages/dtsdpbcshmpeds.pdf}.

\bibitem{Tar25b} {\sc Tarasyuk I.V.}
{\em Comparing dtsdPBC with other stochastic process algebras.}
{\sl International Journal of Parallel, Emergent and Distributed Systems} {\bf 40(5)}, p. 501--569, Tailor and Francis
Gro\-up, In\-for\-ma PLC, London, UK, September 2025 (ISSN 1744-5779), DOI: 10.1080/17445760.2025.2527713,\\
{\tt http://itar.iis.nsk.su/files/itar/pages/dtsdpbccmppeds.pdf}.

\bibitem{TMV10} {\sc Tarasyuk I.V., Maci\`a S.H., Valero R.V.}
{\em Discrete time stochastic Petri box calculus with immediate multiactions.}
{\sl Technical Report} {\bf DIAB-10-03-1}, 25 p., Department of Computer Systems, High School of Computer Science
Engineering, University of Castilla - La Mancha, Albacete, Spain, March 2010,\\
{\tt http://itar.iis.nsk.su/files/itar/pages/dtsipbc.pdf}, {\tt http://www.dsi.uclm.es/descargas/}\\
{\tt technicalreports/DIAB-10-03-1/dtsipbc.pdf}.

\bibitem{TMV13} {\sc Tarasyuk I.V., Maci\`a S.H., Valero R.V.}
{\em Discrete time stochastic Petri box calculus with immediate multiactions dtsiPBC.}
{\sl Proc. $6^{th}$ International Workshop on Practical Applications of Stochastic Modelling - 12 (PASM'12) and
$11^{th}$ International Workshop on Parallel and Distributed Methods in Verification - 12 (PDMC'12)},
{\sl Electronic Notes in Theoretical Computer Science} {\bf 296}, p. 229--252, Elsevier, 2013,
{\tt http://itar.iis.nsk.su/files/itar/pages/dtsipbcentcs.pdf}.

\bibitem{TMV14} {\sc Tarasyuk I.V., Maci\`a S.H., Valero R.V.}
{\em Performance analysis of concurrent systems in algebra dtsiPBC.}
{\sl Programming and Computer Software} {\bf 40(5)}, p. 229--249, Pleiades Publishing, Ltd., September 2014 (ISSN
0361-7688), DOI: 10.1134/S0361768814050089, {\tt http://itar.iis.nsk.su/files/itar/}\\
{\tt pages/pcs14.pdf}.

\bibitem{TMV15} {\sc Tarasyuk I.V., Maci\`a S.H., Valero R.V.}
{\em Stochastic process reduction for performance evaluation in dtsiPBC.}
{\sl Siberian Electronic Mathematical Reports} {\bf 12}, p. 513--551, S.L. Sobolev Institute of Mathe\-matics,
Novosibirsk, September 2015 (ISSN 1813-3304), DOI: 10.17377/semi.2015.12.044,\\
{\tt http://itar.iis.nsk.su/files/itar/pages/dtsipbcsemr.pdf}, {\tt http://semr.math.nsc.ru/v12/}\\
{\tt p513-551.pdf}.

\bibitem{TMV18} {\sc Tarasyuk I.V., Maci\`a S.H., Valero R.V.}
{\em Stochastic equivalence for performance analysis of concurrent systems in dtsiPBC.}
{\sl Siberian Electronic Mathematical Reports} {\bf 15}, p. 1743--1812, S.L. Sobolev Institute of Mathematics,
Novosibirsk, December 2018 (ISSN 1813-3304), DOI: 10.33048/semi.2018.15.144,\\
{\tt http://itar.iis.nsk.su/files/itar/pages/dtsipbceqsemr.pdf}, {\tt http://semr.math.nsc.ru/v15/}\\
{\tt p1743-1812.pdf}.

\bibitem{Tof00} {\sc Tofts C.}
{\em Symbolic approaches to probability distributions in process algebra.}
{\sl Formal Aspects of Computing} {\bf 12(5)}, p. 392--415, December 2000.

\bibitem{Tri16} {\sc Trivedi K.S.}
{\em Probability and statistics with reliability, queuing, and computer science applications.}
$2^{nd}$ edition, 857 p., John Wiley and Sons, Hoboken, NJ, USA, 2016 (ISBN 978-1-119-28542-7).

\bibitem{TBo17} {\sc Trivedi K.S., Bobbio A.}
{\em Reliability and availability engineering: modeling, analysis and applications.}
$1^{st}$ edition, 712 p., Cambridge University Press, Cambridge, UK, 2017 (ISBN 978-1-107-09950-0).

\bibitem{VC17} {\sc Valero R.V., Cambronero P.M.E.}
{\em Using unified modelling language to model the publish/subscribe paradigm in the context of timed Web services with
distributed resources.}
{\sl Mathematical and Computer Modelling of Dynamical Systems} {\bf 23(6)}, p. 570--594, Tailor and Francis, 2017.

\bibitem{Wolf08} {\sc Wolf V.}
{\em Equivalences on phase type processes.}
{\sl Ph.D. thesis}, 202 p., University of Mannheim, Mannheim, Germany, April 2008,
{\tt https://madoc.bib.uni-mannheim.de/1911/1/thesisA4.pdf}.

\bibitem{ZC96} {\sc Zijal R., Ciardo G.}
{\em Discrete deterministic and stochastic Petri nets.}
{\sl ICASE Report} {\bf 96-72}, i+23 p., Institute for Computer Applications in Science and Engineering (ICASE), NASA,
Langley Research Centre, Hampton, VA, USA, December 1996,
{\tt http://www.cs.odu.edu/\symbol{126}mln/ltrs-pdfs/icase-1996-72.pdf},\\
{\tt http://www.dtic.mil/dtic/tr/fulltext/u2/a322409.pdf}.

\bibitem{ZCH97} {\sc Zijal R., Ciardo G., Hommel G.}
{\em Discrete deterministic and stochastic Petri nets.}
{\sl Proceedings of $9^{th}$ ITG/GI Professional Meeting on Measuring, Modeling and Evaluation of Computer and
Communication Systems - 97 (MMB'97)} (K. Irmscher, Ch. Mittasch, K. Richter, eds.), Freiberg, Germany, September 1997,
Volume {\bf 1}, p. 103--117, VDE-Verlag, Berlin, Germany, 1997,
{\tt http://www.cs.ucr.edu/\symbol{126}ciardo/pubs/}\\
{\tt 1997MMB-DDSPN.pdf}.

\end{thebibliography}
\end{document}